%% file: main.tex
\documentclass[jair,twoside,11pt,theapa]{article}
\usepackage{jair, theapa, rawfonts}

\usepackage{graphicx}
 
\usepackage{caption}

\usepackage{tikz}
\usetikzlibrary{shapes, arrows, automata, positioning}
\usetikzlibrary{shapes.geometric}
\usetikzlibrary{arrows,decorations.pathmorphing,positioning,fit,trees,shapes,shadows,automata,calc} 
\usetikzlibrary{patterns,arrows,arrows.meta,calc,shapes,shadows,decorations.pathmorphing,decorations.pathreplacing,automata,shapes.multipart,positioning,shapes.geometric,fit,circuits,trees,shapes.gates.logic.US,fit, matrix,arrows.meta, quotes}
\usetikzlibrary{backgrounds,scopes}
\usepackage{listofitems} 

\usepackage{pgfplots}
\usepackage[linesnumbered,ruled,vlined]{algorithm2e}

\usepackage{subcaption}
\usepackage{mathtools}
\usepackage{bbm}
\usepackage{pbox}
\usepackage{xcolor,colortbl}

\usepackage[utf8]{inputenc} 
\usepackage{booktabs}       
\usepackage{amsfonts}       
\usepackage{amsmath}
\usepackage{nicefrac}       
\usepackage{amsthm}
\usepackage{bm}
\usepackage{url}

\newtheorem{theorem}{Theorem}
\newtheorem{lemma}{Lemma}

\pgfplotsset{
    layers/my layer set/.define layer set={
        background,
        main,
        foreground
    }{
    },
    set layers=my layer set,
    compat=1.18,
}

\usepgfplotslibrary{external}
\pgfplotsset{compat=newest}
\newenvironment{customlegend}[1][]{%
	\begingroup
	\csname pgfplots@init@cleared@structures\endcsname
	\pgfplotsset{#1}%
}{%
	\csname pgfplots@createlegend\endcsname
	\endgroup
}%

\def\addlegendimage{\csname pgfplots@addlegendimage\endcsname}

\jairheading{}{2022}{}{10/22}{}
\ShortHeadings{Verifiable Reinforcement Learning Systems via Compositionality}
{Neary, Samyal, Verginis, Cubuktepe, \& Topcu}
\firstpageno{1}

\makeatletter
\newcommand\HUGE{\@setfontsize\Huge{200}{240}}
\makeatother   

\begin{document}

\newtheorem{manualtheoreminner}{Theorem}
\newenvironment{manualtheorem}[1]{%
  \renewcommand\themanualtheoreminner{#1}%
  \manualtheoreminner
}{\endmanualtheoreminner}

\title{Verifiable Reinforcement Learning Systems \\ via Compositionality}

\author{\name Cyrus Neary \email cneary@utexas.edu \\
        \addr The University of Texas at Austin, USA
        \AND
        \name Aryaman Singh Samyal \email aryamansinghsamyal@utexas.edu \\
        \addr The University of Texas at Austin, USA
        \AND
        \name Christos Verginis \email christos.verginis@angstrom.uu.se \\
        \addr Uppsala University, Sweden 
        \AND
        \name Murat Cubuktepe \email mcubuktepe@utexas.edu \\
        \addr The University of Texas at Austin, USA
        \AND
        \name Ufuk Topcu \email utopcu@utexas.edu \\
        \addr The University of Texas at Austin, USA
      }

\maketitle

\input{macros}

\input{tex/00_abstract}
\input{tex/01_intro}
\input{tex/02_related_work}
\input{tex/03_problem_statement}
\input{tex/04_high_level_model}
\input{tex/05_icrl}
\input{tex/06_experiments}

\input{tex/09_conclusions}

\acks{This work was supported in part by ONR N00014-20-1-2115, ARO W911NF-20-1-0140, NSF 1652113, and NSF 2214939.}

\appendix
\input{appendices/thm1_proof}
\input{appendices/ppo_algorithm_parameters}

\bibliography{bibliography.bib}
\bibliographystyle{theapa}

\end{document}

%% file: macros.tex
\newcommand{\defeq}{\vcentcolon=}

\setlength\marginparwidth{74pt}
\newcommand{\colorpar}[3]{\colorbox{#1}{\parbox{#2}{#3}}}
\newcommand{\marginremark}[3]{\marginpar{\colorpar{#2}{\linewidth}{\color{#1}#3}}}
\newcommand{\cn}[1]{\marginremark{red}{white}{\scriptsize{[CN]~ #1}}}
\newcommand{\mc}[1]{\marginremark{red}{white}{\scriptsize{[MC]~ #1}}}

\newcommand{\bernoulliProb}{p}
\newcommand{\rewardRV}{R}
\newcommand{\rewardFunction}{R}
\newcommand{\rewardExpectedVal}{r}

\newcommand{\mdp}{M}
\newcommand{\mdpStateSet}{S}
\newcommand{\mdpState}{s}
\newcommand{\mdpActionSet}{A}
\newcommand{\mdpAction}{a}
\newcommand{\mdpRewardFunction}{R}
\newcommand{\mdpCommonReward}{r}
\newcommand{\mdpTransition}{P}
\newcommand{\mdpDiscount}{\gamma}
\newcommand{\mdpInitialState}{\mdpState_I}
\newcommand{\mdpInitialDist}{\alpha_I}
\newcommand{\mdpCostFunction}{C}

\newcommand{\pomdp}{\mdp}
\newcommand{\pomdpObservationSet}{Z}
\newcommand{\pomdpObservation}{z}
\newcommand{\pomdpObservationFunction}{\mathcal{O}}

\newcommand{\distribution}{\Delta}

\newcommand{\stateAbstraction}{\phi}

\newcommand{\policy}{\pi}
\newcommand{\targetStateSet}{\mdpStateSet_{targ}}
\newcommand{\initialStateSet}{\mdpStateSet_{init}}
\newcommand{\timeHorizon}{T}
\newcommand{\mdpSuccessState}{\mdpState_{g}}
\newcommand{\valueFunction}{V}
\newcommand{\terminationTime}{\tau}

\newcommand{\abstractMDP}{\tilde{\mdp}}
\newcommand{\abstractStateSet}{\tilde{\mdpStateSet}}
\newcommand{\abstractState}{\tilde{\mdpState}}
\newcommand{\abstractActionSet}{\tilde{\mdpActionSet}}
\newcommand{\abstractAction}{\tilde{\mdpAction}}
\newcommand{\abstractTransition}{\tilde{\mdpTransition}}
\newcommand{\abstractRewardFunction}{\tilde{\mdpRewardFunction}}
\newcommand{\abstractFailureState}{\abstractState_{\times}}
\newcommand{\abstractSuccessState}{\abstractState_{g}}
\newcommand{\abstractPolicy}{\mu}
\newcommand{\abstractInitialState}{\tilde{\mdpInitialState}}
\newcommand{\abstractInitialStateSet}{\abstractStateSet_{init}}
\newcommand{\abstractInitDist}{\alpha}
\newcommand{\abstractTargetStateSet}{\tilde{\targetStateSet}}

\newcommand{\boundMDP}{\bar{\mdp}}
\newcommand{\bernoulliProbBound}{\bar{\bernoulliProb}}
\newcommand{\boundMDPReward}{\bar{R}}
\newcommand{\boundMDPTransition}{\bar{\mdpTransition}}

\newcommand{\probThreshold}{P_{threshold}}
\newcommand{\failThreshold}{F_{threshold}}
\newcommand{\rewardThreshold}{R_{threshold}}

\newcommand{\hlmFailProb}{\delta}
\newcommand{\hlmPolicy}{\abstractPolicy}

\newcommand{\controller}{c}
\newcommand{\controllerSet}{\mathcal{C}}
\newcommand{\controllerInitialStateSet}{\mathcal{I}}
\newcommand{\controllerFinalStateSet}{\mathcal{F}}
\newcommand{\controllerTimeHorizon}{T}
\newcommand{\numControllers}{k}

\newcommand{\successProb}{\sigma}

\newcommand{\eqRelation}{R}

\newcommand{\occupancyVar}{x}

\newcommand{\controllerInfProb}{\Bar{\successProb}}

\newcommand{\numSteps}{N}
\newcommand{\initTrainingSteps}{N_{init}}
\newcommand{\estimationRollouts}{N_{est}}
\newcommand{\trainingSteps}{N_{train}}
\newcommand{\maxTrainingSteps}{N_{max}}
\newcommand{\controllerPerformanceEstimate}{\hat{\successProb}}

\newcommand{\lbList}{\mathcal{L}}
\newcommand{\ubList}{\mathcal{U}}

\newcommand{\perfAwareOptProblem}{\Omega}

 \newcommand\titlesize{\fontsize{8.1pt}{10.2pt}\selectfont}

\newcommand{\rewardDistanceConstant}{C}

\newcommand{\entropy}{H}
\newcommand{\trajectory}{\tau}
\newcommand{\dataset}{\mathcal{D}}
\newcommand{\numTrajectories}{N}
\newcommand{\observedOccupancyVar}{\bar{\occupancyVar}^{\mdpDiscount}}

\newcommand{\hlmTimeHorizon}{N}
\newcommand{\metaDecisionTime}{\tau}
\newcommand{\numMetaDecision}{m}
\newcommand{\numTimeStep}{n}
\newcommand{\reachTrajectories}{\Gamma}
\newcommand{\sigmaAlg}{\Sigma}
\newcommand{\measure}{\mathbb{P}}
\newcommand{\history}{h}

%% file: tex/00_abstract.tex
\begin{abstract}
    We propose a framework for verifiable and compositional reinforcement learning (RL) in which a collection of RL subsystems, each of which learns to accomplish a separate subtask, are composed to achieve an overall task.
    The framework consists of a \textit{high-level} model, represented as a parametric Markov decision process, which is used to plan and analyze compositions of subsystems, and of the collection of \textit{low-level} subsystems themselves.
    The subsystems are implemented as deep RL agents operating under partial observability.
    By defining interfaces between the subsystems, the framework enables automatic decompositions of task specifications, \textit{e.g., reach a target set of states with a probability of at least 0.95}, into individual subtask specifications, \textit{i.e. achieve the subsystem's exit conditions with at least some minimum probability, given that its entry conditions are met}.
    This in turn allows for the independent training and testing of the subsystems. 
    We present theoretical results guaranteeing that if each subsystem learns a policy satisfying its subtask specification, then their composition is guaranteed to satisfy the overall task specification.
    Conversely, if the subtask specifications cannot all be satisfied by the learned policies, we present a method, formulated as the problem of finding an optimal set of parameters in the high-level model, to automatically update the subtask specifications to account for the observed shortcomings.
    The result is an iterative procedure for defining subtask specifications, and for training the subsystems to meet them.
    Experimental results demonstrate the presented framework's novel capabilities in environments with both full and partial observability, discrete and continuous state and action spaces, as well as deterministic and stochastic dynamics.
\end{abstract}

%% file: tex/01_intro.tex
\section{Introduction}
\label{sec:intro}

Reinforcement learning (RL) algorithms offer tremendous capabilities in systems that work with unknown environments and limited observational capabilities.
However, there remain significant barriers to their deployment in safety-critical engineering applications.
Autonomous vehicles, manufacturing robotics, and power systems management are examples of complex application domains that require strict adherence of the system's behavior to stakeholder requirements.
However, the verification of RL systems is difficult.
This is particularly true of monolithic end-to-end RL approaches; many model-free RL algorithms, for instance, only output the learned policy and its estimated value function, rendering them opaque for verification purposes.
The difficulty of verification is compounded in engineering application domains, which often require large observation and action spaces, and complicated reward functions.

How do we build complex engineering systems we can trust?
Engineering design principles have long prescribed system modularity as a means to reduce the complexity of individual subsystems \shortcite{haberfellner2019systems,nuseibeh2000requirements}.
By creating well-defined interfaces between subsystems, system-level requirements may be decomposed into component-level ones. 
Conversely, each component may be developed and tested independently, and the satisfaction of component-level requirements may then be used to place assurances on the behavior of the system as a whole.
Building RL systems that incorporate such engineering practices and guarantees is a crucial step toward their widespread deployment.

Toward this end, we develop a framework for verifiable and compositional reinforcement learning.
The framework comprises two levels of abstraction.
The \textit{high level} is used to plan \textit{meta-policies} and to verify their adherence to task specifications, \textit{e.g., reach a particular goal state with a probability of at least 0.9}.
Meta-policies dictate sequences of \textit{subsystems} to execute, each of which is designed to accomplish a specific \textit{subtask}, \textit{i.e. achieve a particular exit condition, given the subsystem is executed from one of its entry conditions}.
We assume a collection of \textit{partially instantiated} subsystems to be given a priori; their entry and exit conditions are known, but the policies they implement are not.
These entry and exit conditions might be defined by pre-existing engineering capabilities, explicitly by a task designer, or by entities within the environment.
At the \textit{low level} of the framework, each subsystem employs RL algorithms to learn policies accomplishing its subtask. 
Figure \ref{fig:approach_illustration} illustrates the major components of the proposed framework.

\begin{figure*}[t]
    \centering
    \input{figures/approach_illustration}
    \caption{
    An illustration of the proposed framework.
    The task specification and the subtask entry and exit conditions are used to build the \textit{high-level model} (HLM) of the compositional RL system.
    We use the HLM to formulate an optimization problem whose outputs yield a \textit{meta-policy}, the probability of overall task success, and separate specifications for each subtask.
    The subtask specifications are used to select the next subsystem to train using the RL algorithm of choice.
    Estimates of the resulting subsystem policies are then used to update the HLM.
    This iterative process repeats until either the composite system satisfies the task specification, or a user-defined training budget has been exhausted.
    }
    \label{fig:approach_illustration}
\end{figure*}
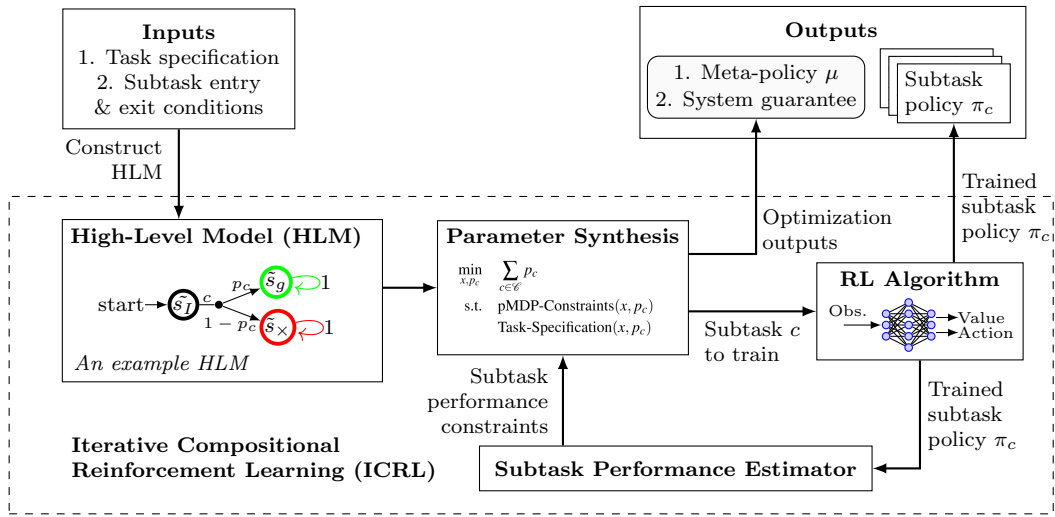 

We model the high level of the framework using a parametric Markov decision process (pMDP)~\shortcite{cubuktepe2018synthesis,junges2020parameter}.
Each action in the pMDP represents an individual RL subsystem, and the parametric transition probabilities in the pMDP thus represent the likelihoods of outcomes that could occur when the subsystem is executed.
Using sampling-based estimates of subsystem policies, we assign values to the model parameters and use existing MDP techniques for the planning and verification of meta-policies \cite{puterman2014markov,baier2008principles}.

We provide theoretical results that lower-bound the probabilities of outcomes in the task environment by the probabilities of outcomes predicted by the high-level pMDP model.
By building on this mathematical relationship,
we use the proposed framework to develop algorithms offering the following novel capabilities.

\begin{enumerate}
    \item \textbf{Automatic decomposition of task specifications.}
    We formulate, as the problem of finding an optimal set of parameters in the pMDP, a method to automatically decompose the task specification into subtask specifications, allowing for independent learning and verification of the subsystems.
    In words, the solution to this optimization problem allows us to automatically answer the following question: \textit{if we require our compositional system to satisfy a certain specification, how can we automatically generate the necessary subtask specifications for each of its subsystems}?
    This capability leads to new system-level decision-making problems, which we address in this paper.
    For example, it allows for the goal-oriented allocation of computational resources; subsystems should only be developed insofar as they are necessary to the satisfaction of the system-level specification.
    
    \item \textbf{Learning to satisfy subtask specifications.}
    Any RL method can be used to learn the subsystem policies, so long as the learned policies satisfy the relevant subtask specification.
    We present a subsystem reward function definition, in terms of the exit conditions of the subsystem, that motivates the learning of policies satisfying the subtask specification.
    Furthermore, these subtask specifications provide an \textit{interface} between the subsystems, allowing for the analysis of their compositions. In particular, we provide theoretical guarantees that if each of the learned subsystem policies satisfies its subtask specifications, a composition of them exists satisfying the specifications on the overall task.
    
    \item \textbf{Iterative specification refinement.}
    However, if some of the subtask specifications cannot be satisfied by the corresponding learned policies, sampling-based estimates of their behavior are used to update the high-level model.
    These updates effectively allow for capability-aware system-level planning; upper bounds on the capabilities of particular subsystems might be used to revise system-level plans, in order to better reflect what might realistically be achieved by the subsystems.
    This automatic refinement naturally leads to a compositional RL algorithm that iterates between computing subtask specifications, and training the corresponding subsystems to achieve them.
    
    \item \textbf{System modularity: prediction and verification in task transfer.}
    By providing a mathematical interface between the subtasks, the presented framework allows for previously learned subtask policies to be reused as components of new high-level models, designed to solve different tasks.
    Furthermore, the subtask specifications themselves may be reused to perform verification within these new models, without the need for further training.

\end{enumerate}

Results from a suite of numerical experiments exemplify these novel capabilities in variety of problem settings.
The experiments include task environments with both full and partial observability, discrete and continuous state and action spaces, as well as deterministic and stochastic dynamics. 
We use deep RL algorithms, including vision-to-action models that only receive egocentric image observations of the agent's immediate surroundings, to train individual subsystems to complete subtasks in the environments.
These subsystems are then composed to complete larger and more challenging tasks.
Using the proposed compositional RL algorithm, we automatically decompose probabilistic task specifications into subtask specifications. 
We observe that the system automatically plans to avoid challenging subtasks during its operation, while policies for unnecessary subtasks are not trained at all.
Finally, we demonstrate the benefits of the modular nature of the approach through an example in which previously trained subsystems are reused in new composite systems.
In one particular scenario, we observe that the proposed algorithm correctly asserts that its planned composition of RL subsystems will satisfy the probabilistic task specification in a new task environment.
It does so without taking a single step in the said environment for either training or evaluation.

Building on the results of the conference version \shortcite{neary2022verifiable}, the present paper additionally includes the following novel material:
\begin{itemize}
    \item We extend the proposed compositional RL framework to partially observable settings.
    This extension includes updates to the presentation of the problem formulation and its theoretical analysis, as well as a new experimental case study that trains compositional RL systems using only egocentric image observations.
    \item We provide an expanded collection of numerical experiments.
    In addition to the inclusion of the aforementioned image-based experiments, we also include a new experimental case study that demonstrates how the modular nature of the proposed compositional RL systems enables the transfer of previously trained subsystem policies to the solution of different tasks. 
    \item We provide a proof of the paper's main theoretical result, and additional discussion surrounding its interpretation and application.
\end{itemize}

\paragraph{Structure of the paper.}
In \S \ref{sec:problem_statement} we define the notions of tasks, specifications, and RL systems, which are necessary for the development of the proposed framework.
We subsequently introduce the \textit{high-level model} (HLM) in \S \ref{sec:hlm}. 
Here, we present the paper's main theoretical result and we present how the HLM can be used to plan policies and to automatically decompose the task specification into separate specifications for the individual subtasks.
In \S \ref{sec:icrl}, we discuss how to learn policies satisfying the subtask specifications, as well as how empirical rollouts of existing policies can be used to update the HLM, and refine the subtask specifications.
In the remainder of this section, we present a novel compositional RL algorithm that alternates between computing subtask specifications, and learning subsystem policies to satisfy them.
Finally, in \S \ref{sec:experiments}, we demonstrate the effectiveness of the proposed framework through a suite of numerical case studies.

%% file: figures/approach_illustration.tex
\begin{tikzpicture}[
    labelText/.style={font=\scriptsize, inner xsep=0.0cm, inner ysep=0.0cm},
    eqText/.style={font=\scriptsize, inner xsep=0.0cm, inner ysep=0.0cm},
]

\newlength{\layerVerticalSpacing}
\newlength{\layerHorizontalSpacing}
\newlength{\nnNodeRadius}

\setlength{\layerVerticalSpacing}{0.15cm}
\setlength{\layerHorizontalSpacing}{0.3cm}
\setlength{\nnNodeRadius}{0.1cm}

\tikzstyle{nnNode}=[
        draw=blue,
        fill=blue!20!white,
        circle,
        minimum size=\nnNodeRadius, 
        inner sep=0.0cm
    ]

\newsavebox{\NNBox}
\savebox{\NNBox}{
    \begin{tikzpicture}
        \readlist\Nnod{3,5,3} 
        \foreachitem \N \in \Nnod{ 
        \foreach \i [evaluate={\x=\Ncnt; \y=(\N/2-\i+0.5); \prev=int(\Ncnt-1);}] in {1,...,\N}{ 
            \node[nnNode] (N\Ncnt-\i) at (\x * \layerHorizontalSpacing, \y * \layerVerticalSpacing) {};
            \ifnum\Ncnt>1 
            \foreach \j in {1,...,\Nnod[\prev]}{ 
                \draw[] (N\prev-\j) -- (N\Ncnt-\i); 
            }
            \fi 
        }
        }
    \end{tikzpicture}
}

\node (inputsBox) [
    labelText,
    inner sep=2mm, 
    align=left, 
    draw=black!100!white,
    xshift=0cm, 
    yshift=0,
    align=center,
    ] {\textbf{Inputs} \\       
        1. Task specification\\
        2. Subtask entry \\ \& exit conditions
    };


\node (hlmTitle) [
    labelText,
    below=20mm of inputsBox.south,
    xshift=5mm,
    yshift=7mm,
] {\textbf{High-Level Model (HLM)}};


\node (hlmExample) [
    below=2.0mm of hlmTitle,
    inner sep=0.1mm,
    xshift=0mm,
    ] {\input{figures/simple_hlm2}};

\node (title) [
    labelText,
    below=14mm of hlmTitle, 
    inner sep=0mm, 
    text width=40mm, 
    xshift=1mm, 
    yshift=0mm,
    align=left,
] {\textit{An example HLM}};

\begin{scope}[on background layer]
    \node (hlmEnclosure) [
        draw=black!100!white,
        inner xsep=1mm,
        inner ysep=1mm,
        fit={(hlmTitle) (hlmExample) (title)}] {};
\end{scope}


\node (optTitle) [
    labelText,
    inner xsep=1mm,
    right=10mm of hlmTitle.east,
    xshift=0.0mm,
    yshift=0.0mm,
    align=left,
] {\textbf{Parameter Synthesis}};

\node (optExample) [
    below=0.0mm of optTitle,
    inner sep=0.1mm,
    xshift=-0.5mm,
    ] {\includegraphics[width=30mm]{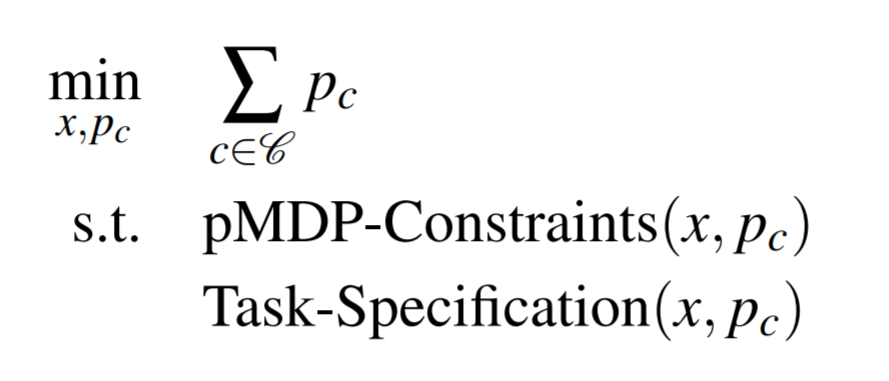}};
    
\begin{scope}[on background layer]
    \node (optEnclosure) [draw=black!100,
                            inner xsep=0.1mm,
                            inner ysep=1mm,
                            fit={(optTitle) (optExample)}] {};
\end{scope}

\node (performanceEstimator) [
    labelText,
    inner sep=2mm,
    draw=black!100!white,
    below=12mm of optEnclosure.south,
    xshift=15mm,
    yshift=0mm,
    align=center,
] {\textbf{Subtask Performance Estimator}};



\node (rlAlg) [
    labelText,
    right=20mm of optEnclosure.east,
    xshift=0mm,
    yshift=1mm,
    align=center,
] {\textbf{RL Algorithm}};

\node (rlNetwork) [
    labelText,
    below=0.1cm of rlAlg,
    xshift=-1.5mm,
] {\usebox{\NNBox}};

\node (rlAlgIn1) [left=0.0mm of rlNetwork.west, xshift=-0.5mm, yshift=2.0mm] {\tiny Obs.};
\node (rlAlgOut1) [right=0.0mm of rlNetwork.east, xshift=1.5mm, yshift=1.0mm] {\tiny Value};
\node (rlAlgOut2) [right=0.0mm of rlNetwork.east, xshift=1.5mm, yshift=-1.0mm] {\tiny Action};

\draw[-latex] ($(rlNetwork.west) + (-5mm, 0.0mm)$) -- (rlNetwork.west);
\draw[-latex] ($(rlNetwork.east) + (0mm, 1.0mm)$) -- ($(rlNetwork.east) + (3mm, 1.0mm)$);
\draw[-latex] ($(rlNetwork.east) + (0mm, -1.0mm)$) -- ($(rlNetwork.east) + (3mm, -1.0mm)$);

\node (rlAlgEnclosure) [
    draw=black!100!white,
    inner xsep=3mm,
    inner ysep=1mm,
    fit={(rlAlg) (rlNetwork)},
] {};

\node (icrlTitle) [
    labelText,
    inner ysep=0mm,
    below=25.0mm of hlmTitle,
    xshift=5.0mm,
    yshift=0.0mm,
    align=left,
] {\textbf{Iterative Compositional} \\ \textbf{Reinforcement Learning (ICRL)}};

\begin{scope}[on background layer]
    \node (icrlEnclosure) [draw=black,
        dashed,
        inner xsep=7mm,
        inner ysep=3mm,
        fit={(icrlTitle) (hlmEnclosure) (performanceEstimator) (rlAlg) (optEnclosure)}] {};
\end{scope}

\node (outputTitle) [
    labelText,
    inner sep=0mm,
    right=65mm of inputsBox.east,
    xshift=0.0mm,
    yshift=5.0mm,
    align=center,
] {\centering \textbf{Outputs}};

\node (metaPolicy) [
    labelText,
    inner sep=0mm,
    below=3.0mm of outputTitle,
    xshift=-10.0mm,
    yshift=0.0mm,
    align=center,
] {\centering 1. Meta-policy \(\abstractPolicy\)};

\node (sysGuarantee) [
    labelText,
    inner sep=0mm,
    below=1.0mm of metaPolicy,
    xshift=0.0mm,
    yshift=0.0mm,
    align=center,
] {\centering 2. System guarantee};

\begin{scope}[on background layer]
    \node (outSubEnclosure) [draw=black,
                    inner xsep=1mm,
                    inner ysep=1mm,
                    rounded corners=0.15cm,
                    fill=black!2!white,
                    fit={(metaPolicy) (sysGuarantee)}] {};
\end{scope}

\node (subPolicy1) [
    labelText,
    inner sep=1mm,
    draw=black!100!white,
    fill=white,
    right=3.0mm of sysGuarantee.east,
    xshift=0.0mm,
    yshift=3.0mm,
    text width=13mm,
] {\centering Subtask policy \(\policy_{\controller}\)};

\node (subPolicy2) [
    labelText,
    inner sep=1mm,
    draw=black!100!white,
    fill=white,
    right=0.0mm of subPolicy1,
    xshift=-14.0mm,
    yshift=-1.0mm,
    text width=13mm,
] {\centering Subtask policy \(\policy_{\controller}\)};

\node (subPolicy3) [
    labelText,
    inner sep=1mm,
    draw=black!100!white,
    fill=white,
    right=0.0mm of subPolicy2,
    xshift=-14.0mm,
    yshift=-1.0mm,
    text width=13mm,
] {\centering Subtask policy \(\policy_{\controller}\)};

\begin{scope}[on background layer]
    \node (outputEnclosure) [draw=black!100!white,
        inner xsep=2mm,
        inner ysep=2mm,
        fit={(outputTitle) (metaPolicy) (sysGuarantee) (subPolicy1) (subPolicy2) (subPolicy3)}] {};
\end{scope}


\draw[
    thick,
    -latex,
] (inputsBox.south) -- ++(0.0, -12mm) node[labelText, left, align=right, yshift=8mm, xshift=-2mm] {Construct \\ HLM};

\draw[
    thick,
    -latex,
] ($(optEnclosure.west) + (-7.25mm, 0.0cm)$) -- (optEnclosure.west);

\draw[
    thick,
    -latex,
] ($(rlAlgEnclosure.west) + (-17mm, 0cm)$) -- node [labelText, align=left, yshift=-4mm] {Subtask \(\controller\) \\ to train} (rlAlgEnclosure.west);

\draw[
    thick,
    -latex,
] (rlAlgEnclosure.south) |- node [labelText, align=left, xshift=7mm, yshift=7mm] {Trained \\ subtask \\ policy \(\policy_{\controller}\)} (performanceEstimator.east);

\draw[
    thick,
    -latex,
] ($(optEnclosure.south) + (0.0mm, -12.0mm)$) -- node [labelText, align=right, xshift=-10mm, yshift=0mm] {Subtask \\ performance \\ constraints} (optEnclosure.south);


\draw[
    thick,
    -latex,
] ($(subPolicy3) + (0mm, -22.5mm)$) -- node [labelText, align=left, xshift=7mm, yshift=-2mm] {Trained \\ subtask \\ policy \(\policy_{\controller}\)} (subPolicy3.south);

\draw[
    thick,
    -latex,
] ($(outSubEnclosure.south) + (-9mm, -18mm)$) -| node [labelText, align=left, xshift=9.5mm, yshift=3.0mm] {Optimization \\ outputs} ($(outSubEnclosure.south) + (0mm, 0mm)$);










\end{tikzpicture}

%% file: figures/simple_hlm.tex




\begin{tikzpicture}[scale=1.0]

    \node[circle, draw, initial, line width=0.5mm, inner sep=0.0cm] (s_init) {\(\abstractInitialState\)};

    \node[state, right=12mm of s_init.east, draw=red, line width=1.0mm, xshift=0.0mm, yshift=-6mm] (s_fail) {\(\abstractFailureState\)};
    \node[state, right=12mm of s_init.east, draw=green, line width=1.0mm, xshift=0.0mm, yshift=6mm] (s_goal) {\(\abstractSuccessState\)};

    \node[branch, right=3mm of s_init.east] (c) {};
    \path [draw, line width=\edgeThickness, -] (s_init.east) -- node [above, yshift=1mm] {\(\controller\)} (c);
    \path [draw, ->, line width=\edgeThickness] (c) -- node [above, xshift=-1.5mm, yshift=0.5mm] {\(\bernoulliProb_{\controller}\)} (s_goal);
    \path [draw, ->, line width=\edgeThickness] (c) -- node [below, xshift=-2.0mm, yshift=-0.5mm] {\(1-\bernoulliProb_{\controller}\)} (s_fail);
    
    \path (s_fail) edge [loop right, draw=red, line width=\edgeThickness] node {1} (s_fail);
    \path (s_goal) edge [loop right, draw=green, line width=\edgeThickness] node {1} (s_goal);
    
\end{tikzpicture}

%% file: figures/simple_hlm2.tex
\begin{tikzpicture}[
    labelText/.style={font=\scriptsize, inner xsep=0.0cm, inner ysep=0.0cm},
    eqText/.style={font=\scriptsize, inner xsep=0.0cm, inner ysep=0.0cm},
]
    \tikzstyle{branch}=[fill, shape=circle, minimum size=3pt, inner sep=0pt]
                        
    \node (s_init) [eqText, circle, draw, line width=0.5mm, inner sep=0.0cm]  {\(\abstractInitialState\)};
    \node (start) [labelText, left=3mm of s_init.west] {start};
    
    \draw [-latex] (start.east) -- (s_init.west);

    \node(s_fail) [
        eqText,
        circle, 
        right=8mm of s_init.east, 
        inner sep=0.0cm, 
        draw=red, 
        line width=0.5mm, 
        xshift=0.0mm, 
        yshift=-3mm,
    ]  {\(\abstractFailureState\)};
    
    \node (s_goal) [
        eqText,
        circle,
        right=8mm of s_init.east,
        draw=green,
        line width=0.5mm,
        xshift=0.0mm,
        yshift=3mm,
        inner sep=0.0mm,
    ] {\(\abstractSuccessState\)};

    \node[branch, right=2mm of s_init.east] (c) {};
    \path [draw, -] (s_init.east) -- node [font=\tiny, yshift=1.5mm, xshift=0.0mm] {\(\controller\)} (c);
    \path [draw, -latex] (c) -- node [font=\tiny, above, xshift=0.0mm, yshift=0.5mm] {\(\bernoulliProb_{\controller}\)} (s_goal);
    \path [draw, -latex] (c) -- node [font=\tiny, below, inner sep=0.0mm, xshift=-1.5mm, yshift=-0.5mm] {\(1-\bernoulliProb_{\controller}\)} (s_fail);
    
    \path (s_fail) edge [loop right, draw=red] node [eqText] {\(1\)} (s_fail);
    \path (s_goal) edge [loop right, draw=green] node [eqText] {\(1\)} (s_goal);
    \end{tikzpicture}

%% file: tex/02_related_work.tex
\section{Related Work}

We note that the proposed framework is closely related to hierarchical RL (HRL) \shortcite{sutton1999between,barto2003recent,kulkarni2016hierarchical,vezhnevets2017feudal,nachum2019data,levy2017learning}.
However, our framework differs from these methods in that it explicitly builds a model of the upper levels of the decision-making hierarchy, and it uses the said model not only for planning, but also for the probabilistic verification of compositions of RL subsystems.
As a result, it adds a number of benefits to existing HRL methods.
These benefits include: a systematic means to decompose and refine task specifications, explicit reasoning over the probabilities of events, the use of planning-based solution techniques (which could incorporate additional problem constraints), and flexibility in the choice of RL algorithm used to learn subsystem policies.
HRL methods use task decompositions to reduce computational complexity, particularly in problems with large state and action spaces \shortcite{pateria2021hierarchical}.
However, they typically focus on the efficient maximization of discounted reward and they require the meta-policy to be learned; models of the high-level problem are typically not explicitly constructed.
We emphasize that the framework we present builds a model of the high-level problem with the specific aim of enabling verifiable RL against a rich set of task specifications (\textit{e.g., safely reach a target set with a required probability of success}), while enjoying a similar reduction in sample complexity to other HRL methods.

Compositional verification has been studied in formal methods \shortcite{namautomatic2008,kwiatkowska2010assume,feng2011automated}, but not in the context of RL.
Conversely, recent works have used structured task knowledge to decompose RL problems, however, they do not study how such information can be used to verify RL systems against probabilistic task specifications. 
\shortciteA{camacho2017non} and \shortciteA{littman2017environment} both define a task specification language based on linear temporal logic, and subsequently use it to generate reward functions for RL. 
\shortciteA{sarathy2020spotter} incorporates RL with symbolic planning models to learn new operators---similar to our subtasks---to aid in the completion of planning objectives. 
These works all use structured task knowledge to decompose RL problems, however, they do not provide methods for the automated verification and decomposition of task success probabilities, or for the targeted training of subsystems.

Meanwhile, \shortciteA{icarte2022reward,icarte2018using} use reward machines, finite-state machines encoding temporally extended tasks in terms of atomic propositions, to break tasks into stages for which separate policies can be learned.  
\shortciteA{neary2020reward} extends the use of reward machines to the multi-agent RL setting, decomposing team tasks into subtasks for individual learners.
\shortciteA{toro2019learning} and \shortciteA{xu2020joint} propose to simultaneously learn reward machines and RL policies. 
Similarly, \shortciteA{furelos2021induction} propose to learn automata-based representations of RL tasks.
However, we note that while reward machines and similar kinds of automata provide structured representations of tasks, they do not necessarily model the outcomes of subsystems in a modular and reusable fashion. 
A reward machine's states are related to progress through the particular task that it represents. 
By contrast, the states of the high-level model that we propose in this work represent task-agnostic collections of environment states.
Our high-level model thus represents its subsystems as entities that are coupled, yet modular.
Furthermore, we estimate the probabilities that each of the subsystem policies succeeds at its subtask.
These differences from the reward machine formulation, along with our theoretical results that relate the high-level model's predictions to the outcomes of compositions of RL subsystems, are what enables our approach to verification against probabilistic task specifications, automated decomposition and iterative refinement of subtask specifications, and modular reuse of subsystems in new tasks.
To the best of our knowledge, these are novel capabilities for RL-based systems.

%% file: tex/03_problem_statement.tex
\section{The Compositional Reinforcement Learning Framework}
\label{sec:problem_statement}

To provide intuitive examples of the notions of tasks, subtasks, systems, and subsystems, we consider the example labyrinth environment shown in Figure \ref{fig:running_example}. 
The \textit{system} executes its constituent \textit{subsystems} in this environment to complete an overall task.
The \textit{task} is to safely navigate from the labyrinth's initial state in the top left corner to the goal state in the bottom left corner. Satisfaction of the \textit{task specification} requires that the system successfully completes the task with a probability of at least \(0.95\).
As an added difficulty, lava exists within some of the rooms, represented in the figure by the orange rectangles. 
If the lava is touched, the task is automatically failed. 
This task is naturally decomposed into separate \textit{subtasks}, each of which navigates an individual room, and is executed by a separate subsystem.

\subsection{Preliminaries}
We model the task environment as a partially observable Markov decision process (POMDP), which is defined by a tuple \(\pomdp\) \(= (\mdpStateSet,\) \(\mdpActionSet,\) \(\mdpTransition,\) \(\pomdpObservationSet,\) \(\pomdpObservationFunction)\). 
Here, \(\mdpStateSet\) is a set of states, \(\mdpActionSet\) is a set of actions, \(\mdpTransition : \mdpStateSet \times \mdpActionSet \times \mdpStateSet \to [0,1]\) is a transition probability function, \(\pomdpObservationSet\) is a set of possible observations, and \(\pomdpObservationFunction : \mdpStateSet \times \pomdpObservationSet \to [0,1]\) is an observation probability function. 
We note that \(\mdpStateSet\), \(\mdpActionSet\), and \(\pomdpObservationSet\) can be uncountable sets---representing environments with continuous state and control actions---or they can be countable sets indexing finite collections of states and actions. 
Our experiments examine both cases.
However, for notational simplicity, we present the framework for countable instances of these sets.

A policy within the POMDP is a function \(\policy : (\pomdpObservationSet \times \mdpActionSet)^* \times \pomdpObservationSet \times \mdpActionSet \to [0,1]\) that maps histories of observations and actions \(\pomdpObservation_{0} \mdpAction_{0} \ldots \pomdpObservation_{t} \in (\pomdpObservationSet \times \mdpActionSet)^* \times \pomdpObservationSet \) to distributions over actions \(\mdpAction \in \mdpActionSet\). 
In this work, we use reinforcement learning (RL) algorithms to learn policies \(\policy\) that accomplish \textit{subtasks} specified in \(\mdp\).
However, we note that the compositional framework we present is agnostic to the implementation details of the policies and the RL algorithms used to train them, so long as the policies are performant with respect to their objectives in \(\mdp\), which we define in \S \ref{sec:subsystems_and_subtasks}.

Given a POMDP \(\pomdp\), a policy \(\policy\), and a target set of states \(\targetStateSet \subseteq \mdpStateSet\), we define \(\mathbb{P}^{\mdpState}_{\mdp}(\Diamond \targetStateSet | \policy)\) to be the probability of eventually reaching some state \(\mdpState' \in \targetStateSet\), beginning from the initial state \(\mdpState\), under policy \(\policy\).
Similarly, \(\mathbb{P}^{\mdpState}_{\mdp}(\Diamond_{\leq \timeHorizon} \targetStateSet| \policy)\) denotes the probability of reaching the target set from state \(\mdpState\) within some finite time horizon \(\timeHorizon\).

\begin{figure*}[t!]
    \centering
    \begin{subfigure}[t]{0.3\textwidth}
        \centering \input{figures/labyrinth_both_overlay}
        \label{fig:labyrinth_gridworld}
    \end{subfigure}%
    ~
    \hspace*{18mm}
    \begin{subfigure}[t]{0.55\textwidth}
        \centering \input{figures/hlm}
        \label{fig:labyrinth_HLM}
    \end{subfigure}
    \caption{
    An example labyrinth navigation task. Left: An illustration of the environment, as well as an example collection of subsystems, represented by the colored paths. 
    Entry and exit conditions for the various subsystems are shown as blue circles. 
    Right: An illustration of the corresponding HLM. 
    Each subsystem \(\controller\) causes a transition to its successor state with probability \(\bernoulliProb_{\controller}\). 
    Otherwise, the HLM transitions to the failure state \(\abstractFailureState\) with probability \(1 - \bernoulliProb_{\controller}\), visualized by the red transitions.
    } 
    \label{fig:running_example}
\end{figure*}
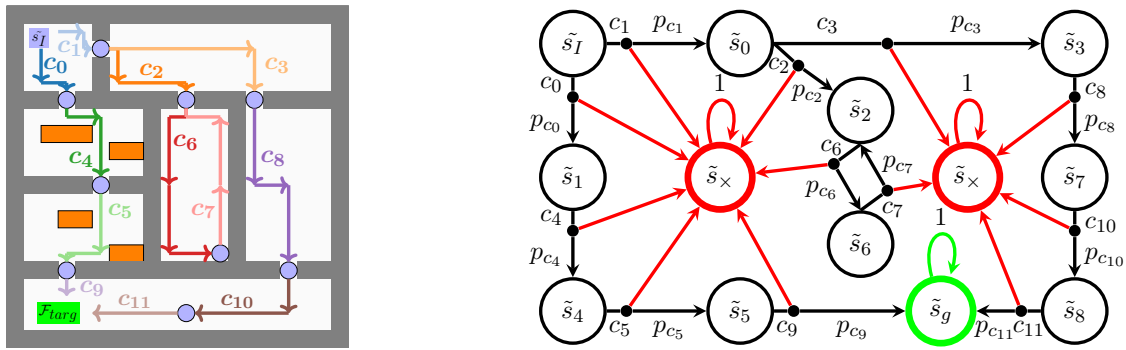

\subsection{Reinforcement Learning Subsystems and Subtasks}
\label{sec:subsystems_and_subtasks}
We define RL subsystems \(\controller\) acting within the environment using the \textit{options} framework \cite{sutton1999between}. 
In particular, we define each RL subsystem by the tuple \(\controller = (\controllerInitialStateSet_{\controller}, \controllerFinalStateSet_{\controller}, \controllerTimeHorizon_{\controller}, \policy_{\controller})\).
Here, \(\controllerInitialStateSet_{\controller} \subseteq \mdpStateSet\) is a set defining the subsystem's \textit{entry conditions}, \(\controllerFinalStateSet_{\controller} \subseteq \mdpStateSet\) is a set representing the subsystem's \textit{exit conditions}, and \(\timeHorizon_{\controller} \in \mathbb{N}\) is the subsystem's allowed \textit{time horizon}.
The \textit{subtask} associated with each subsystem, is to navigate from any entry condition \(\mdpState \in \controllerInitialStateSet_{\controller}\) to any exit condition \(\mdpState' \in \controllerFinalStateSet_{\controller}\) within the subsystem's time horizon \(\timeHorizon_{\controller}\). 
The time horizon is included to ensure that the compositional system will complete its task in a finite time.
We assume that each subsystem may only be \textit{executed}, or begun, from an entry condition \(\mdpState \in \controllerInitialStateSet_{\controller}\) and that its execution ends either when it achieves an exit condition \(\mdpState \in \controllerFinalStateSet_{\controller}\), or when it runs out of time.
Finally, \(\policy_{\controller}\) is the policy that the component implements to complete this objective.

We note that when implementing such subsystems in practice, a mechanism is required to ensure that the subsystem's execution terminates once it reaches an exit condition \(\mdpState' \in \controllerFinalStateSet_{\controller}\).
In the fully observable setting---when the agent may directly observe its current state at every time step---this may be taken for granted.
In settings with partial observability, the agent should either be able use its observations to determine when it has reached an exit condition, or some external mechanism should be designed that triggers when the subsystem completes its subtask.
In this work, we assume that such a mechanism exists and that the subsystem terminates its execution whenever it first reaches some \(\mdpState' \in \controllerFinalStateSet_{\controller}\).

For notational convenience, we define \(\successProb_{\policy_{\controller}}^{\controller}(\mdpState) \defeq \mathbb{P}^{\mdpState}_{\mdp}(\Diamond_{\leq \timeHorizon_{\controller}}\controllerFinalStateSet_{\controller} | \policy_{\controller})\).
A \textit{subtask specification}, is then defined as the requirement that \(\successProb_{\policy_{\controller}}^{\controller}(s) \geq \bernoulliProb_{\controller}\) for every entry condition \(\mdpState \in \controllerInitialStateSet_{\controller}\) of the subsystem. Here, \(\bernoulliProb_{\controller} \in [0,1]\) is a value representing the minimum allowable probability of the subtask success. 
We note that such reachability-based task specifications are very expressive. 
As an example, temporal logic specifications can be expressed as reachability specifications in a so-called product MDP \cite{baier2008principles,hahn2019omega}.

We say a subsystem \(\controller\) is \textit{partially instantiated} when \(\controllerInitialStateSet_{\controller}\), \(\controllerFinalStateSet_{\controller}\), and \(\timeHorizon_{\controller}\) are defined, but its policy \(\policy_{\controller}\) is not.
We define a collection \(\controllerSet = \{\controller_1, \controller_2, ..., \controller_{\numControllers}\}\) of subsystems to be \textit{composable}, if and only if for every \(i,j \in \{1,2,\ldots, \numControllers\}\), either \(\controllerFinalStateSet_{\controller_i} \subseteq \controllerInitialStateSet_{\controller_j}\) or \(\controllerFinalStateSet_{\controller_i} \cap \controllerInitialStateSet_{\controller_j} = \emptyset\).
In words, subsystems are composable when the set of exit conditions of each subsystem is a subset of all the sets of entry conditions that it intersects.
This requirement ensures that regardless of the specific exit condition \(\mdpState \in \controllerFinalStateSet_{\controller}\) in which subsystem \(\controller\) terminates, \(\mdpState\) will be a valid entry condition for the \textit{same} collection of other subsystems.

\subsection{Compositions of Reinforcement Learning Subsystems}
\label{sec:problem_statement_meta_policies}
Compositions of subsystems are specified by \textit{meta-policies} \(\abstractPolicy\), which are formally defined in terms of the high-level models that we introduce in \S \ref{sec:hlm}. 
Intuitively, after reaching the exit condition of any given subsystem, the meta-policy selects which subsystem to execute next.
So, execution of the composite system occurs as follows.
From an initial state \(\mdpState\), the meta-policy is used to select a subsystem \(\controller\) to execute, from the finite collection of \(\numControllers\) subsystems \(\controllerSet = \{\controller_{1}, \ldots, \controller_{\numControllers}\}\). 
The subsystem's policy \(\policy_{\controller}\) is then followed until it either reaches an exit condition \(\mdpState' \in \controllerFinalStateSet_{\controller}\), or it reaches the end of its time horizon \(\timeHorizon_{\controller}\). 
If the former is true, the meta-policy selects the next subsystem to execute from \(\mdpState'\), and the process repeats.
Conversely, if the latter is true, the subsystem has failed to complete its subtask in time, and the execution of the meta-policy stops.
In the labyrinth example, the meta-policy selects which rooms to pass through, while the subsystems' policies navigate the individual rooms.
We note that although the environment may be partially observable, the \textit{abstract states} that are relevant to the meta-policy's decision-making are always fully observable.
This follows from our definition of composable subsystems, so long as the subsystems all complete their subtasks.
We define the notion of abstract states and elaborate on this point further in \S \ref{sec:hlm}.

The \textit{task} of the composite system is, beginning from an initial state \(\mdpInitialState\), to eventually reach a particular target exit condition \(\controllerFinalStateSet_{targ} \subseteq \mdpStateSet\). 
We assume that  \(\controllerFinalStateSet_{targ}\) is equivalent to \(\controllerFinalStateSet_{\controller}\) for at least one of the subsystems. 
That is, there is some subsystem \(\controller \in \controllerSet\) such that \(\controllerFinalStateSet_{targ} = \controllerFinalStateSet_{\controller}\).
Furthermore, to simplify theoretical analysis, we assume that for every \(\controller \in \controllerSet\), either \(\controllerFinalStateSet_{\controller} = \controllerFinalStateSet_{targ}\) or \(\controllerFinalStateSet_{\controller} \cap \controllerFinalStateSet_{targ} = \emptyset\).
This assumption removes ambiguity as to whether or not completion of a given subtask results in the immediate completion of the system's task.
Finally, we assume that at least one subsystem \(\controller\) can be executed from the initial state \(\mdpInitialState\), i.e. there exists a subsystem \(\controller \in \controllerSet\) such that \(\mdpInitialState \in \controllerInitialStateSet_{\controller}\).
We say that the execution of a meta-policy reaches the target set \(\controllerFinalStateSet_{targ}\), when one of the subsystems \(\controller\) with \(\controllerFinalStateSet_{\controller} = \controllerFinalStateSet_{targ}\) is executed, and successfully completes its subtask. 
We denote the probability of eventually reaching the target set from initial state \(\mdpInitialState\) by using meta-policy \(\abstractPolicy\) to compose the subsystem policies \(\policy_{\controller_{1}}, \ldots, \policy_{\controller_{\numControllers}}\), by \(\mathbb{P}_{\mdp}^{\mdpInitialState}(\Diamond \controllerFinalStateSet_{targ} | \abstractPolicy, \policy_{\controller_{1}}, \ldots, \controller_{\numControllers})\).

A \textit{task specification} places a requirement on the probability of the compositional RL system reaching \(\controllerFinalStateSet_{targ}\). That is, for some allowable failure probability \(\hlmFailProb \in [0,1]\), the task specification is satisfied if \(\mathbb{P}_{\mdp}^{\mdpInitialState}(\Diamond \controllerFinalStateSet_{targ} | \abstractPolicy, \policy_{\controller_{1}}, \ldots, \policy_{\controller_{\numControllers}}) \geq 1 - \hlmFailProb\).
With these definitions in place, we now deliver our problem statement.

\textbf{Problem Statement. }\textit{Given an allowable failure probability \(\hlmFailProb \in [0,1]\), an initial state \(\mdpInitialState\), a target set \(\controllerFinalStateSet_{targ}\), and a partially instantiated collection \(\controllerSet = \{\controller_{1}, \ldots, \controller_{\numControllers}\}\) of composable subsystems, learn policies \(\policy_{\controller_{i}}\) for each subsystem \(\controller_{i} \in \controllerSet\) and compute a meta-policy \(\abstractPolicy\) such that \(\mathbb{P}_{\mdp}^{\mdpInitialState}(\Diamond \controllerFinalStateSet_{targ} | \abstractPolicy, \policy_{\controller_{1}}, \ldots, \policy_{\controller_{\numControllers}}) \geq 1 - \hlmFailProb\).}

%% file: figures/labyrinth_both_overlay.tex
\newcommand{\pathwidth}{2mm}
\newcommand{\circleradius}{0.5}
\newcommand{\circlecolor}{blue!30!white}
\resizebox{1.0\textwidth}{!}{
\begin{tikzpicture}
\definecolor{color0}{rgb}{0.12156862745098,0.466666666666667,0.705882352941177}
\definecolor{color1}{rgb}{0.682352941176471,0.780392156862745,0.909803921568627}
\definecolor{color2}{rgb}{1,0.498039215686275,0.0549019607843137}
\definecolor{color3}{rgb}{1,0.733333333333333,0.470588235294118}
\definecolor{color4}{rgb}{0.172549019607843,0.627450980392157,0.172549019607843}
\definecolor{color5}{rgb}{0.596078431372549,0.874509803921569,0.541176470588235}
\definecolor{color6}{rgb}{0.83921568627451,0.152941176470588,0.156862745098039}
\definecolor{color7}{rgb}{1,0.596078431372549,0.588235294117647}
\definecolor{color8}{rgb}{0.580392156862745,0.403921568627451,0.741176470588235}
\definecolor{color9}{rgb}{0.772549019607843,0.690196078431373,0.835294117647059}
\definecolor{color10}{rgb}{0.549019607843137,0.337254901960784,0.294117647058824}
\definecolor{color11}{rgb}{0.768627450980392,0.611764705882353,0.580392156862745}

\filldraw[fill=black!2!white] (0,0) rectangle (20,20);
\fill[fill=black!50!white] (0,0) rectangle (20,1);
\fill[fill=black!50!white] (0,0) rectangle (1,20);
\fill[fill=black!50!white] (0,19) rectangle (20,20);
\fill[fill=black!50!white] (19,0) rectangle (20,20);
\fill[fill=black!50!white] (0,14) rectangle (3,15);
\fill[fill=black!50!white] (4,14) rectangle (10,15);
\fill[fill=black!50!white] (5,14) rectangle (6,17);
\fill[fill=black!50!white] (5,18) rectangle (6,20);
\fill[fill=black!50!white] (11,14) rectangle (14,15);
\fill[fill=black!50!white] (15,14) rectangle (20,15);
\fill[fill=black!50!white] (8,4) rectangle (9,15);
\fill[fill=black!50!white] (0,9) rectangle (5,10);
\fill[fill=black!50!white] (6,9) rectangle (9,10);
\fill[fill=black!50!white] (0,4) rectangle (3,5);
\filldraw[fill=orange] (2,12) rectangle (5,13);
\filldraw[fill=orange] (6,11) rectangle (8,12);
\filldraw[fill=orange] (6,5) rectangle (8,6);
\filldraw[fill=orange] (3,7) rectangle (5,8);
\fill[fill=black!50!white] (4,4) rectangle (16,5);
\fill[fill=black!50!white] (13,4) rectangle (14,15);
\fill[fill=black!50!white] (17,4) rectangle (20,5);
\node at (1.9, 18.1) [fill=\circlecolor] {\fontsize{80}{80}\selectfont\(\abstractInitialState\)};
\node at (3.0,2.0) [fill=green, rectangle] {\fontsize{50}{60}\selectfont\(\controllerFinalStateSet_{targ}\)};
\filldraw[fill=\circlecolor] (3.5, 14.5) circle (\circleradius);
\filldraw[fill=\circlecolor] (5.5, 17.5) circle (\circleradius);
\filldraw[fill=\circlecolor] (10.5, 14.5) circle (\circleradius);
\filldraw[fill=\circlecolor] (14.5, 14.5) circle (\circleradius);
\filldraw[fill=\circlecolor] (5.5, 9.5) circle (\circleradius);
\filldraw[fill=\circlecolor] (3.5, 4.5) circle (\circleradius);
\filldraw[fill=\circlecolor] (12.5, 5.5) circle (\circleradius);
\filldraw[fill=\circlecolor] (16.5, 4.5) circle (\circleradius);
\filldraw[fill=\circlecolor] (10.5, 2.0) circle (\circleradius);

\draw[->, color=color0, line width=\pathwidth] (2.0, 17.5) -- (2.0, 15.5);
\draw[color=color0, line width=\pathwidth] (2.0, 15.5) -- (3.5, 15.5);
\draw[->, color=color0, line width=\pathwidth] (3.5, 15.5) -- (3.5, 15.0);
\node at (2.7, 16.5) [color=color0, xshift=5] {\resizebox{1.5cm}{!}{\(\bm{\controller_{0}}\)}};
\draw[->, color=color1, line width=\pathwidth] (3.0, 18.5) -- (4.5, 18.5);
\draw[color=color1, line width=\pathwidth] (4.5, 18.5) -- (4.5, 17.5);
\draw[->, color=color1, line width=\pathwidth] (4.5, 17.5) -- (5.0, 17.5);
\node at (3.7, 17.5) [color=color1] {\resizebox{1.5cm}{!}{\(\bm{\controller_{1}}\)}};
\draw[color=color2, line width=\pathwidth] (6.0, 17.5) -- (6.5, 17.5);
\draw[->, color=color2, line width=\pathwidth] (6.5, 17.5) -- (6.5, 15.5);
\draw[color=color2, line width=\pathwidth] (6.5, 15.5) -- (10.5, 15.5);
\draw[->, color=color2, line width=\pathwidth] (10.5, 15.5) -- (10.5, 15.0);
\node at (8.5, 16.2) [color=color2] {\resizebox{1.5cm}{!}{\(\bm{\controller_{2}}\)}};
\draw[->, color=color3, line width=\pathwidth] (6.0, 17.5) -- (14.5, 17.5);
\draw[->, color=color3, line width=\pathwidth] (14.5, 17.5) -- (14.5, 15.0);
\node at (15.3, 16.5) [color=color3, xshift=15] {\resizebox{1.5cm}{!}{\(\bm{\controller_{3}}\)}};
\draw[color=color4, line width=\pathwidth] (3.5, 14.0) -- (3.5, 13.5);
\draw[->, color=color4, line width=\pathwidth] (3.5, 13.5) -- (5.5, 13.5);
\draw[->, color=color4, line width=\pathwidth] (5.5, 13.5) -- (5.5, 10.0);
\node at (4.7, 11.0) [color=color4, xshift=-10] {\resizebox{1.5cm}{!}{\(\bm{\controller_{4}}\)}};
\draw[->, color=color5, line width=\pathwidth] (5.5, 9.0) -- (5.5, 5.5);
\draw[color=color5, line width=\pathwidth] (5.5, 5.5) -- (3.5, 5.5);
\draw[->, color=color5, line width=\pathwidth] (3.5, 5.5) -- (3.5, 5.0);
\node at (6.3, 8.0) [color=color5, xshift=10] {\resizebox{1.5cm}{!}{\(\bm{\controller_{5}}\)}};
\draw[color=color6, line width=\pathwidth] (10.5, 14.0) -- (10.5, 13.5);
\draw[color=color6, line width=\pathwidth] (10.5, 13.5) -- (9.5, 13.5);
\draw[->, color=color6, line width=\pathwidth] (9.5, 13.5) -- (9.5, 9.5);
\draw[->, color=color6, line width=\pathwidth] (9.5, 9.5) -- (9.5, 5.5);
\draw[->, color=color6, line width=\pathwidth] (9.5, 5.5) -- (12.0, 5.5);
\node at (10.3, 12.0) [color=color6, xshift=5] {\resizebox{1.5cm}{!}{\(\bm{\controller_{6}}\)}};
\draw[->, color=color7, line width=\pathwidth] (12.5, 6.0) -- (12.5, 9.5);
\draw[->, color=color7, line width=\pathwidth] (12.5, 9.5) -- (12.5, 13.5);
\draw[color=color7, line width=\pathwidth] (12.5, 13.5) -- (10.5, 13.5);
\draw[color=color7, line width=\pathwidth] (10.5, 13.5) -- (10.5, 14.0);
\node at (11.7, 8.0) [color7, xshift=-5] {\resizebox{1.5cm}{!}{\(\bm{\controller_{7}}\)}};
\draw[->, color=color8, line width=\pathwidth] (14.5, 14.0) -- (14.5, 9.5);
\draw[->, color=color8, line width=\pathwidth] (14.5, 9.5) -- (16.5, 9.5);
\draw[->, color=color8, line width=\pathwidth] (16.5, 9.5) -- (16.5, 5.0);
\node at (17.0, 10.0) [color=color8, xshift=-40, yshift=30] {\resizebox{1.5cm}{!}{\(\bm{\controller_{8}}\)}};
\draw[->, color=color9, line width=\pathwidth] (3.5, 4.0) -- (3.5, 3.0);
\node at (4.3, 3.5) [color=color9, xshift=20, yshift=-2] {\resizebox{1.5cm}{!}{\(\bm{\controller_{9}}\)}};
\draw[->, color=color10, line width=\pathwidth] (16.5, 4.0) -- (16.5, 2.0);
\draw[->, color=color10, line width=\pathwidth] (16.5, 2.0) -- (11.0, 2.0);
\node at (13.5, 2.7) [color=color10] {\resizebox{2.0cm}{!}{\(\bm{\controller_{10}}\)}};
\draw[->, color=color11, line width=\pathwidth] (10.0, 2.0) -- (5.0, 2.0);
\node at (7.5, 2.7) [color=color11] {\resizebox{2.0cm}{!}{\(\bm{\controller_{11}}\)}};
\end{tikzpicture}%
}%

%% file: figures/hlm.tex

\tikzstyle{branch}=[fill,shape=circle,minimum size=5pt,inner sep=0pt]
\def\horizontaldistance{2.5cm}
\def\verticaldistance{2.0cm}
\def\nodedistance{2.5cm}
\def\branchdist{0.8cm}

\def\stateOutlineThickness{0.5mm}
\def\edgeThickness{0.5mm}

\tikzset{auto, ->, >=stealth, node distance=\nodedistance, node/.style={scale=0.8, minimum size=0pt, inner sep=0pt}}

\resizebox{1.0\textwidth}{!}{
\begin{tikzpicture}[scale=1.0]
    \node[state, line width=\stateOutlineThickness] (s_init) {\(\abstractInitialState\)};
    \path (s_init)+(\horizontaldistance, 0.0) node (s0) [state, line width=\stateOutlineThickness] {\(\abstractState_0\)};
    \path (s_init)+(0.0, -\verticaldistance) node (s1) [state, line width=\stateOutlineThickness] {\(\abstractState_1\)};
    \path (s1)+(0.0, -\verticaldistance) node (s4) [state, line width=\stateOutlineThickness] {\(\abstractState_4\)};
    \path (s4)+(\horizontaldistance, 0.0) node (s5) [state, line width=\stateOutlineThickness] {\(\abstractState_5\)};
    
    \path (s0)+(1.8cm, -1.0cm) node (s2) [state, line width=\stateOutlineThickness] {\(\abstractState_2\)};
    \path (s2)+(0.0, -\verticaldistance) node (s6) [state, line width=\stateOutlineThickness] {\(\abstractState_6\)};
    
    \path (s0)+(5.0cm, 0.0) node (s3) [state, line width=\stateOutlineThickness] {\(\abstractState_3\)};
    \path (s3)+(0.0, -\verticaldistance) node (s7) [state, line width=\stateOutlineThickness] {\(\abstractState_7\)};
    \path (s7)+(0.0, -\verticaldistance) node (s8) [state, line width=\stateOutlineThickness] {\(\abstractState_8\)};
    
    \path (s1)+(2.2cm, 0.0) node (s_fail) [state, draw=red, line width=1.0mm] {\(\abstractFailureState\)};
    \path (s_fail)+(3.7cm, 0.0) node (s_fail2) [state, draw=red, line width=1.0mm] {\(\abstractFailureState\)};
    \path (s5)+(3.0cm, 0.0) node (s_goal) [state, draw=green, line width=1.0mm] {\(\abstractSuccessState\)};
    
    \path (s_init)+(\branchdist, 0.0) node [branch] (c1) {};
    \path [draw, line width=\edgeThickness, -] (s_init.east) -- node [above, xshift=1mm] {\(\controller_1\)} (c1);
    \path [draw, ->, line width=\edgeThickness] (c1) -- node [above] {\(\bernoulliProb_{\controller_1}\)} (s0);
    \path [draw=red, ->, line width=\edgeThickness] (c1) -- node [above] {} (s_fail);
    \path (s_init)+(0.0, -\branchdist) node [branch] (c0) {};
    \path [draw, -, line width=\edgeThickness] (s_init.south) -- node [left] {\(\controller_0\)} (c0);
    \path [draw, ->, line width=\edgeThickness] (c0) -- node [left] {\(\bernoulliProb_{\controller_0}\)} (s1);
    \path [draw=red, ->, line width=\edgeThickness] (c0) -- node [above] {} (s_fail);
    \path (s0)+(0.87cm, -0.33cm) node [branch] (c2) {};
    \path [draw, -, line width=\edgeThickness] (s0.east) -- node [below, xshift=-0.7mm, yshift=0.7mm] {\(\controller_2\)} (c2);
    \path [draw, ->, line width=\edgeThickness] (c2) -- node [below, xshift=-1.5mm, yshift=0.8mm] {\(\bernoulliProb_{\controller_2}\)} (s2);
    \path [draw=red, ->, line width=\edgeThickness] (c2) -- node [above] {} (s_fail);
    \path (s0)+(2.2cm, 0.0cm) node [branch] (c3) {};
    \path [draw, -, line width=\edgeThickness] (s0.east) -- node [above] {\(\controller_3\)} (c3);
    \path [draw, ->, line width=\edgeThickness] (c3) -- node [above] {\(\bernoulliProb_{\controller_3}\)} (s3);
    \path (c3)+(-0.5, -0.5) node (c3_fail) {};
    \path [draw=red, ->, line width=\edgeThickness] (c3) -- node [above] {} (s_fail2);
    \path (s2)+(-0.4cm, -\branchdist) node [branch] (c6) {};
    \path [draw, -, line width=\edgeThickness] (s2.south) -- node [left, yshift=0.7mm, xshift=0.5mm] {\(\controller_6\)} (c6);
    \path [draw, ->, line width=\edgeThickness] (c6) -- node [left] {\(\bernoulliProb_{\controller_6}\)} (s6.north);
    \path [draw=red, ->, line width=\edgeThickness] (c6) -- node [above] {} (s_fail);
    \path (s6)+(0.4cm, 0.8cm) node [branch] (c7) {};
    \path [draw, -, line width=\edgeThickness] (s6.north) -- node [right, yshift=-1mm] {\(\controller_7\)} (c7);
    \path [draw, ->, line width=\edgeThickness] (c7) -- node [right] {\(\bernoulliProb_{\controller_7}\)} (s2.south);
    \path (c7)+(-0.6, 0.05) node (c7_fail) {};
    \path [draw=red, ->, line width=\edgeThickness] (c7) -- node [above] {} (s_fail2);
    \path (s1)+(0.0, -\branchdist) node [branch] (c4) {};
    \path [draw, -, line width=\edgeThickness] (s1.south) -- node [left] {\(\controller_4\)} (c4);
    \path [draw, ->, line width=\edgeThickness] (c4) -- node [left] {\(\bernoulliProb_{\controller_4}\)} (s4);
    \path [draw=red, ->, line width=\edgeThickness] (c4) -- node [above] {} (s_fail);
    \path (s4)+(\branchdist, 0.0) node [branch] (c5) {};
    \path [draw, -, line width=\edgeThickness] (s4.east) -- node [below, xshift=1mm] {\(\controller_5\)} (c5);
    \path [draw, ->, line width=\edgeThickness] (c5) -- node [below] {\(\bernoulliProb_{\controller_5}\)} (s5);
    \path [draw=red, ->, line width=\edgeThickness] (c5) -- node [above] {} (s_fail);
    \path (s5)+(\branchdist, 0.0) node [branch] (c9) {};
    \path [draw, -, line width=\edgeThickness] (s5.east) -- node [below, xshift=1mm] {\(\controller_9\)} (c9);
    \path [draw, ->, line width=\edgeThickness] (c9) -- node [below] {\(\bernoulliProb_{\controller_9}\)} (s_goal);
    \path [draw=red, ->, line width=\edgeThickness] (c9) -- node [above] {} (s_fail);
    \path (s3)+(0.0, -\branchdist) node [branch] (c8) {};
    \path [draw, -, line width=\edgeThickness] (s3.south) -- node [right, yshift=-1mm] {\(\controller_8\)} (c8);
    \path [draw, ->, line width=\edgeThickness] (c8) -- node [right] {\(\bernoulliProb_{\controller_8}\)} (s7);
    \path (c8)+(-1.0, -0.3) node (c8_fail) {};
    \path [draw=red, ->, line width=\edgeThickness] (c8) -- node [above] {} (s_fail2);
    
    \path (s7)+(0.0, -\branchdist) node [branch] (c10) {};
    \path [draw, -, line width=\edgeThickness] (s7.south) -- node [right, yshift=-1mm] {\(\controller_{10}\)} (c10);
    \path [draw, ->, line width=\edgeThickness] (c10) -- node [right] {\(\bernoulliProb_{\controller_{10}}\)} (s8);
    \path (c10)+(-1.0, 0.2) node (c10_fail) {};
    \path [draw=red, ->, line width=\edgeThickness] (c10) -- node [above] {} (s_fail2);
    \path (s8)+(-\branchdist, 0.0) node [branch] (c11) {};
    \path [draw, -, line width=\edgeThickness] (s8.west) -- node [below, xshift=-0.8mm] {\(\controller_{11}\)} (c11);
    \path [draw, ->, line width=\edgeThickness] (c11) -- node [below] {\(\bernoulliProb_{\controller_{11}}\)} (s_goal);
    \path (c11)+(-0.8, 0.5) node (c11_fail) {};
    \path [draw=red, ->, line width=\edgeThickness] (c11) -- node [above] {} (s_fail2);
    
    \path (s_fail) edge [loop above, draw=red, line width=\edgeThickness] node {1} (s_fail);
    \path (s_fail2) edge [loop above, draw=red, line width=\edgeThickness] node {1} (s_fail2);
    \path (s_goal) edge [loop above, draw=green, line width=\edgeThickness] node {1} (s_goal);
\end{tikzpicture}
}

%% file: tex/04_high_level_model.tex
\section{The High-Level Decision-Making Model}\label{sec:hlm}
We now introduce the high-level model (HLM) of the compositional RL framework, which is used to compute meta-policies, and to decompose task specifications into subtask specifications to be satisfied by the individual subsystems.

\subsection{Defining the High-Level Model (HLM)}
To construct the HLM, we use a given collection \(\controllerSet = \{\controller_1, \controller_2, \ldots, \controller_{\numControllers}\}\) of partially instantiated subsystems, an initial state \(\mdpInitialState\), and a target set \(\controllerFinalStateSet_{targ}\).
We begin by defining a state abstraction, which groups together environment states in order to define the state space of the HLM.
To do so, we define the equivalence relation \(\eqRelation \subseteq \mdpStateSet \times \mdpStateSet\) such that \((\mdpState, \mdpState') \in \eqRelation\) if and only if the following two conditions hold.
\begin{align*}
	\begin{array}{l}
		 \text{1. For every \(\controller \in \controllerSet, \mdpState \in \controllerInitialStateSet_{\controller}\) if and only if \(\mdpState' \in \controllerInitialStateSet_{\controller}\), and}, \\
		 \text{2. \(\mdpState \in \controllerFinalStateSet_{targ}\) if and only if \(\mdpState' \in \controllerFinalStateSet_{targ}\).}
	\end{array}
\end{align*}%
The equivalence class of any state \(\mdpState \in \mdpStateSet\) under equivalence relation \(\eqRelation\) is given by \([\mdpState]_{\eqRelation} = \{\mdpState' \in \mdpStateSet | (\mdpState, \mdpState') \in \eqRelation\}\). 
The quotient set of \(\mdpStateSet\) by \(\eqRelation\) is defined as the set of all equivalence classes \(\mdpStateSet /_{\eqRelation} = \{[\mdpState]_{\eqRelation} | \mdpState\in \mdpStateSet\}\).
Intuitively, this equivalence relation groups together all the states in the target set, and it also groups together states that are entry conditions to the same subset of subsystems.

We may now define the HLM corresponding to the collection of subsystems \(\controllerSet\) by the parametric MDP \(\abstractMDP = (\abstractStateSet, \abstractInitialState, \abstractSuccessState, \abstractFailureState, \controllerSet, \abstractTransition)\).
Here, the high-level states \(\abstractStateSet\) are defined to be \(\mdpStateSet/_{\eqRelation}\); states in the HLM correspond to equivalence classes of environment states.
The initial state \(\abstractInitialState\) of the HLM is defined as \(\abstractInitialState = [\mdpInitialState]_{\eqRelation}\), the equivalence class of the environment's initial state. 
The \textit{goal state} \(\abstractSuccessState \in \abstractStateSet\) is similarly defined as \([\mdpState]_{\eqRelation}\) such that \(\mdpState \in \controllerFinalStateSet_{targ}\). 
Recall that \(\controllerFinalStateSet_{targ} = \controllerFinalStateSet_{\controller}\) for at least one of the subsystems \(\controller \in \controllerSet\).
Finally, the \textit{failure state} \(\abstractFailureState \in \abstractStateSet\) is defined as \([\mdpState]_{\eqRelation}\) such that \(\mdpState \in \mdpStateSet \setminus [\bigcup_{\controller \in \controllerSet} \controllerInitialStateSet_{\controller}] \cup \controllerFinalStateSet_{targ}\), i.e., the equivalence class of states \textit{not} belonging to the initial states of any component, or to the target set.

As an example, Figure \ref{fig:running_example} illustrates the HLM corresponding to the collection of subsystems from labyrinth example. 
The overlapping entry and exit conditions, represented by the blue circles in Figure \ref{fig:running_example}, define the states of the HLM. 
The target set \(\controllerFinalStateSet_{targ}\) defines the HLM's goal state \(\abstractSuccessState\), and all other environment states are absorbed into the failure state \(\abstractFailureState\).

The collection of subsystems \(\controllerSet\) defines the HLM's set of actions. By definition of the equivalence relation \(\eqRelation\), for every HLM state \(\abstractState \in \abstractStateSet\) there is a well-defined subset of the subsystems \(\controllerSet(\abstractState) \subseteq \controllerSet\) that can be executed.
That is, for every environment state \(\mdpState \in \abstractState\), \(\mdpState \in \controllerInitialStateSet_{\controller}\) for all \(\controller \in \controllerSet(\abstractState)\).
We accordingly define \(\controllerSet(\abstractState)\) as the set of \textit{available subsystems} at high-level state \(\abstractState\).

Furthermore, consider any subsystem \(\controller \in \controllerSet(\abstractState)\). 
As a direct result of the definition of equivalence relation \(\eqRelation\) and of the subsystems in collection \(\controllerSet\) being composable, every state \(\mdpState\) within set \(\controllerFinalStateSet_{\controller}\) belongs to the \textit{same} equivalence class \([\mdpState]_{\eqRelation}\).
In other words, we may uniquely define the successor HLM state of any component \(\controller \in \controllerSet\) as \(succ(\controller) = [\mdpState]_{\eqRelation}\) such that \(\mdpState \in \controllerFinalStateSet_{\controller}\). We then construct the HLM transition probability function in terms of parameters \(\bernoulliProb_{\controller} \in [0,1]\) as follows.%
\begin{align*}
\abstractTransition(\abstractState, \controller, \abstractState') = \begin{cases} 
  \bernoulliProb_{\controller}, & if\; \; \controller \in \controllerSet(\abstractState), \;\; \abstractState' = succ(\controller) \\
  1 - \bernoulliProb_{\controller}, & if \; \; \controller\in \controllerSet(\abstractState), \; \; \abstractState' = \abstractFailureState\\
  0, & \text{Otherwise}
\end{cases}\end{align*}%
The interpretation of this definition of \(\abstractTransition\) is as follows. 
After selecting component \(\controller \in \controllerSet(\abstractState)\) from HLM state \(\abstractState\), the component either succeeds in reaching an exit condition  \(\mdpState \in \controllerFinalStateSet_{\controller}\) within its time horizon \(\timeHorizon_{\controller}\) with probability \(\bernoulliProb_{\controller}\), resulting in an HLM transition to \(succ(\controller)\), or it fails to do so with probability \(1 - \bernoulliProb_{\controller}\), resulting in a transition to the HLM failure state \(\abstractFailureState\). 

The parameters \(\bernoulliProb_{\controller}\) may thus be interpreted as estimates of the probabilities that the subsystems complete their subtasks, given they are executed from one of their entry conditions.
Their values come either from empirical rollouts of learned subsystem policies \(\policy_{\controller}\), or as the solution to the aforementioned automatic decomposition of the task specification, which is discussed further below.

\subsection{Relating the HLM to Compositions of RL Subsystems}
We note that while the parameters \(\bernoulliProb_{\controller}\) are meant to estimate the probabilities of successful subtask completion, they cannot capture these probabilities exactly.
In reality, while parameter \(\bernoulliProb_{\controller}\) is constant, it's possible for this probability to vary, given the entry condition \(\mdpState \in \controllerInitialStateSet_{\controller}\) from which the component is executed.
However, the simplicity of the presented parametrization of \(\abstractTransition\) enables tractable solutions to planning and verification problems in \(\abstractMDP\).
Furthermore, by establishing relationships between policies in \(\abstractMDP\), and meta-policies composing RL subsystems, the HLM becomes practically useful in the analysis of composite RL systems.

Towards this idea, we recall that each state \(\abstractState \in \abstractStateSet\) of the HLM uniquely defines a set of available subsystems.
As a result, any stationary policy \(\hlmPolicy : \abstractStateSet \times \controllerSet \to [0,1]\) may be used to define a meta-policy that composes RL subsystems, as described in \S \ref{sec:problem_statement_meta_policies}.
Thus, solutions to planning problems in \(\abstractMDP\) can be used directly as meta-policies to specify compositions of the RL subsystems. 
Of particular interest is the problem of computing an HLM policy \(\hlmPolicy\) that maximizes \(\mathbb{P}^{\abstractInitialState}_{\abstractMDP}(\Diamond \abstractSuccessState | \hlmPolicy)\), i.e., the probability of eventually reaching the goal state \(\abstractSuccessState\) from the HLM's initial state \(\abstractInitialState\).
Theorem \ref{thm:hlm_bounds_true_performance} relates this probability value to the probability of the corresponding compositional system completing its task, \(\mathbb{P}_{\mdp}^{\mdpInitialState}(\Diamond \controllerFinalStateSet_{targ} | \abstractPolicy, \policy_{\controller_{1}, \ldots, \policy_{\controller_{\numControllers}}})\), in the task environment \(\mdp\).%
\begin{theorem}
\label{thm:hlm_bounds_true_performance}
Let \(\controllerSet = \{\controller_1, \controller_2, ..., \controller_{\numControllers}\}\) be a collection of composable subsystems with respect to initial state \(\mdpInitialState\) and target set \(\controllerFinalStateSet_{targ}\) within the environment POMDP \(\mdp\). Define \(\abstractMDP\) to be the corresponding HLM and let \(\hlmPolicy\) be a policy in \(\abstractMDP\). 
If, for every subsystem \(\controller_{i} \in \controllerSet\) and for every entry condition \(\mdpState \in \controllerInitialStateSet_{\controller_{i}}\), 
\(\mathbb{P}^{\mdpState}_{\mdp}(\Diamond_{\leq \timeHorizon_{\controller_{i}}}\controllerFinalStateSet_{\controller_{i}} | \policy_{\controller_{i}}) \geq \bernoulliProb_{\controller_{i}}\)
, then 
\[\mathbb{P}^{\mdpInitialState}_{\mdp}(\Diamond \controllerFinalStateSet_{targ} | \abstractPolicy, \policy_{\controller_{1}}, \ldots, \policy_{\controller_{\numControllers}}) \geq \mathbb{P}^{\abstractInitialState}_{\abstractMDP}(\Diamond \abstractSuccessState | \hlmPolicy).\]
\end{theorem}%
In words, Theorem \ref{thm:hlm_bounds_true_performance} says that if each subsystem's probability of success is lower bounded by \(\bernoulliProb_{\controller}\), then under any given meta-policy, the probability of completing the task in the true environment is lower-bounded by the probability of reaching the goal state in the HLM.

We provide a full proof of this result in Appendix \ref{sec:thm1_proof}.
However, a sketch of the procedure used to prove the result is summarized as follows.
We begin by expressing \(\mathbb{P}_{\abstractMDP}^{\abstractInitialState}(\Diamond \abstractSuccessState | \hlmPolicy)\) as the sum of the probabilities of every sequence \(\abstractState_0 \controller_0 \abstractState_1 \controller_1 ... \abstractState_{\numMetaDecision}\) that eventually reaches the goal state \(\abstractSuccessState\) in the HLM \(\abstractMDP\). 
Then, we show that for each such sequence there exists a corresponding set of \textit{consistent} sequences \(\mdpState_0 \hat{\controller}_0 \mdpAction_0 \mdpState_1 \hat{\controller}_1 \mdpAction_1 ... \mdpState_{\numTimeStep}\) of states, subsystems, and actions in the environment POMDP \(\mdp\). 
Furthermore, the probability value \(\mathbb{P}_{\mdp}^{\abstractInitialState}(\cdot | \abstractPolicy, \policy_{\controller_{1}}, \ldots, \policy_{\controller_{\numControllers}})\) of each such set of consistent sequences in \(\mdp\) is higher than the probability value \(\mathbb{P}_{\abstractMDP}^{\abstractInitialState}(\abstractState_0 \controller_0 \abstractState_1 \controller_1 ... \abstractState_{\numMetaDecision} | \hlmPolicy)\) of the original sequence in the HLM \(\abstractMDP\).
Finally, we note that every such sequence \\ \(\mdpState_0 \hat{\controller}_0 \mdpAction_0 \mdpState_1 \hat{\controller}_1 \mdpAction_1 ... \mdpState_{\numTimeStep}\) of environment states, subsystems, and actions eventually reaches the target set \(\controllerFinalStateSet_{targ}\) and that the aforementioned sets of these sequences are pairwise disjoint. 
From these points, we are able to conclude the result of the theorem.

As a practical demonstration of the theorem's application, consider the labyrinth task from and its corresponding HLM, illustrated in Figure \ref{fig:running_example}.
Suppose the HLM's parameters \(\bernoulliProb_{\controller}\) are specified such that they lower bound the true probabilities of subtask success, i.e. the transition probabilities in the HLM lower bound the probabilities of the subsystems successfully navigating their respective rooms in the labyrinth in Figure \ref{fig:running_example}. 
By planning a policy \(\hlmPolicy\) in the HLM that, for example, reaches \(\abstractSuccessState\) with probability \(0.95\), we ensure that the corresponding composition of the subsystems will reach \(\controllerFinalStateSet_{targ}\) in the labyrinth with a probability of \textit{at least} \(0.95\).

\subsection{Automatic Decomposition of Task Specifications}
Recall that our objective is not only to compute a meta-policy \(\abstractPolicy\), but also to \textit{learn} the subsystem policies \(\policy_{\controller_1}, \policy_{\controller_2},..., \policy_{\controller_{\numControllers}}\) that this meta-policy will execute, such that the system's task specification \(\mathbb{P}^{\mdpInitialState}_{\mdp}(\Diamond \controllerFinalStateSet_{targ}| \abstractPolicy, \policy_{\controller_{1}}, \ldots, \policy_{\controller_{\numControllers}}) \geq 1 - \hlmFailProb\) is satisfied.
Suppose that we choose a set of HLM parameters \(\{\bernoulliProb_{\controller_1}, \bernoulliProb_{\controller_2}, ..., \bernoulliProb_{\controller_{\controller_{\numControllers}}}\}\) such that a policy \(\hlmPolicy\) in the HLM exists with \(\mathbb{P}^{\abstractInitialState}_{\abstractMDP}(\Diamond \abstractSuccessState | \hlmPolicy) \geq 1 - \hlmFailProb\).
Then, so long as each of the corresponding subsystems \(\controller\) are able to learn a policy \(\policy_{\controller}\) such that \(\successProb_{\policy_{\controller}}^{\controller}(\mdpState) \geq \bernoulliProb_{\controller}\) for every \(\mdpState \in \controllerInitialStateSet_{\controller}\), Theorem \ref{thm:hlm_bounds_true_performance} tells us that the meta-policy will satisfy the task specification.

We may thus interpret the values of parameters \(\bernoulliProb_{\controller}\) as \textit{subtask specifications}.
Each subsystem must achieve one of its exit conditions \(\mdpState' \in \controllerFinalStateSet_{\controller}\) within its allowed time horizon \(\timeHorizon_{\controller}\) with a probability of at least \(\bernoulliProb_{\controller}\), given its execution began from some entry condition \(\mdpState \in \controllerInitialStateSet_{\controller}\).
With this interpretation in mind, we take the following approach to the decomposition of the task specification: find the smallest values of parameters \(\bernoulliProb_{\controller_1}, \bernoulliProb_{\controller_2}, ..., \bernoulliProb_{\controller_{\numControllers}}\) such that an HLM policy \(\hlmPolicy\) exists satisfying \(\mathbb{P}^{\abstractInitialState}_{\abstractMDP}(\Diamond \abstractSuccessState | \hlmPolicy) \geq 1 - \hlmFailProb\).
We formulate this constrained parameter optimization problem as the bilinear program given in equations (\ref{eq:hlm_opt_objective})-(\ref{eq:hlm_opt_task_sat_constraints}).
In (\ref{eq:hlm_opt_dynamics_constraints}) and (\ref{eq:hlm_opt_task_sat_constraints}), we define \(pred(\abstractState) \defeq \{(\abstractState', \controller') | \controller' \in \controllerSet(\abstractState') \; and \; \abstractState = succ(\controller')\}\).%
\begin{align}
    \min_{\occupancyVar, \bernoulliProb_{\controller}} \quad & \sum_{\controller \in \controllerSet} \bernoulliProb_{\controller}\label{eq:hlm_opt_objective} \\
    \textrm{s.t.} \quad & \sum_{\controller \in \controllerSet(\abstractState)} \occupancyVar(\abstractState, \controller) = \delta_{\abstractInitialState}(\abstractState) + \sum_{(\abstractState', \controller') \in pred(\abstractState)} \occupancyVar(\abstractState', \controller') \bernoulliProb_{\controller'}, \quad \forall \abstractState \in  \abstractStateSet \setminus \{\abstractFailureState, \abstractSuccessState\} \label{eq:hlm_opt_dynamics_constraints} \\
    & \occupancyVar(\abstractState, \controller) \geq 0, \; \forall \abstractState \in \abstractStateSet \setminus \{\abstractFailureState, \abstractSuccessState\},\; \forall \controller \in \controllerSet(\abstractState) \\
    & 0 \leq \bernoulliProb_{\controller} \leq 1, \; \forall \controller \in \controllerSet \label{eq:hlm_opt_pc_0_1_constraints} \\
    & \sum_{(\abstractState', \controller') \in pred(\abstractSuccessState)} \occupancyVar(\abstractState', \controller') \bernoulliProb_{\controller'} \geq 1 - \hlmFailProb
    \label{eq:hlm_opt_task_sat_constraints}
\end{align}%
The decision variables in (\ref{eq:hlm_opt_objective})-(\ref{eq:hlm_opt_task_sat_constraints}) are the HLM parameters \(\bernoulliProb_{\controller}\) for every \(\controller \in \controllerSet\), and \(\occupancyVar(\abstractState, \controller)\) for every \(\abstractState \in \abstractStateSet \setminus \{\abstractFailureState, \abstractSuccessState\}\).
The value of \(\delta_{\abstractInitialState}(\abstractState)\) is $1$ if \(\abstractState=\abstractInitialState\) and $0$ otherwise.
The constraint~\eqref{eq:hlm_opt_dynamics_constraints} is the so-called Bellman-flow constraint; it ensures that the variable \(\occupancyVar(\abstractState, \controller)\) defines the expected number of times subsystem \(\controller\) is executed in state \(\abstractState\).
Given the values of the variables \(\occupancyVar(\abstractState, \controller)\) for every \(\abstractState \in \abstractStateSet \setminus \{\abstractFailureState, \abstractSuccessState\}\) and \(\controller \in \controllerSet(\abstractState)\), we may define the corresponding HLM policy as \(\hlmPolicy(\abstractState, \controller) \defeq \frac{\occupancyVar(\abstractState, \controller)}{\sum_{\controller' \in \controllerSet(\abstractState)} \occupancyVar(\abstractState, \controller')}\) \cite{puterman2014markov}.
The constraint \eqref{eq:hlm_opt_task_sat_constraints} enforces the HLM policy \(\hlmPolicy\)'s satisfaction of \(\mathbb{P}^{\abstractInitialState}_{\abstractMDP}(\Diamond \abstractSuccessState | \hlmPolicy) \geq 1 - \hlmFailProb\).
We refer to~\shortciteA{etessami2007multi} and~\citeA{puterman2014markov} for further details on these variables and the constraints.

%% file: tex/05_icrl.tex
\section{Iterative Compositional Reinforcement Learning (ICRL)}
\label{sec:icrl}

In this section, we discuss how subsystem policies are learned to satisfy the subtask specifications discussed above, and we present how the bilinear program given in (\ref{eq:hlm_opt_objective})-(\ref{eq:hlm_opt_task_sat_constraints}) is modified to refine the subtask specifications, after some training of the subsystems has been completed.

\subsection{Learning and Verifying Subsystem Policies}
Let \(\bernoulliProb_{\controller_1}, \bernoulliProb_{\controller_2}, ..., \bernoulliProb_{\controller_{\numControllers}}\) be the parameter values output as a solution to problem (\ref{eq:hlm_opt_objective})-(\ref{eq:hlm_opt_task_sat_constraints}). We want each subsystem \(\controller\) to learn a policy \(\policy_{\controller}\) satisfying the subtask specification:  \(\successProb_{\policy_{\controller}}^{\controller}(\mdpState) \geq \bernoulliProb_{\controller}\) for each entry condition \(\mdpState \in \controllerInitialStateSet_{\controller}\) of the subsystem. 
We note that any RL algorithm and reward function may be used, so long as the resulting learned policy can be verified to satisfy its subtask specification.
A particularly simple candidate reward function \(\mdpRewardFunction_{\controller}\) outputs \(1\) when an exit condition \(\mdpState \in \controllerFinalStateSet_{\controller}\) is first reached, and outputs \(0\) otherwise. 
Under this reward function, we have  \(\successProb_{\policy_{\controller}}^{\controller}(\mdpState) = \mathbb{E}[\sum_{t\in[\timeHorizon_{\controller}]}\mdpRewardFunction_{\controller}(\mdpState_t) | \policy_{\controller}, \mdpState_0 = \mdpState]\). We can maximize the probability of reaching an exit condition by maximizing the expected undiscounted sum of rewards.

To verify that a learned subsystem policy \(\policy_{\controller}\) satisfies its subtask specification, we consider \(\controllerInfProb_{\controller} = \inf \{ \successProb_{\policy_{\controller}}^{\controller}(\mdpState)|\mdpState \in \controllerInitialStateSet_{\controller}\}\), the greatest lower bound of the policy's probability of subtask succcess, beginning from any of the subsystem's entry conditions. So long as \(\controllerInfProb_{\controller} \geq \bernoulliProb_{\controller}\), the subtask specification is satisfied.
In practice, the value of \(\controllerInfProb_{\controller}\) cannot be known exactly, but we may obtain an estimate \(\controllerPerformanceEstimate_{\controller}\) of its value through empirical rollouts of \(\policy_{\controller}\), beginning from the different entry conditions \(\mdpState \in \controllerInitialStateSet_{\controller}\).
We refer to \(\controllerPerformanceEstimate_{\controller}\) as the \textit{estimated performance value} of policy \(\policy_{\controller}\).

\subsection{Automatic Refinement of the Subtask Specifications}
The estimated performance values \(\controllerPerformanceEstimate_{\controller}\) are useful not only for the empirical verification of the learned policies, but also as additional information used periodically during training to refine the subtask specifications. To do so, we re-solve the optimization problem (\ref{eq:hlm_opt_objective})-(\ref{eq:hlm_opt_task_sat_constraints}), with a modified objective (\ref{eq:hlm_opt_modified_obj}), and additional constraints \eqref{eq:hlm_opt_lb_constraints}-\eqref{eq:hlm_opt_ub_constraints}.%
\begin{align}
    & obj(\lbList) = \sum_{\controller \in \controllerSet} (\bernoulliProb_{\controller} - \controllerPerformanceEstimate_{\controller})
    \label{eq:hlm_opt_modified_obj}\\
    & LBConst(\lbList) = \{\bernoulliProb_{\controller} \geq \controllerPerformanceEstimate_{\controller} | \forall \controllerPerformanceEstimate_{\controller} \in \lbList\}
    \label{eq:hlm_opt_lb_constraints}\\
    & UBConst(\ubList) = \{\bernoulliProb_{\controller} \leq \controllerPerformanceEstimate_{\controller} | \forall \controllerPerformanceEstimate_{\controller} \in \ubList\}
    \label{eq:hlm_opt_ub_constraints}
\end{align}
Here, we assume that the subsystems have learned policies \(\policy_{\controller_1}, \policy_{\controller_2}, ..., \policy_{\controller_{\numControllers}}\). Let \(\lbList = \{\controllerPerformanceEstimate_{\controller_1}, \controllerPerformanceEstimate_{\controller_2}, ..., \controllerPerformanceEstimate_{\controller_{\numControllers}}\}\) be the set of the corresponding estimated performance values.
The objective function (\ref{eq:hlm_opt_modified_obj}) minimizes the performance gap between the subtask specifications \(\bernoulliProb_{\controller}\) and the current estimated performance values \(\controllerPerformanceEstimate_{\controller}\).
The rationale behind the additional constraints defined by \(LBConst(\lbList)\) is as follows: the subsystems have already learned policies achieving probabilities of subtask success greater than the estimated performance values \(\controllerPerformanceEstimate_{\controller}\), and so there is no reason to consider subtask specifications \(\bernoulliProb_{\controller}\) that are below these values.

\input{algorithms/icrlV2}

Conversely, if the RL algorithm of a particular subsystem \(\controller\) has \textit{converged} -- i.e. the value of \(\controllerPerformanceEstimate_{\controller}\) will no longer increase with additional training steps -- we add the constraint \(\bernoulliProb_{\controller} \leq \controllerPerformanceEstimate_{\controller}\). 
This ensures that solutions to the optimization problem will \textit{not} yield a subtask specification \(\bernoulliProb_{\controller}\) that is larger than what the subsystem can realistically achieve. 
In practice, as a proxy to convergence, we allow each subsystem a maximum budget of \(\maxTrainingSteps\) training steps. 
Once any subsystem \(\controller\) has exceeded this training budget, we append \(\controllerPerformanceEstimate_{\controller}\) to the set \(\ubList\), which is used to define \(UBConst(\ubList)\) in (\ref{eq:hlm_opt_ub_constraints}).

\subsection{Iterative Compositional Reinforcement Learning (ICRL)}
By alternating between the training of the subsystems and the refinement of the subtask specifications, we obtain Algorithm \ref{alg:ICRL}.
In lines \(1\) to \(3\), the HLM is constructed from the collection of partially instantiated subsystems \(\controllerSet\) and the subsystem policies are initialized.
The while loop in lines \(4\) to \(12\) is the main loop controlling the subtask specifications and training of the subsystems.
In line \(5\), the bilinear program \eqref{eq:hlm_opt_objective}-\eqref{eq:hlm_opt_ub_constraints} is solved to update the values of \(\bernoulliProb_{\controller}\). 
These values are used, along with the estimated performance values, to select a subsystem to train.
A simple selection scheme is to choose the subsystem \(\controller_j\) maximizing the current performance gap between \(\bernoulliProb_{\controller_j}\) and \(\controllerPerformanceEstimate_{\controller_j}\).
In line \(7\), the subsystem is trained for \(\trainingSteps\) steps using the RL algorithm of choice.
The subsystem's initial state is sampled uniformly from its entry conditions during training.
Finally, in line \(12\), the HLM \(\abstractMDP\) and the current estimated performance values \(\lbList\) are used to plan a meta-policy \(\abstractPolicy\) maximizing the probability \(\controllerPerformanceEstimate_{\abstractPolicy}\) of reaching the HLM goal state \(\abstractSuccessState\).
This step uses standard MDP algorithms \cite{puterman2014markov}.

We note that the conditions in lines \(4\) and \(5\) ensure that the algorithm only terminates once a meta-policy that satisfies the task specification exists, or the optimization problem~\eqref{eq:hlm_opt_objective}-\eqref{eq:hlm_opt_ub_constraints} has become infeasible.
One of these two outcomes is guarateed to eventually occur.
In particular, by our construction of \(\ubList\) and the corresponding constraints in \eqref{eq:hlm_opt_ub_constraints}, the problem will become infeasible if all of the allotted subsystem training budgets \(\numSteps_{max}\) have been exhausted and a satisfactory meta-policy still does not exist.
In such circumstances the task designer may wish to lower \(\delta\), to increase \(N_{max}\), or to further decompose the task using additional subtasks.

%% file: algorithms/icrlV2.tex
\begin{algorithm}[t]
    \DontPrintSemicolon 
    \KwIn{ Partially instantiated subsystems \(\controllerSet = \{\controller_1, \controller_2, ..., \controller_{\numControllers}\}\), \(\hlmFailProb\), \(\trainingSteps\), \(\maxTrainingSteps\).}
    \KwOut{Subsystem policies \(\{\policy_{\controller_1}, ..., \policy_{\controller_{\numControllers}}\}\), meta-policy \(\abstractPolicy\), success probability \(\controllerPerformanceEstimate_{\abstractPolicy}\).}
    
    \(\abstractMDP \gets ConstructHLM(\controllerSet)\)\;
    
    

    \(\controllerPerformanceEstimate_{\controller_1}, \controllerPerformanceEstimate_{\controller_2}, ..., \controllerPerformanceEstimate_{\controller_{\numControllers}}, \controllerPerformanceEstimate_{\abstractPolicy} \gets 0\)\;
    \(\numSteps_{\controller_1}, \numSteps_{\controller_2}, ..., \numSteps_{\controller_{\numControllers}} \gets 0\)\;
    
    \(\lbList \gets \{\controllerPerformanceEstimate_{\controller_1}, \controllerPerformanceEstimate_{\controller_2}, ..., \controllerPerformanceEstimate_{\controller_{\numControllers}}\}\)\; 
    \(\ubList \gets \{\}\)\;
    
    \While{\(\controllerPerformanceEstimate_{\abstractPolicy} \leq 1 - \delta\)}{
    
        \If{~\eqref{eq:hlm_opt_objective}-\eqref{eq:hlm_opt_ub_constraints} {\normalfont infeasible}}{
            \Return{{\normalfont Problem is infeasible.}}
        }
        
        \(\{\bernoulliProb_{\controller_1}, \ldots, \bernoulliProb_{\controller_{\numControllers}}\} \gets\)
        \text{Solve~\eqref{eq:hlm_opt_objective}-\eqref{eq:hlm_opt_ub_constraints} using} \((\abstractMDP, \lbList, \ubList)\)\;
        
        \(\controller_{j} \gets selectSubSystem(\bernoulliProb_{\controller_1}, \ldots, \bernoulliProb_{\controller_{\numControllers}}, \controllerPerformanceEstimate_{\controller_1},\ldots, \controllerPerformanceEstimate_{\controller_{\numControllers}})\)\;
    
        \(\policy_{\controller_{j}} \gets RLTrain(\controller_j, \policy_{\controller_j}, \trainingSteps)\)\;
        \(\numSteps_{\controller_j} \gets \numSteps_{\controller_j} + \trainingSteps\)\;
        
        
        \(\controllerPerformanceEstimate_{\controller_j} \gets estimateSubTaskSuccessProb(\controller_j, \policy_{\controller_j})\)\;
        
        \(\lbList.update(\controllerPerformanceEstimate_{\controller_j})\)\;
        
        \If{\(\numSteps_{\controller_j} \geq \maxTrainingSteps\)}{\(\ubList.add(\controllerPerformanceEstimate_{\controller_j})\)\;}
        
        \(\abstractPolicy \gets solveOptimalHLMPolicy(\abstractMDP, \lbList)\)\;
        
        \(\controllerPerformanceEstimate_{\abstractPolicy} \gets predictTaskSuccessProbability(\abstractMDP, \abstractPolicy, \lbList)\)\;
    }
    
    \Return{\(\{\policy_{\controller_1}, \policy_{\controller_2}, ..., \policy_{\controller_{\numControllers}}\}\), \(\abstractPolicy\), \(\controllerPerformanceEstimate_{\abstractPolicy}\)}\;
    
    \caption{Iterative Compositional RL (ICRL)}
    \label{alg:ICRL}
\end{algorithm}

%% file: tex/06_experiments.tex
\section{Numerical Case Studies}

\label{sec:experiments}

In this section, we present a suite of numerical case studies that apply the proposed compositional RL approach to a variety of learning tasks.
These case studies highlight the framework's goal-directed decomposition of task specifications, its iterative training of the subsystem policies and refinement of the corresponding subtask specifications, and its ability to aid in task transfer.
The tasks environments include both discrete and continuous state and action spaces, deterministic and stochastic dynamics, and full and partial observability.

We begin in \S \ref{sec:experiment_discrete_labyrinth} by presenting numerical results for a discrete gridworld implementation of the labyrinth task from Figure \ref{fig:running_example}.
To highlight the framework's ability to handle environments with partial observations, in \S \ref{sec:experiments_partial_observability} we present the results of training the compositional RL system when the agent's observations are limited to egocentric image observations of its immediate surroundings.
To help demonstrate the framework's scalability with respect to the complexity of the underlying system dynamics, in \S\ref{sec:experiments_continuous_labyrinth} we present results for a continuous-state and continuous-action version of the labyrinth, whose dynamics are goverened by a rigid-body physics simulator.
Finally, in \S \ref{sec:experiments_modular_task_transfer} we present an example in which subsystem policies learned on one task are transferred to be used as components of a system that is designed to complete entirely new tasks.

All experiments were run locally on a laptop computer with an Intel i9-11900H 2.5 GHz CPU, an NVIDIA RTX 3060 GPU, and with 16 GB of RAM. 
Project code is available at: https://github.com/cyrusneary/verifiable-compositional-rl.

\input{tex/06_1_discrete_labyrinth}

\input{tex/06_2_partially_observable_labyrinth}

\input{tex/06_3_continuous_labyrinth}

\input{tex/06_4_modular_task_transfer}

%% file: tex/06_1_discrete_labyrinth.tex
\begin{figure*}[t!]
    \centering
    \begin{subfigure}[t]{0.45\textwidth}
        \centering 
        \input{figures/minigrid_labyrinth}
        \label{fig:supp_discrete_labyrinth_environment}
    \end{subfigure}%
    ~
    \hspace*{\fill}
    \begin{subfigure}[t]{0.45\textwidth}
        \centering 
        \includegraphics[height=4.7cm]{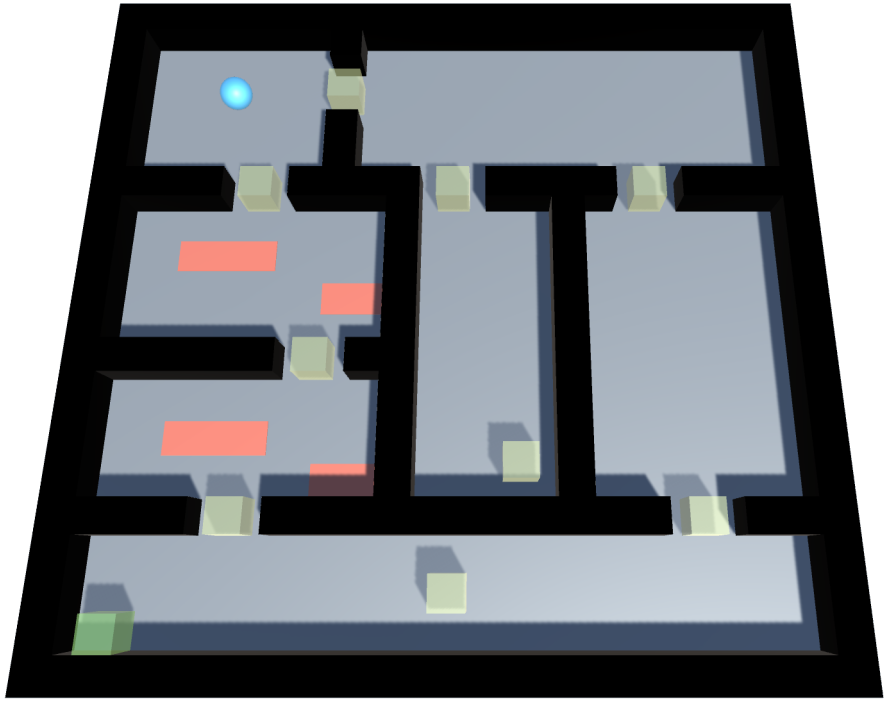}
        \label{fig:supp_continuous_labyrinth_environment}
    \end{subfigure}
    \caption{Left: A discrete version of the labyrinth environment. The egocentric image observations that are available to the agent in the partial observation experiments of \S \ref{sec:experiments_partial_observability} are illustrated to the left of the labyrinth.
    Right: A continuous version of the labyrinth environment. The environment dynamics are simulated using a rigid-body physics engine and the action space is the collection of all force vectors, with magnitude of at most 1, that can be applied to the ball in the horizontal plane.} 
    \label{fig:experiments_labyrinth_environments}
\end{figure*}

\subsection{Experiments on a Discrete Gridworld Labyrinth Environment}
\label{sec:experiment_discrete_labyrinth}
Recall that Figure~\ref{fig:running_example} illustrates the labyrinth environment that we use as a running example throughout this manuscript. 
Each subtask is highlighted with a different color, matching the colors that we use to represent the different subtasks in the following presentation of the numerical results.
Recall that the overall task specification is to safely navigate from the labyrinth's initial state in the top left corner to the goal state marked by a green square in the bottom left corner, while avoiding unsafe states, with a probability of at least \(0.95\). 

\subsubsection{The Discrete Labyrinth Environment}
We implement a discrete version of the gridworld labyrinth environment using MiniGrid \shortcite{gym_minigrid}, as illustrated on the left of Figure \ref{fig:experiments_labyrinth_environments}.
The environment's state space consists of the agent's current position and orientation within the labyrinth, resulting in 1600 total states.
We begin by considering a problem domain with full observability ---at every timestep the agent is observes its current position and orientation. 
There are three available actions in this environment: \textit{turn left}, \textit{turn right}, and \textit{move forward}.
We introduce stochasticity into the environment dynamics by adding a \textit{slip probability}---each action has a \(10\%\) chance of failing and instead causing the result of one of the other actions to occur.

\subsubsection{Implementation of the Compositional Reinforcement Learning System}
\label{sec:experimental_implementation_icrl_algorithm}

To implement Algorithm \ref{alg:ICRL}, we begin by constructing
the HLM from the subtask entry and exit conditions, which are implemented as finite collections of environment states in the discrete labyrinth environment.

During each loop of the algorithm, in order to decompose the task specification and automatically select the next subsystem to train, we solve the bilinear program in~\eqref{eq:hlm_opt_objective}--\eqref{eq:hlm_opt_task_sat_constraints} using \textit{Gurobi}~\cite{gurobi}.
Gurobi transforms the bilinear program into an equivalent mixed-integer linear program, and computes a globally optimal solution to this program by using cutting plane and branch and bound methods.

Each RL subsystem is trained to reach its exit conditions, given that its initial state is sampled uniformly form one of its entry conditions. 
The reward functions used to train the subsystems return \(1.0\) when one of the corresponding exit state(s) have been reached and 0.0 otherwise.
If the subsystem collides with any of the unsafe lava states, then the episode ends and the subsystem has no chance to receive a positive reward.
Alternatively, the episode will end once an exit condition is reached or the episode's time horizon is exhausted.

The subsystem policies are represented by neural networks with two fully-connected hidden layers, each of which has 64 units with hyperbolic tangent activation functions.
The network outputs a categorical distribution over the available actions.
Whenver a particular subsystem is selected by the HLM for training, we train it for \(\trainingSteps = 50,000\) steps using the proximal policy optimization (PPO) algorithm \shortcite{schulman2017proximal} with a discount factor of \(\mdpDiscount = 0.99\).
We use the Stable-Baselines3 implementation of the algorithm \shortcite{stable-baselines3}.
The critic network used by the training algorithm shares an identical structure to the network used to represent the policy, with the exception of its final layer which outputs an estimate of the value function.
We set the training algorithm's other parameters to the recommended default values \shortcite{stable-baselines3}, which are listed in Appendix \ref{sec:appendix_ppo_parameters}.

We note that while we use the PPO algorithm for training, any algorithm could in principle be used to obtain the subsystem policies.
The specifics of the subsystem training procedure do not affect the operation of the ICRL algorithm; the subsystems should be trained using the algorithms that are best suited to their corresponding subtasks.

After training each subsystem policy, we estimate its probability of subtask success \(\controllerPerformanceEstimate_{\controller}\) by rolling out the learned policy \(300\) times from initial states that are sampled uniformly from its entry conditions \(\controllerInitialStateSet_{\controller}\).
Each subsystem is given a maximum allowable training budget of \(N_{max}=500,000\) training steps before its most recent estimated performance value \(\controllerPerformanceEstimate_{\controller}\) is added as an upper bound constraint on the corresponding parameter value \(\bernoulliProb_{\controller}\) in the bilinear program \eqref{eq:hlm_opt_objective}-\eqref{eq:hlm_opt_ub_constraints} (as described in lines 10-11 in Algorithm \ref{alg:ICRL}).

A complete training run for the discrete gridworld labyrinth, which consists of roughly 1,500,000 total training iterations across all subsystems, takes approximately 25 minutes of wall-clock time.

\begin{figure}[t]%
    \centering
    \begin{subfigure}[t]{0.48\textwidth}
        \centering\input{figures/training_curves}
        \label{fig:experiments:b}
    \end{subfigure}
    \hspace*{\fill}
    \centering
    \begin{subfigure}[t]{0.48\textwidth}
        \centering\input{figures/training_schedule}
    \label{fig:experiments:c}
    \end{subfigure}
    \caption{Discrete labyrinth experimental results.
    Left: Estimated task and subtask success probabilities during training.
    The blue and black curves show the HLM-predicted and the empirically observed task success probabilities for the entire compositional system.
    The smaller colored curves correspond to the subtask success probabilities achieved by the subsystems.
    Right: Automatically generated subsystem training schedule.
    Each subtask is represented by a different color, matching those used in Figure \ref{fig:running_example}.
    The dotted red lines illustrate the point in training at which the HLM automatically refines the subtask specifications.
    Step counts do not include the rollouts used to estimate subtask success probabilities.
    }
    \label{fig:experiments}
\end{figure}
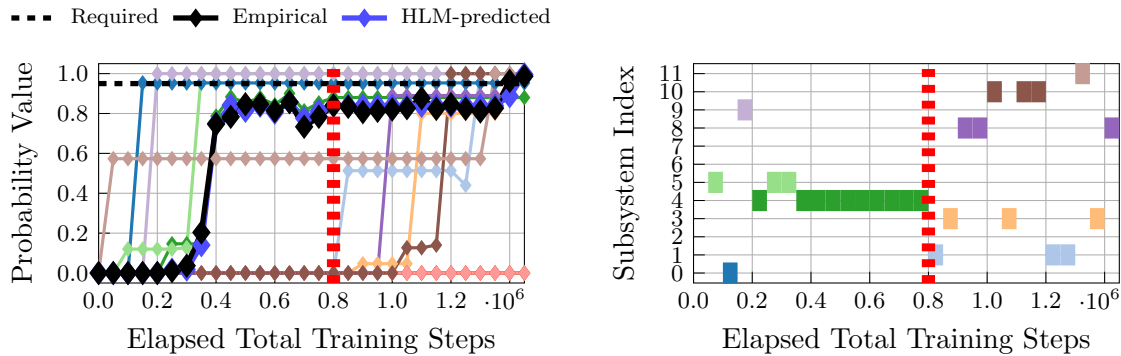

\subsubsection{Empirical Validation of Theorem \ref{thm:hlm_bounds_true_performance}}
At regular intervals during training, marked by diamonds in Figure \ref{fig:experiments}, each subsystem's probability of subtask success is estimated and used to update \(\lbList\) and \(\ubList\), as described in \S \ref{sec:icrl}. That is, each diamond in Figure \ref{fig:experiments} corresponds to a pass through the main loop of algorithm \ref{alg:ICRL}.
The HLM-predicted probability of the meta-policy completing the overall task is illustrated in Figure \ref{fig:experiments} by the navy blue curve.
For comparison, we plot empirical measurements of the success rate of the compositional system in black.
We clearly observe that the HLM predictions closely match the empirical measurements, indicating that they are suitably accurate for the planning and verification of meta-policies. 

\subsubsection{Subtask Specification Refinements Lead to Meta-Policy Adaptation and Targeted Subsystem Training}
\label{sec:experiments_discrete_adaptation}
The plot on the right of Figure \ref{fig:experiments} illustrates the subsystem training schedule. Table \ref{tab:sub_task_specifications} lists the values of \(\bernoulliProb_{\controller}\) for each subsystem \(\controller\).
We observe from Table \ref{tab:sub_task_specifications} that prior to \(8e5\) elapsed training steps, the value of \(\bernoulliProb_{\controller}\) is only specified to be close to \(1.0\) for subsystems \(\controller_0\), \(\controller_4\), \(\controller_5\), and, \(\controller_9\). 
As can be seen in Figure \ref{fig:running_example}, these are the subsystems needed to move straight down, through the rooms containing lava, to the goal. 
The HLM has selected a meta-policy that will only use these subsystems because their composition yields the shortest path to goal; this path only requires training of 4 of the subsystems.
Furthermore, because the meta-policy does not use any of the other subsystems, it places no requirements on their probability of subtask success.

The training schedule in Figure \ref{fig:experiments} agrees with this observation: only this small collection of the subsystems are trained prior to \(800,000\) elapsed training steps.
In particular, subsystem \(4\), which must navigate the top lava room and is represented by dark green, is trained extensively.
However, due to the environment slip probability, this subsystem is unable to meet its subtask specification, \textit{safely navigate to the room's exit with probability 0.97}, regardless of the number of training iterations it receives.
As a result, subsystem \(4\) exhausts its individual training budget after \(800,000\) elapsed system training steps, marked by the vertical dotted red lines in Figure \ref{fig:experiments}.
At this point, subsystem \(4\)'s empirically estimated success rate of 0.88 is used to update the HLM, which then refines the subtask specifications.

The result of this refinement is a new meta-policy, which instead uses subsystems \(\controller_1\), \(\controller_3\), \(\controller_8\), \(\controller_{10}\), and \(\controller_{11}\) to take an alternate path that avoids the lava rooms altogether.
The updated subtask specifications are listed in the second row of Table \ref{tab:sub_task_specifications}, and in Figure \ref{fig:experiments} we observe a distinct change in the subsystems that are trained.
Once subsystems \(\controller_1\), \(\controller_3\), \(\controller_8\), \(\controller_{10}\), and \(\controller_{11}\) learn to satisfy their new subtask specifications with the required probability, the composite system's probability of task success rises above \(0.95\), satisfying the overall task specification.

\input{tables/sub_task_requirements}

\subsubsection{Comparison with a Monolithic Approach to Reinforcement Learning}
As a point of comparison, we also trained a monolithic RL agent---which treats the entire task as a single subsystem---to reach the target set \(\mathcal{F}_{targ}\).
This agent was trained using the same PPO algorithm parameters and neural network architecture as described in \S \ref{sec:experimental_implementation_icrl_algorithm}.
The proposed ICRL algorithm takes less than two million training steps to learn to satisfy the task specification. 
Meanwhile, the monolithic RL agent takes roughly thirty million training steps.
We note that this is not a fair comparison because the proposed compositional approach is provided with a priori knowledge of the subsystem entry and exit conditions.
However, such information is often available through natural decompositions of complex systems.
The proposed framework provides a method to take advantage of such information when it is available.

Furthermore, we remark that the primary focus of our work is on the verification of the compositional RL system against the task specification, and on the high-level decision-making capabilities that are provided by the HLM formulation.
These novel decision-making capabilities, such as the automated meta-policy adaptation and the targeted subsystem training that are described in \S \ref{sec:experiments_discrete_adaptation}, are not possible using monolithic approaches to RL or using existing hierarchical RL methods.

%% file: figures/minigrid_labyrinth.tex
\begin{tikzpicture}[
    labelText/.style={font=\scriptsize, inner xsep=0.0cm, inner ysep=0.0cm},
    eqText/.style={font=\scriptsize, inner xsep=0.0cm, inner ysep=0.0cm},
]
    \node (labyrinth_picture) [labelText,]  {\includegraphics[width=4.7cm]{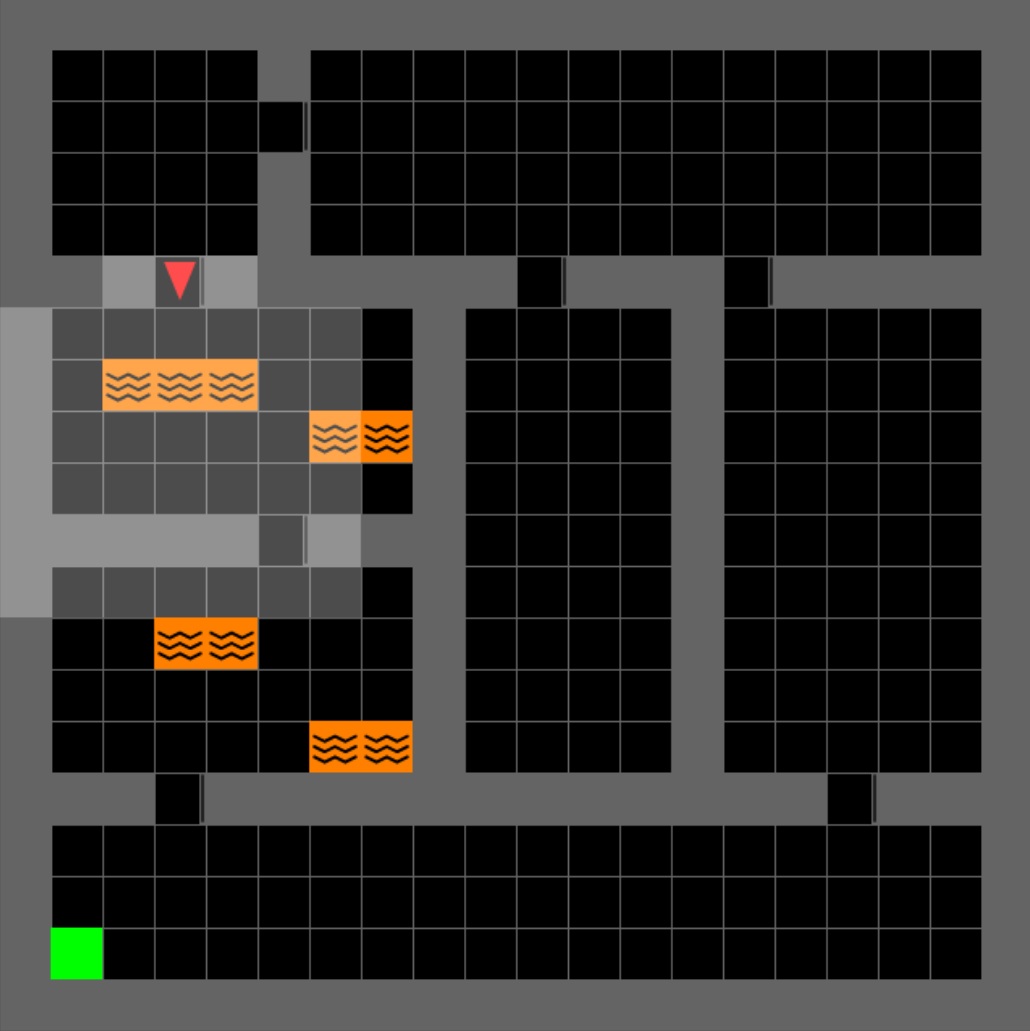}};
    
    \node (egocentric_view) [labelText, left=5mm of labyrinth_picture.west, yshift=1.2cm] {\includegraphics[width=1.8cm]{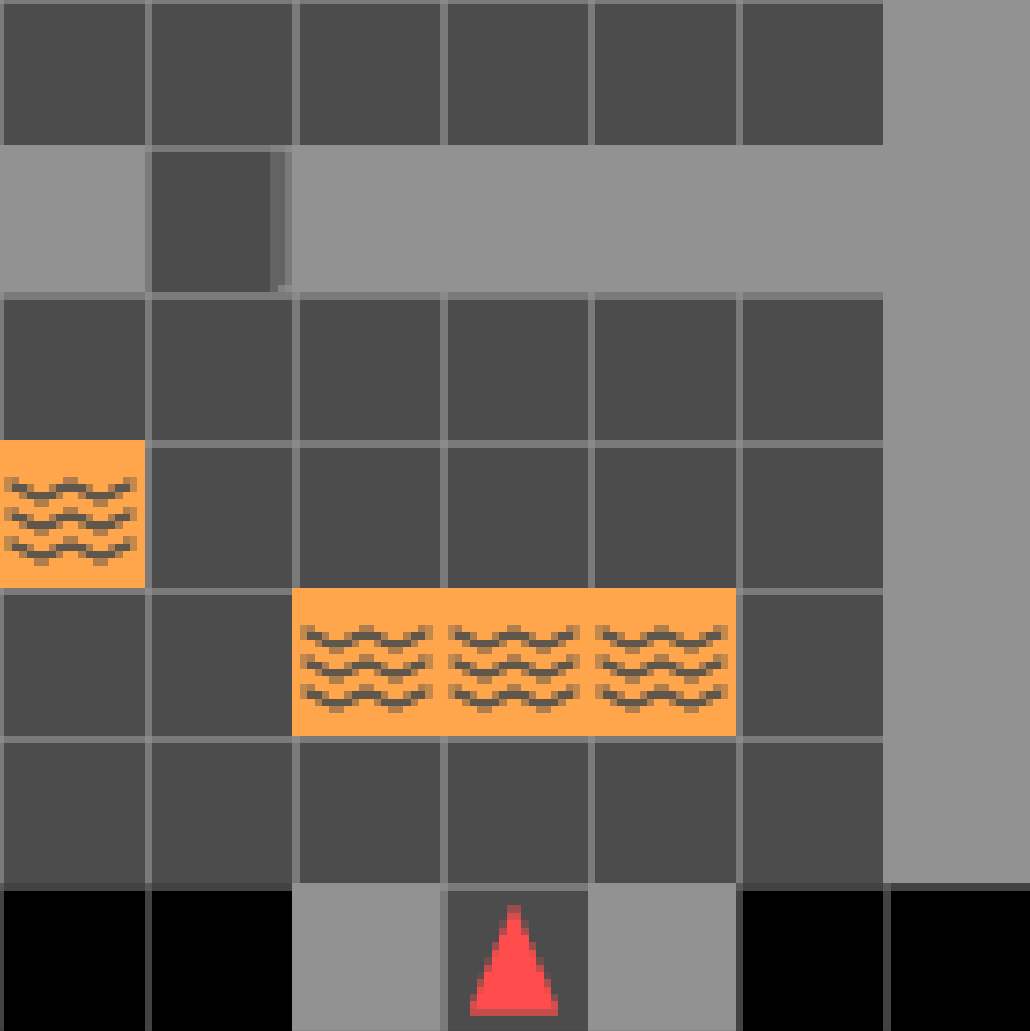}};
    
    \node (egocentric_label) [font=\scriptsize, below=1mm of egocentric_view.south, align=center] {Egocentric \\ image \\ observation};
    
    \node (egocentric_box) [draw, inner xsep=2mm, inner ysep=2mm, rounded corners=0.15cm, fit={(egocentric_label) (egocentric_view)}] {};
    
\end{tikzpicture}

%% file: figures/training_curves.tex
\begin{tikzpicture}

\definecolor{color0}{rgb}{0.12156862745098,0.466666666666667,0.705882352941177}
\definecolor{color1}{rgb}{0.682352941176471,0.780392156862745,0.909803921568627}
\definecolor{color2}{rgb}{1,0.498039215686275,0.0549019607843137}
\definecolor{color3}{rgb}{1,0.733333333333333,0.470588235294118}
\definecolor{color4}{rgb}{0.172549019607843,0.627450980392157,0.172549019607843}
\definecolor{color5}{rgb}{0.596078431372549,0.874509803921569,0.541176470588235}
\definecolor{color6}{rgb}{0.83921568627451,0.152941176470588,0.156862745098039}
\definecolor{color7}{rgb}{1,0.596078431372549,0.588235294117647}
\definecolor{color8}{rgb}{0.580392156862745,0.403921568627451,0.741176470588235}
\definecolor{color9}{rgb}{0.772549019607843,0.690196078431373,0.835294117647059}
\definecolor{color10}{rgb}{0.549019607843137,0.337254901960784,0.294117647058824}
\definecolor{color11}{rgb}{0.768627450980392,0.611764705882353,0.580392156862745}

\begin{customlegend}[legend columns=3, legend style={align=center, column sep=0.5ex, font=\scriptsize, draw=none, at={(62.0mm, 38.0mm)}}, legend entries={Required, Empirical, HLM-predicted}]
\addlegendimage{line width=2pt, black, dashed}
\addlegendimage{line width=2pt, black, mark=diamond*, mark size=2}
\addlegendimage{line width=2pt, blue!70!white, mark=diamond*, mark size=2}
\end{customlegend}

\begin{axis}[
tick align=inside,
tick pos=left,
x grid style={white!69.0196078431373!black},
  ticklabel style = {font=\footnotesize},
xmajorgrids,
xmin=0, xmax=1450000,
xtick style={color=black},
y grid style={white!69.0196078431373!black},
ymajorgrids,
ymin=-0.05, ymax=1.05,
height=4.5cm,
width=0.99\textwidth,
ytick={0.0, 0.2, 0.4, 0.6, 0.8, 1.0},
yticklabels={0.0, 0.2, 0.4, 0.6, 0.8, 1.0},
xtick={0, 200000, 400000, 600000, 800000, 1000000, 1200000, 1400000},
xticklabels={0.0, 0.2, 0.4, 0.6, 0.8, 1.0, 1.2, },
ytick style={color=black},
legend to name=named,
every x tick scale label/.style={
    at={(1,0)},xshift=-15.5pt,yshift=-10.0pt,anchor=south west,inner sep=0pt},
legend style={
    anchor=north,
    fill opacity=1, 
    draw opacity=1,
    legend columns=2, 
    text opacity=1, 
    draw=white!0.0!black,
    scale=0.40,font=\scriptsize\selectfont
    },
ylabel={Probability Value},
xlabel={Elapsed Total Training Steps}
]
\addplot [very thick, color0, mark=diamond*, mark size=2]
table {%
0 0
50000 0
100000 0
150000 0.953333333333333
200000 0.953333333333333
250000 0.953333333333333
300000 0.953333333333333
350000 0.953333333333333
400000 0.953333333333333
450000 0.953333333333333
500000 0.953333333333333
550000 0.953333333333333
600000 0.953333333333333
650000 0.953333333333333
700000 0.953333333333333
750000 0.953333333333333
800000 0.953333333333333
850000 0.953333333333333
900000 0.953333333333333
950000 0.953333333333333
1000000 0.953333333333333
1050000 0.953333333333333
1100000 0.953333333333333
1150000 0.953333333333333
1200000 0.953333333333333
1250000 0.953333333333333
1300000 0.953333333333333
1350000 0.953333333333333
1400000 0.953333333333333
1450000 0.953333333333333
};
\addlegendentry{1}
\addplot [very thick, color1, mark=diamond*, mark size=2]
table {%
0 0
50000 0
100000 0
150000 0
200000 0
250000 0
300000 0
350000 0
400000 0
450000 0
500000 0
550000 0
600000 0
650000 0
700000 0
750000 0
800000 0
850000 0.513333333333333
900000 0.513333333333333
950000 0.513333333333333
1000000 0.513333333333333
1050000 0.513333333333333
1100000 0.513333333333333
1150000 0.513333333333333
1200000 0.513333333333333
1250000 0.44
1300000 1
1350000 1
1400000 1
1450000 1
};
\addplot [very thick, color2, mark=diamond*, mark size=2]
table {%
0 0
50000 0
100000 0
150000 0
200000 0
250000 0
300000 0
350000 0
400000 0
450000 0
500000 0
550000 0
600000 0
650000 0
700000 0
750000 0
800000 0
850000 0
900000 0
950000 0
1000000 0
1050000 0
1100000 0
1150000 0
1200000 0
1250000 0
1300000 0
1350000 0
1400000 0
1450000 0
};
\addplot [very thick, color3, mark=diamond*, mark size=2]
table {%
0 0
50000 0
100000 0
150000 0
200000 0
250000 0
300000 0
350000 0
400000 0
450000 0
500000 0
550000 0
600000 0
650000 0
700000 0
750000 0
800000 0
850000 0
900000 0.0466666666666667
950000 0.0466666666666667
1000000 0.0466666666666667
1050000 0.0466666666666667
1100000 0.8
1150000 0.8
1200000 0.8
1250000 0.8
1300000 0.8
1350000 0.8
1400000 0.996666666666667
1450000 0.996666666666667
};
\addplot [very thick, color4, mark=diamond*, mark size=2]
table {%
0 0
50000 0
100000 0
150000 0
200000 0
250000 0.146666666666667
300000 0.146666666666667
350000 0.146666666666667
400000 0.786666666666667
450000 0.883333333333333
500000 0.853333333333333
550000 0.88
600000 0.843333333333333
650000 0.9
700000 0.813333333333333
750000 0.853333333333333
800000 0.88
850000 0.88
900000 0.88
950000 0.88
1000000 0.88
1050000 0.88
1100000 0.88
1150000 0.88
1200000 0.88
1250000 0.88
1300000 0.88
1350000 0.88
1400000 0.88
1450000 0.88
};
\addplot [very thick, color5, mark=diamond*, mark size=2]
table {%
0 0
50000 0
100000 0.12
150000 0.12
200000 0.12
250000 0.12
300000 0.126666666666667
350000 1
400000 1
450000 1
500000 1
550000 1
600000 1
650000 1
700000 1
750000 1
800000 1
850000 1
900000 1
950000 1
1000000 1
1050000 1
1100000 1
1150000 1
1200000 1
1250000 1
1300000 1
1350000 1
1400000 1
1450000 1
};
\addplot [very thick, color6, mark=diamond*, mark size=2]
table {%
0 0
50000 0
100000 0
150000 0
200000 0
250000 0
300000 0
350000 0
400000 0
450000 0
500000 0
550000 0
600000 0
650000 0
700000 0
750000 0
800000 0
850000 0
900000 0
950000 0
1000000 0
1050000 0
1100000 0
1150000 0
1200000 0
1250000 0
1300000 0
1350000 0
1400000 0
1450000 0
};
\addplot [very thick, color7, mark=diamond*, mark size=2]
table {%
0 0
50000 0
100000 0
150000 0
200000 0
250000 0
300000 0
350000 0
400000 0
450000 0
500000 0
550000 0
600000 0
650000 0
700000 0
750000 0
800000 0
850000 0
900000 0
950000 0
1000000 0
1050000 0
1100000 0
1150000 0
1200000 0
1250000 0
1300000 0
1350000 0
1400000 0
1450000 0
};
\addplot [very thick, color8, mark=diamond*, mark size=2]
table {%
0 0
50000 0
100000 0
150000 0
200000 0
250000 0
300000 0
350000 0
400000 0
450000 0
500000 0
550000 0
600000 0
650000 0
700000 0
750000 0
800000 0
850000 0
900000 0
950000 0
1000000 0.89
1050000 0.89
1100000 0.89
1150000 0.89
1200000 0.89
1250000 0.89
1300000 0.89
1350000 0.89
1400000 0.89
1450000 1
};
\addplot [very thick, color9, mark=diamond*, mark size=2]
table {%
0 0
50000 0
100000 0
150000 0
200000 1
250000 1
300000 1
350000 1
400000 1
450000 1
500000 1
550000 1
600000 1
650000 1
700000 1
750000 1
800000 1
850000 1
900000 1
950000 1
1000000 1
1050000 1
1100000 1
1150000 1
1200000 1
1250000 1
1300000 1
1350000 1
1400000 1
1450000 1
};
\addplot [very thick, color10, mark=diamond*, mark size=2]
table {%
0 0
50000 0
100000 0
150000 0
200000 0
250000 0
300000 0
350000 0
400000 0
450000 0
500000 0
550000 0
600000 0
650000 0
700000 0
750000 0
800000 0
850000 0
900000 0
950000 0
1000000 0
1050000 0.126666666666667
1100000 0.126666666666667
1150000 0.14
1200000 1
1250000 1
1300000 1
1350000 1
1400000 1
1450000 1
};
\addplot [very thick, color11, mark=diamond*, mark size=2]
table {%
0 0
50000 0.573333333333333
100000 0.573333333333333
150000 0.573333333333333
200000 0.573333333333333
250000 0.573333333333333
300000 0.573333333333333
350000 0.573333333333333
400000 0.573333333333333
450000 0.573333333333333
500000 0.573333333333333
550000 0.573333333333333
600000 0.573333333333333
650000 0.573333333333333
700000 0.573333333333333
750000 0.573333333333333
800000 0.573333333333333
850000 0.573333333333333
900000 0.573333333333333
950000 0.573333333333333
1000000 0.573333333333333
1050000 0.573333333333333
1100000 0.573333333333333
1150000 0.573333333333333
1200000 0.573333333333333
1250000 0.573333333333333
1300000 0.573333333333333
1350000 1
1400000 1
1450000 1
};
\addplot [line width=2pt, black, dashed ,on layer=foreground]
table {%
0 0.95
1450000 0.95
};
\addlegendentry{Required probability of task success}
\addplot [line width=2pt, blue!70!white, mark=diamond*, mark size=3]
table {%
0 0
50000 0
100000 0
150000 0
200000 0
250000 0.0167786666666667
300000 0.0177108148148148
350000 0.139822222222222
400000 0.749955555555555
450000 0.842111111111111
500000 0.813511111111111
550000 0.838933333333333
600000 0.803977777777778
650000 0.858
700000 0.775377777777778
750000 0.813511111111111
800000 0.838933333333333
850000 0.838933333333333
900000 0.838933333333333
950000 0.838933333333333
1000000 0.838933333333333
1050000 0.838933333333333
1100000 0.838933333333333
1150000 0.838933333333333
1200000 0.838933333333333
1250000 0.838933333333333
1300000 0.838933333333333
1350000 0.838933333333333
1400000 0.887033333333333
1450000 0.996666666666667
};
\addplot [line width=2pt, black, mark=diamond*, mark size=3]
table {%
0 0
50000 0
100000 0
150000 0
200000 0
250000 0.01
300000 0.0366666666666667
350000 0.203333333333333
400000 0.746666666666667
450000 0.783333333333333
500000 0.846666666666667
550000 0.846666666666667
600000 0.813333333333333
650000 0.856666666666667
700000 0.733333333333333
750000 0.783333333333333
800000 0.84
850000 0.833333333333333
900000 0.806666666666667
950000 0.813333333333333
1000000 0.82
1050000 0.826666666666667
1100000 0.876666666666667
1150000 0.823333333333333
1200000 0.846666666666667
1250000 0.823333333333333
1300000 0.81
1350000 0.826666666666667
1400000 0.963333333333333
1450000 0.986666666666667
};
\addlegendentry{Empirically Measured Probability of Task Success}
\addplot [line width=5pt, red, dashed  ,on layer=foreground]
table {%
800000 -0.05
800000 1.05
};
\end{axis}
\end{tikzpicture}


%% file: figures/training_schedule.tex
\begin{tikzpicture}

\definecolor{color0}{rgb}{0.596078431372549,0.874509803921569,0.541176470588235}
\definecolor{color1}{rgb}{0.12156862745098,0.466666666666667,0.705882352941177}
\definecolor{color2}{rgb}{0.772549019607843,0.690196078431373,0.835294117647059}
\definecolor{color3}{rgb}{0.172549019607843,0.627450980392157,0.172549019607843}
\definecolor{color4}{rgb}{0.682352941176471,0.780392156862745,0.909803921568627}
\definecolor{color5}{rgb}{1,0.733333333333333,0.470588235294118}
\definecolor{color6}{rgb}{0.580392156862745,0.403921568627451,0.741176470588235}
\definecolor{color7}{rgb}{0.549019607843137,0.337254901960784,0.294117647058824}
\definecolor{color8}{rgb}{0.768627450980392,0.611764705882353,0.580392156862745}

\begin{axis}[
height=4.5cm,
width=0.99\textwidth,
tick align=inside,
tick pos=left,
x grid style={white!69.0196078431373!black},
ticklabel style = {font=\footnotesize},
xmajorgrids,
xmin=0, xmax=1450000,
xtick style={color=black},
y grid style={white!69.0196078431373!black},
ymajorgrids,
ymin=-0.55, ymax=11.55,
ytick style={color=black},
ytick={0,1,2,3,4,5,6,7,8,9,10,11},
yticklabels = {0,1,2,3,4,5,6,7,8,9,10,11},
xtick={0, 200000, 400000, 600000, 800000, 1000000, 1200000, 1400000},
xticklabels={0.0, 0.2, 0.4, 0.6, 0.8, 1.0, 1.2, },
every x tick scale label/.style={
    at={(1,0)},xshift=-15.5pt,yshift=-10.0pt,anchor=south west,inner sep=0pt},
ylabel={Subsystem Index},
xlabel={Elapsed Total Training Steps}
]
\addplot [line width=8pt, color0]
table {%
50000 5
100000 5
};
\addplot [line width=8pt, color1]
table {%
100000 0
150000 0
};
\addplot [line width=8pt, color2]
table {%
150000 9
200000 9
};
\addplot [line width=8pt, color3]
table {%
200000 4
250000 4
};
\addplot [line width=8pt, color0]
table {%
250000 5
300000 5
};
\addplot [line width=8pt, color0]
table {%
300000 5
350000 5
};
\addplot [line width=8pt, color3]
table {%
350000 4
400000 4
};
\addplot [line width=8pt, color3]
table {%
400000 4
450000 4
};
\addplot [line width=8pt, color3]
table {%
450000 4
500000 4
};
\addplot [line width=8pt, color3]
table {%
500000 4
550000 4
};
\addplot [line width=8pt, color3]
table {%
550000 4
600000 4
};
\addplot [line width=8pt, color3]
table {%
600000 4
650000 4
};
\addplot [line width=8pt, color3]
table {%
650000 4
700000 4
};
\addplot [line width=8pt, color3]
table {%
700000 4
750000 4
};
\addplot [line width=8pt, color3]
table {%
750000 4
800000 4
};
\addplot [line width=8pt, color4]
table {%
800000 1
850000 1
};
\addplot [line width=8pt, color5]
table {%
850000 3
900000 3
};
\addplot [line width=8pt, color6]
table {%
900000 8
950000 8
};
\addplot [line width=8pt, color6]
table {%
950000 8
1000000 8
};
\addplot [line width=8pt, color7]
table {%
1000000 10
1050000 10
};
\addplot [line width=8pt, color5]
table {%
1050000 3
1100000 3
};
\addplot [line width=8pt, color7]
table {%
1100000 10
1150000 10
};
\addplot [line width=8pt, color7]
table {%
1150000 10
1200000 10
};
\addplot [line width=8pt, color4]
table {%
1200000 1
1250000 1
};
\addplot [line width=8pt, color4]
table {%
1250000 1
1300000 1
};
\addplot [line width=8pt, color8]
table {%
1300000 11
1350000 11
};
\addplot [line width=8pt, color5]
table {%
1350000 3
1400000 3
};
\addplot [line width=8pt, color6]
table {%
1400000 8
1450000 8
};
\addplot [line width=5pt, red, dashed]
table {%
800000 -0.55
800000 11.55
};
\end{axis}

\end{tikzpicture}

%% file: tables/sub_task_requirements.tex
\definecolor{color0}{rgb}{0.12156862745098,0.466666666666667,0.705882352941177}
\definecolor{color1}{rgb}{0.682352941176471,0.780392156862745,0.909803921568627}
\definecolor{color2}{rgb}{1,0.498039215686275,0.0549019607843137}
\definecolor{color3}{rgb}{1,0.733333333333333,0.470588235294118}
\definecolor{color4}{rgb}{0.172549019607843,0.627450980392157,0.172549019607843}
\definecolor{color5}{rgb}{0.596078431372549,0.874509803921569,0.541176470588235}
\definecolor{color6}{rgb}{0.83921568627451,0.152941176470588,0.156862745098039}
\definecolor{color7}{rgb}{1,0.596078431372549,0.588235294117647}
\definecolor{color8}{rgb}{0.580392156862745,0.403921568627451,0.741176470588235}
\definecolor{color9}{rgb}{0.772549019607843,0.690196078431373,0.835294117647059}
\definecolor{color10}{rgb}{0.549019607843137,0.337254901960784,0.294117647058824}
\definecolor{color11}{rgb}{0.768627450980392,0.611764705882353,0.580392156862745}

\definecolor{highlightgrey}{gray}{0.85}

\begin{table}[t]
\centering
 \begin{tabular}{|c | c | c | c | c | c | c | c | c | c | c | c | c |} 
 \hline
 Subsystem Index & 
 \cellcolor{color0}0 & 
 \cellcolor{color1}1 & 
 \cellcolor{color2}2 & 
 \cellcolor{color3}3 & 
 \cellcolor{color4}4 & 
 \cellcolor{color5}5 & 
 \cellcolor{color6}6 & 
 \cellcolor{color7}7 & 
 \cellcolor{color8}8 & 
 \cellcolor{color9}9 & 
 \cellcolor{color10}10 & 
 \cellcolor{color11}11 \\
 \hline\hline
 \(\bernoulliProb_{\controller}\) at \(t = 600,000\) 
 & \cellcolor{highlightgrey}.97 
 & .00 
 & .00 
 & .00 
 & \cellcolor{highlightgrey}.97 
 & \cellcolor{highlightgrey}1.0 
 & .00 
 & .00 
 & .00 
 & \cellcolor{highlightgrey}1.0 
 & .00 
 & .57\\ 
 \hline
 \(\bernoulliProb_{\controller}\) at \(t = 1 \text{ million}\) 
 & .95 
 & \cellcolor{highlightgrey}.99 
 & .00 
 & \cellcolor{highlightgrey}.99 
 & .88 
 & 1.0 
 & .00 
 & .00 
 & \cellcolor{highlightgrey}.99 
 & 1.0 
 & \cellcolor{highlightgrey}.99 
 & \cellcolor{highlightgrey}.99\\
 \hline
\end{tabular}
\caption{Demonstration of automatic subtask specification refinement. Each value corresponds to a subtask specification, i.e. the minimum allowable probability of subtask success. The two rows of the table show these values at two distinct points of the system's training; before and after the subtask specification refinement illustrated by the dotted red lines in Figure \ref{fig:experiments}. The cells highlighted in grey indicate which subsystems are used by the meta-policy, at the specified point in training.}
\label{tab:sub_task_specifications}
\end{table}

%% file: tex/06_2_partially_observable_labyrinth.tex
\subsection{Experiments in Environments with Partial Observability}
\label{sec:experiments_partial_observability}

We note that while the subtask entry and exit conditions are defined in terms of collections of environments states, the task abstraction provided by the HLM, as well as the high-level decision-making procedure proposed in Algorithm \ref{alg:ICRL}, are both agnostic to whether or not the agent is able to fully observe its own state.
In other words, the compositional RL system can operate effectively under partial observations, so long as its observations are sufficiently informative to train policies that complete the necessary subtasks.

\subsubsection{The Partially Observable Labyrinth Environment}
To numerically demonstrate the above claim, we apply the ICRL algorithm to a partially observable version of the task described in \S \ref{sec:experiment_discrete_labyrinth}. 
In this experiment, the agent may only make egocentric image-based observations of its surroundings.
These image observations are \(140\times140\) pixels with RGB color channels.
An example observation is illustrated on the left of Figure \ref{fig:experiments_labyrinth_environments}.
We assume that the subsystem controllers will terminate their execution once they have reached one of their exit conditions.

\subsubsection{Network Architecture Used to Represent the Subsystem Policies}
Our implementation of the ICRL algorithm remains identical to the fully observable setting described in \ref{sec:experimental_implementation_icrl_algorithm}, 
with the exception of the neural network architecture used to represent the subsystem policies and the estimated value functions, which are used by the PPO training algorithm.
To extract useful features from the observations, we use a convolutional neural network with the same architecture as the one used by \shortciteA{mnih2015human} to train RL agents to play Atari games.
The output of the feature extractor is a \(512\)-dimensional vector that is passed into two separate linear layers. The first layer outputs action probability values, while the second outputs an estimate of the value function. 

\subsubsection{Compositionality Enables Probabilistic Reasoning in Vision-to-Action Deep Learning Agents}

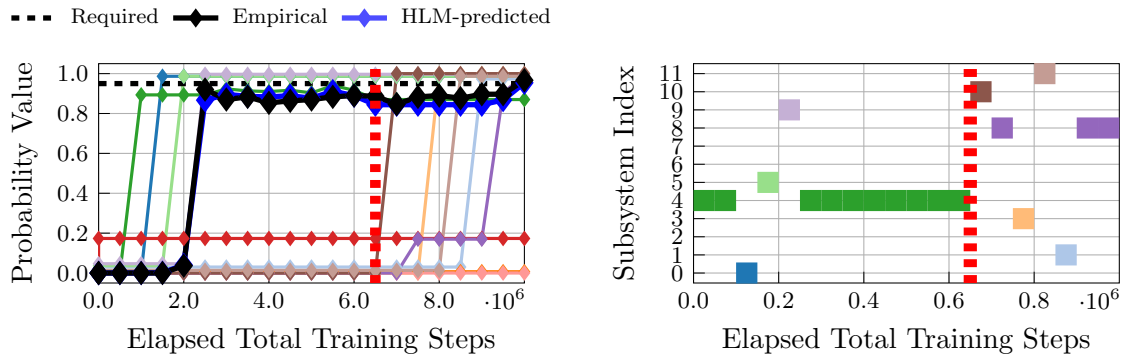
\begin{figure*}[t!]
    \centering
    \begin{subfigure}[t]{0.48\textwidth}
        \centering \input{figures/minigrid_pixel_training_curves}
        \label{fig:partially_observable_training_results}
    \end{subfigure}%
    \hspace*{\fill}
    \begin{subfigure}[t]{0.48\textwidth}
        \centering \input{figures/minigrid_pixel_training_schedule}
        \label{fig:partially_observable_training_schedule}
    \end{subfigure}
    \caption{
        Results for the partially observable labyrinth environment. 
        Left: Estimated task and subtask success probabilities during training.
        Right: Automatically generated subsystem training schedule.
    } 
    \label{fig:partially_observable_training_curves}
\end{figure*}

Figure \ref{fig:partially_observable_training_curves} illustrates the results of this experiment.
Qualitatively, the performance of the compositional RL system is identical to the results in the fully observable environment described in \S \ref{sec:experiment_discrete_labyrinth}.
ICRL initially attempts to train the subsystems \(\controller_0, \controller_4, \controller_5,\) and \(\controller_9\), which are used to move straight down through the rooms containing lava to reach the goal.
Subsystem \(\controller_{4}\) is able to learn a policy that achieves its objective with an estimated probability value of \(0.94\), resulting in a compositional system that the HLM predicts will complete the overall task with probability \(0.91\).
However, this is not sufficient to satisfy the required overall task success probability of \(0.95\).
As a result, similarly to as described in \S \ref{sec:experiments_discrete_adaptation}, the system makes a significant adjustment to its overall plan at the point during training marked by the dottend red line.
Subsystems \(\controller_1, \controller_3, \controller_8, \controller_{10}\), and \(\controller_{11}\) are trained to reach the goal instead.

We emphasize that the results demonstrate the ability of the compositional RL system to automatically reason about the probability of its own success, even while using deep learning vision-to-action policies to complete its subtasks.
With only partial observations that give it a local view of its environment, the system automatically replans during training because it accurately predicts that it will fall short of its required probability of task success by a mere four percentage points.

%% file: figures/minigrid_pixel_training_curves.tex
\begin{tikzpicture}

\definecolor{crimson2143940}{RGB}{214,39,40}
\definecolor{darkgray176}{RGB}{176,176,176}
\definecolor{darkorange25512714}{RGB}{255,127,14}
\definecolor{forestgreen4416044}{RGB}{44,160,44}
\definecolor{lightgreen152223138}{RGB}{152,223,138}
\definecolor{lightsalmon255152150}{RGB}{255,152,150}
\definecolor{lightsalmon255187120}{RGB}{255,187,120}
\definecolor{lightsteelblue174199232}{RGB}{174,199,232}
\definecolor{mediumpurple148103189}{RGB}{148,103,189}
\definecolor{rosybrown196156148}{RGB}{196,156,148}
\definecolor{sienna1408675}{RGB}{140,86,75}
\definecolor{steelblue31119180}{RGB}{31,119,180}
\definecolor{thistle197176213}{RGB}{197,176,213}

\begin{customlegend}[legend columns=3, legend style={align=center, column sep=0.5ex, font=\scriptsize, draw=none, at={(62.0mm, 38.0mm)}}, legend entries={Required, Empirical, HLM-predicted}]
\addlegendimage{line width=2pt, black, dashed}
\addlegendimage{line width=2pt, black, mark=diamond*, mark size=2}
\addlegendimage{line width=2pt, blue!70!white, mark=diamond*, mark size=2}
\end{customlegend}

\begin{axis}[
tick align=inside,
tick pos=left,
x grid style={white!69.0196078431373!black},
  ticklabel style = {font=\footnotesize},
xmajorgrids,
xmin=0, xmax=1000000,
xtick style={color=black},
y grid style={white!69.0196078431373!black},
ymajorgrids,
ymin=-0.05, ymax=1.05,
height=4.5cm,
width=0.99\textwidth,
ytick={0.0, 0.2, 0.4, 0.6, 0.8, 1.0},
yticklabels={0.0, 0.2, 0.4, 0.6, 0.8, 1.0},
xtick={0, 200000, 400000, 600000, 800000},
xticklabels={0.0, 2.0, 4.0, 6.0, 8.0 },
ytick style={color=black},
legend to name=named,
every x tick scale label/.style={
    at={(1,0)},xshift=-15.5pt,yshift=-10.0pt,anchor=south west,inner sep=0pt},
ylabel={Probability Value},
xlabel={Elapsed Total Training Steps}
]

\addplot [very thick, steelblue31119180, mark=diamond*, mark size=2, mark options={solid}]
table {%
0 0
0 0
50000 0
100000 0
150000 0.986666666666667
200000 0.986666666666667
250000 0.986666666666667
300000 0.986666666666667
350000 0.986666666666667
400000 0.986666666666667
450000 0.986666666666667
500000 0.986666666666667
550000 0.986666666666667
600000 0.986666666666667
650000 0.986666666666667
700000 0.986666666666667
750000 0.986666666666667
800000 0.986666666666667
850000 0.986666666666667
900000 0.986666666666667
950000 0.986666666666667
1000000 0.986666666666667
};
\addplot [very thick, lightsteelblue174199232, mark=diamond*, mark size=2, mark options={solid}]
table {%
0 0
0 0.03
50000 0.03
100000 0.03
150000 0.03
200000 0.03
250000 0.03
300000 0.03
350000 0.03
400000 0.03
450000 0.03
500000 0.03
550000 0.03
600000 0.03
650000 0.03
700000 0.03
750000 0.03
800000 0.03
850000 0.03
900000 0.973333333333333
950000 0.973333333333333
1000000 0.973333333333333
};
\addplot [very thick, darkorange25512714, mark=diamond*, mark size=2, mark options={solid}]
table {%
0 0
0 0.00666666666666667
50000 0.00666666666666667
100000 0.00666666666666667
150000 0.00666666666666667
200000 0.00666666666666667
250000 0.00666666666666667
300000 0.00666666666666667
350000 0.00666666666666667
400000 0.00666666666666667
450000 0.00666666666666667
500000 0.00666666666666667
550000 0.00666666666666667
600000 0.00666666666666667
650000 0.00666666666666667
700000 0.00666666666666667
750000 0.00666666666666667
800000 0.00666666666666667
850000 0.00666666666666667
900000 0.00666666666666667
950000 0.00666666666666667
1000000 0.00666666666666667
};
\addplot [very thick, lightsalmon255187120, mark=diamond*, mark size=2, mark options={solid}]
table {%
0 0
0 0
50000 0
100000 0
150000 0
200000 0
250000 0
300000 0
350000 0
400000 0
450000 0
500000 0
550000 0
600000 0
650000 0
700000 0
750000 0
800000 1
850000 1
900000 1
950000 1
1000000 1
};
\addplot [very thick, forestgreen4416044, mark=diamond*, mark size=2, mark options={solid}]
table {%
0 0
0 0
50000 0
100000 0.893333333333333
150000 0.893333333333333
200000 0.893333333333333
250000 0.893333333333333
300000 0.923333333333333
350000 0.913333333333333
400000 0.91
450000 0.916666666666667
500000 0.903333333333333
550000 0.94
600000 0.916666666666667
650000 0.87
700000 0.87
750000 0.87
800000 0.87
850000 0.87
900000 0.87
950000 0.87
1000000 0.87
};
\addplot [very thick, lightgreen152223138, mark=diamond*, mark size=2, mark options={solid}]
table {%
0 0
0 0.0366666666666667
50000 0.0366666666666667
100000 0.0366666666666667
150000 0.0366666666666667
200000 0.986666666666667
250000 0.986666666666667
300000 0.986666666666667
350000 0.986666666666667
400000 0.986666666666667
450000 0.986666666666667
500000 0.986666666666667
550000 0.986666666666667
600000 0.986666666666667
650000 0.986666666666667
700000 0.986666666666667
750000 0.986666666666667
800000 0.986666666666667
850000 0.986666666666667
900000 0.986666666666667
950000 0.986666666666667
1000000 0.986666666666667
};
\addplot [very thick, crimson2143940, mark=diamond*, mark size=2, mark options={solid}]
table {%
0 0
0 0.173333333333333
50000 0.173333333333333
100000 0.173333333333333
150000 0.173333333333333
200000 0.173333333333333
250000 0.173333333333333
300000 0.173333333333333
350000 0.173333333333333
400000 0.173333333333333
450000 0.173333333333333
500000 0.173333333333333
550000 0.173333333333333
600000 0.173333333333333
650000 0.173333333333333
700000 0.173333333333333
750000 0.173333333333333
800000 0.173333333333333
850000 0.173333333333333
900000 0.173333333333333
950000 0.173333333333333
1000000 0.173333333333333
};
\addplot [very thick, lightsalmon255152150, mark=diamond*, mark size=2, mark options={solid}]
table {%
0 0
0 0
50000 0
100000 0
150000 0
200000 0
250000 0
300000 0
350000 0
400000 0
450000 0
500000 0
550000 0
600000 0
650000 0
700000 0
750000 0
800000 0
850000 0
900000 0
950000 0
1000000 0
};
\addplot [very thick, mediumpurple148103189, mark=diamond*, mark size=2, mark options={solid}]
table {%
0 0
0 0
50000 0
100000 0
150000 0
200000 0
250000 0
300000 0
350000 0
400000 0
450000 0
500000 0
550000 0
600000 0
650000 0
700000 0
750000 0.17
800000 0.17
850000 0.17
900000 0.17
950000 0.903333333333333
1000000 0.993333333333333
};
\addplot [very thick, thistle197176213, mark=diamond*, mark size=2, mark options={solid}]
table {%
0 0
0 0.0466666666666667
50000 0.0466666666666667
100000 0.0466666666666667
150000 0.0466666666666667
200000 0.0466666666666667
250000 0.996666666666667
300000 0.996666666666667
350000 0.996666666666667
400000 0.996666666666667
450000 0.996666666666667
500000 0.996666666666667
550000 0.996666666666667
600000 0.996666666666667
650000 0.996666666666667
700000 0.996666666666667
750000 0.996666666666667
800000 0.996666666666667
850000 0.996666666666667
900000 0.996666666666667
950000 0.996666666666667
1000000 0.996666666666667
};
\addplot [very thick, sienna1408675, mark=diamond*, mark size=2, mark options={solid}]
table {%
0 0
0 0
50000 0
100000 0
150000 0
200000 0
250000 0
300000 0
350000 0
400000 0
450000 0
500000 0
550000 0
600000 0
650000 0
700000 1
750000 1
800000 1
850000 1
900000 1
950000 1
1000000 1
};
\addplot [very thick, rosybrown196156148, mark=diamond*, mark size=2, mark options={solid}]
table {%
0 0
0 0.0133333333333333
50000 0.0133333333333333
100000 0.0133333333333333
150000 0.0133333333333333
200000 0.0133333333333333
250000 0.0133333333333333
300000 0.0133333333333333
350000 0.0133333333333333
400000 0.0133333333333333
450000 0.0133333333333333
500000 0.0133333333333333
550000 0.0133333333333333
600000 0.0133333333333333
650000 0.0133333333333333
700000 0.0133333333333333
750000 0.0133333333333333
800000 0.0133333333333333
850000 0.986666666666667
900000 0.986666666666667
950000 0.986666666666667
1000000 0.986666666666667
};
\addplot [line width=2pt, black, dashed]
table {%
0 0.95
1000000 0.95
};
\addplot [line width=2pt, blue, mark=diamond*, mark size=3, mark options={solid}]
table {%
0 0
0 0
50000 0
100000 0
150000 0.00150821135802469
200000 0.0405845965432099
250000 0.86677102617284
300000 0.895879008395062
350000 0.886176347654321
400000 0.882942127407408
450000 0.889410567901235
500000 0.87647368691358
550000 0.91205010962963
600000 0.889410567901235
650000 0.844131484444445
700000 0.844131484444445
750000 0.844131484444445
800000 0.844131484444445
850000 0.844131484444445
900000 0.844131484444445
950000 0.867521185185185
1000000 0.953953185185185
};
\addplot [line width=2pt, black, mark=diamond*, mark size=3, mark options={solid}]
table {%
0 0
0 0
50000 0
100000 0
150000 0
200000 0.0366666666666667
250000 0.92
300000 0.87
350000 0.883333333333333
400000 0.853333333333333
450000 0.863333333333333
500000 0.87
550000 0.883333333333333
600000 0.89
650000 0.886666666666667
700000 0.846666666666667
750000 0.883333333333333
800000 0.89
850000 0.876666666666667
900000 0.896666666666667
950000 0.896666666666667
1000000 0.966666666666667
};
\addplot [line width=4pt, red, dashed, on layer=foreground]
table {%
650000 -0.05
650000 1.05
};
\end{axis}

\end{tikzpicture}

%% file: figures/minigrid_pixel_training_schedule.tex
\begin{tikzpicture}

\definecolor{darkgray176}{RGB}{176,176,176}
\definecolor{forestgreen4416044}{RGB}{44,160,44}
\definecolor{lightgreen152223138}{RGB}{152,223,138}
\definecolor{lightsalmon255187120}{RGB}{255,187,120}
\definecolor{lightsteelblue174199232}{RGB}{174,199,232}
\definecolor{mediumpurple148103189}{RGB}{148,103,189}
\definecolor{rosybrown196156148}{RGB}{196,156,148}
\definecolor{sienna1408675}{RGB}{140,86,75}
\definecolor{steelblue31119180}{RGB}{31,119,180}
\definecolor{thistle197176213}{RGB}{197,176,213}

\begin{axis}[
height=4.5cm,
width=0.99\textwidth,
tick align=inside,
tick pos=left,
x grid style={white!69.0196078431373!black},
ticklabel style = {font=\footnotesize},
xmajorgrids,
xmin=0, xmax=1000000,
xtick style={color=black},
y grid style={white!69.0196078431373!black},
ymajorgrids,
ymin=-0.55, ymax=11.55,
ytick style={color=black},
ytick={0,1,2,3,4,5,6,7,8,9,10,11},
yticklabels = {0,1,2,3,4,5,6,7,8,9,10,11},
xtick={0, 200000, 400000, 600000, 800000},
xticklabels={0.0, 0.2, 0.4, 0.6, 0.8},
every x tick scale label/.style={
    at={(1,0)},xshift=-15.5pt,yshift=-10.0pt,anchor=south west,inner sep=0pt},
ylabel={Subsystem Index},
xlabel={Elapsed Total Training Steps}
]

\addplot [line width=8pt, forestgreen4416044]
table {%
0 4
50000 4
};
\addplot [line width=8pt, forestgreen4416044]
table {%
50000 4
100000 4
};
\addplot [line width=8pt, steelblue31119180]
table {%
100000 0
150000 0
};
\addplot [line width=8pt, lightgreen152223138]
table {%
150000 5
200000 5
};
\addplot [line width=8pt, thistle197176213]
table {%
200000 9
250000 9
};
\addplot [line width=8pt, forestgreen4416044]
table {%
250000 4
300000 4
};
\addplot [line width=8pt, forestgreen4416044]
table {%
300000 4
350000 4
};
\addplot [line width=8pt, forestgreen4416044]
table {%
350000 4
400000 4
};
\addplot [line width=8pt, forestgreen4416044]
table {%
400000 4
450000 4
};
\addplot [line width=8pt, forestgreen4416044]
table {%
450000 4
500000 4
};
\addplot [line width=8pt, forestgreen4416044]
table {%
500000 4
550000 4
};
\addplot [line width=8pt, forestgreen4416044]
table {%
550000 4
600000 4
};
\addplot [line width=8pt, forestgreen4416044]
table {%
600000 4
650000 4
};
\addplot [line width=8pt, sienna1408675]
table {%
650000 10
700000 10
};
\addplot [line width=8pt, mediumpurple148103189]
table {%
700000 8
750000 8
};
\addplot [line width=8pt, lightsalmon255187120]
table {%
750000 3
800000 3
};
\addplot [line width=8pt, rosybrown196156148]
table {%
800000 11
850000 11
};
\addplot [line width=8pt, lightsteelblue174199232]
table {%
850000 1
900000 1
};
\addplot [line width=8pt, mediumpurple148103189]
table {%
900000 8
950000 8
};
\addplot [line width=8pt, mediumpurple148103189]
table {%
950000 8
1000000 8
};
\addplot [line width=5pt, red, dashed]
table {%
650000 -0.55
650000 11.55
};
\end{axis}

\end{tikzpicture}

%% file: tex/06_3_continuous_labyrinth.tex
\subsection{Experiments in Environments with Continuous State and Action Spaces}
\label{sec:experiments_continuous_labyrinth}

To demonstrate the framework's ability to generalize to RL settings with large state and action spaces and complex environment dynamics, we also implemented a continuous-state and continuous-action version of the labyrinth environment in the video game engine \textit{Unity} \shortcite{juliani2018unity}.
The Unity environment is illustrated on the right in Figure \ref{fig:experiments_labyrinth_environments}.

\subsubsection{The Continuous Labyrinth Environment}
In this version of the task, the RL system must roll a ball from the initial location to the goal location.
The set \(\mdpActionSet\) of available actions consists of all of the force vectors, with magnitude of at most 1 \(N\), that can be applied to the ball in the horizontal plane.
The set \(\mdpStateSet\) of environment states is given by all possible locations \((x,y)\) and velocities \((\dot{x}, \dot{y})\) of the ball within the labyrinth.
The action space \(\mdpActionSet\) is thus a compact subset of \(\mathbb{R}^2\) while the state space \(\mdpStateSet\) is a compact subset of \(\mathbb{R}^4\). 
The environment's dynamics are governed by \textit{Unity}'s rigid-body physics simulator.

Subtask entry and exit conditions are implemented as compact subsets of points \\ \((x,y,\dot{x}, \dot{y}) \in \mathbb{R}^4\) such that \(\sqrt{(x - x_{\controller})^2 + (y-y_{\controller})^2} \leq 0.5m\) and \(\sqrt{\dot{x}^2 + \dot{y}^2} \leq 0.5\frac{m}{s}\) respectively, for some pre-specified \(x_c\) and \(y_c\).
Here, \((x,y)\) and \((\dot{x}, \dot{y})\) correspond to the position and velocity values, respectively, and \((x_{\controller}\), \(y_{\controller})\) represents pre-specified entry or exit position.

\subsubsection{Training the Subsystem Policies}
Similarly to as in the discrete environment, the agent receives a reward of \(1\) when it reaches an exit condition.
However, in the continuous environment, a more dense reward signal is required for training the subsystems.
If \((x_{\controller}, y_{\controller})\) is the exit position of subsystem \(\controller\), then the subsystem receives a reward that is proportional to \(- \sqrt{(x - x_{\controller})^2 + (y-y_{\controller})^2}\) every timestep.
To prevent the agent from learning to purposefully move into the lava in order to end the episode early and stop itself from accumulating this negative reward, we additionally penalize the agent with a reward of \(-1\) if it touches the lava.
The observations provided to the agent consist of its own position and velocity, as well as the position of the target.

The subsystem policy networks are identical to those described in \S \ref{sec:experiment_discrete_labyrinth}, with the exception of the output layer, which is adapted for the continuous action space of the environment.
More specifically, the outputs of the policy networks parametrize a Gaussian distribution with a diagonal covariance matrix.
We use the PPO algorithm with the same parameters as in \S \ref{sec:experiment_discrete_labyrinth} to train the RL subsystem policies.

In this continuous labyrinth environment, we use \(N_{max} = 250,000\) training steps. 
A complete training run for the continuous labyrinth environment, which consists of roughly 800,000 total environment interactions across all subsystems, takes approximately 60 minutes of wall-clock time.

\subsubsection{Compositional Reinforcement Learning Systems are Agnostic to the Complexity of the Underlying Dynamics}
Figure \ref{fig:cont_lab_training_curves} illustrates the experimental results in the continuous labyrinth environment.
Qualitatively, these results closely resemble our observations from the discrete labyrinth presented in \S \ref{sec:experiment_discrete_labyrinth} and \S \ref{sec:experiments_partial_observability}, 
despite significant differences in the environment's dynamics and in its representations of states and actions.
The ICRL algorithm again initially attempts to move straight down past the lava, before automatically refining the subtask specifications in order to focus on training the subsystems that take the alternate route through the labyrinth.
This similarity in the algorithm's behavior when applied to different types of environments helps illustrate the generality of the proposed framework; ICRL is agnostic to the details of the environment dynamics, so long as the RL subsystems are trained effectively.

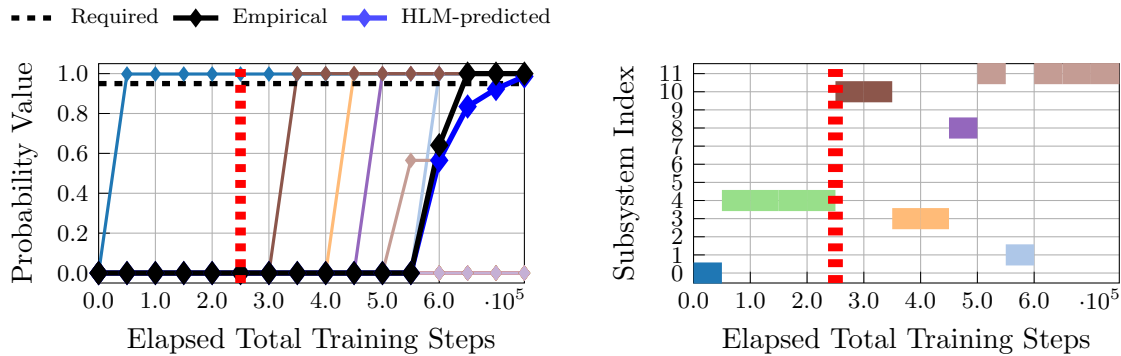
\begin{figure*}[t!]
    \centering
    \begin{subfigure}[t]{0.48\textwidth}
        \centering \input{figures/contLab_training_curves}
        \label{fig:supp_cont_lab_training_curves}
    \end{subfigure}%
    \hspace*{\fill}
    \centering
    \begin{subfigure}[t]{0.48\textwidth}
        \centering \input{figures/contLab_training_schedule}
        \label{fig:supp_cont_lab_training_schedule}
    \end{subfigure}
    \caption{
        Numerical results for the continuous labyrinth environment.
        Left: Estimated task and subtask success probabilities during training.
        Right: Automatically generated subsystem training schedule.
    } 
    \label{fig:cont_lab_training_curves}
\end{figure*}

\subsubsection{Comparison with a Monolithic Approach to Reinforcement Learning}
We again compare the proposed ICRL algorithm to a monolithic approach to learning, in which the entire task is treated as a single subsystem.
The task specification requires that the agent reaches the target state \(\mdpState \in \controllerFinalStateSet_{target}\) with a probability of at least \(0.95\) from the labyrinth's initial state. 
The observations, actions, and PPO algorithm parameters are identical for both the monolithic and the compositional approaches to learning.

In this continuous labyrinth environment, the monolithic RL agent was entirely unsuccessful; it was unable to learn a policy that satisfies the task specification, although it was allowed more than \(13\) million training time-steps.
By contrast, our compositional RL system learned to satisfy the task specification after only 450,000 time-steps, which is $30$ times less than the number of training steps that were taken by the unsuccessful monolithic agent.
We reiterate that the compositional RL agent is taking advantage of information that is not available to the monolithic agent (namely, the entry and exit conditions of the subtasks). 
We include this comparison only to briefly highlight the benefits that are enjoyed when such information is available and is used to decompose the overall task.

\subsubsection{Additional Discussion.}
\label{sec:additional_discussion}

We note that all predictions made using the HLM will be sensitive to the values of \(\controllerPerformanceEstimate_{\controller}\) -- the estimated lower bounds on the probability of subtask success.
In our experiments, we compute \(\controllerPerformanceEstimate_{\controller}\) empirically by rolling out the subsystems from randomly sampled entry conditions.
While this technique provides only rough estimates of the true value of the lower bound (particularly in the case of the continuous labyrinth environment which has an uncountably infinite number of entry conditions per subtask), our results demonstrate that these empirical approximations are sufficient for high-level decision making.
The algorithm makes effective use of the HLM predictions to automatically select the subsystems that require training.
Any methods to further improve the estimates of \(\controllerPerformanceEstimate_{\controller}\) will only improve the performance of the ICRL algorithm.

%% file: figures/contLab_training_curves.tex
\begin{tikzpicture}

\definecolor{color0}{rgb}{0.12156862745098,0.466666666666667,0.705882352941177}
\definecolor{color1}{rgb}{0.682352941176471,0.780392156862745,0.909803921568627}
\definecolor{color2}{rgb}{1,0.498039215686275,0.0549019607843137}
\definecolor{color3}{rgb}{1,0.733333333333333,0.470588235294118}
\definecolor{color4}{rgb}{0.172549019607843,0.627450980392157,0.172549019607843}
\definecolor{color5}{rgb}{0.596078431372549,0.874509803921569,0.541176470588235}
\definecolor{color6}{rgb}{0.83921568627451,0.152941176470588,0.156862745098039}
\definecolor{color7}{rgb}{1,0.596078431372549,0.588235294117647}
\definecolor{color8}{rgb}{0.580392156862745,0.403921568627451,0.741176470588235}
\definecolor{color9}{rgb}{0.772549019607843,0.690196078431373,0.835294117647059}
\definecolor{color10}{rgb}{0.549019607843137,0.337254901960784,0.294117647058824}
\definecolor{color11}{rgb}{0.768627450980392,0.611764705882353,0.580392156862745}

\begin{customlegend}[legend columns=3, legend style={align=center, column sep=0.5ex, font=\scriptsize, draw=none, at={(62.0mm, 38.0mm)}}, legend entries={Required, Empirical, HLM-predicted}]
\addlegendimage{line width=2pt, black, dashed}
\addlegendimage{line width=2pt, black, mark=diamond*, mark size=2}
\addlegendimage{line width=2pt, blue!70!white, mark=diamond*, mark size=2}
\end{customlegend}

\begin{axis}[
tick align=inside,
tick pos=left,
x grid style={white!69.0196078431373!black},
  ticklabel style = {font=\footnotesize},
xmajorgrids,
xmin=0, xmax=750000,
xtick style={color=black},
y grid style={white!69.0196078431373!black},
ymajorgrids,
ymin=-0.05, ymax=1.05,
height=4.5cm,
width=0.99\textwidth,
ytick={0.0, 0.2, 0.4, 0.6, 0.8, 1.0},
yticklabels={0.0, 0.2, 0.4, 0.6, 0.8, 1.0},
xtick={0, 100000, 200000, 300000, 400000, 500000, 600000, 800000, 1000000, 1200000, 1400000},
xticklabels={0.0, 1.0, 2.0, 3.0, 4.0, 5.0, 6.0, 8.0, 10.0, 12.0, },
ytick style={color=black},
legend to name=named,
every x tick scale label/.style={
    at={(1,0)},xshift=-15.5pt,yshift=-10.0pt,anchor=south west,inner sep=0pt},
ylabel={Probability Value},
xlabel={Elapsed Total Training Steps}
]
\addplot [very thick, color0, mark=diamond*, mark size=2, mark options={solid}, forget plot]
table {%
0 0
0 0
50000 0.998
100000 0.998
150000 0.998
200000 0.998
250000 0.998
300000 0.998
350000 0.998
400000 0.998
450000 0.998
500000 0.998
550000 0.998
600000 0.998
650000 0.998
700000 0.998
750000 0.998
};
\addplot [very thick, color1, mark=diamond*, mark size=2, mark options={solid}, forget plot]
table {%
0 0
0 0
50000 0
100000 0
150000 0
200000 0
250000 0
300000 0
350000 0
400000 0
450000 0
500000 0
550000 0
600000 1
650000 1
700000 1
750000 1
};
\addplot [very thick, color2, mark=diamond*, mark size=2, mark options={solid}, forget plot]
table {%
0 0
0 0
50000 0
100000 0
150000 0
200000 0
250000 0
300000 0
350000 0
400000 0
450000 0
500000 0
550000 0
600000 0
650000 0
700000 0
750000 0
};
\addplot [very thick, color3, mark=diamond*, mark size=2, mark options={solid}, forget plot]
table {%
0 0
0 0
50000 0
100000 0
150000 0
200000 0
250000 0
300000 0
350000 0
400000 0
450000 1
500000 1
550000 1
600000 1
650000 1
700000 1
750000 1
};
\addplot [very thick, color4, mark=diamond*, mark size=2, mark options={solid}, forget plot]
table {%
0 0
0 0
50000 0
100000 0
150000 0
200000 0
250000 0
300000 0
350000 0
400000 0
450000 0
500000 0
550000 0
600000 0
650000 0
700000 0
750000 0
};
\addplot [very thick, color5, mark=diamond*, mark size=2, mark options={solid}, forget plot]
table {%
0 0
0 0
50000 0
100000 0
150000 0
200000 0
250000 0
300000 0
350000 0
400000 0
450000 0
500000 0
550000 0
600000 0
650000 0
700000 0
750000 0
};
\addplot [very thick, color6, mark=diamond*, mark size=2, mark options={solid}, forget plot]
table {%
0 0
0 0
50000 0
100000 0
150000 0
200000 0
250000 0
300000 0
350000 0
400000 0
450000 0
500000 0
550000 0
600000 0
650000 0
700000 0
750000 0
};
\addplot [very thick, color7, mark=diamond*, mark size=2, mark options={solid}, forget plot]
table {%
0 0
0 0
50000 0
100000 0
150000 0
200000 0
250000 0
300000 0
350000 0
400000 0
450000 0
500000 0
550000 0
600000 0
650000 0
700000 0
750000 0
};
\addplot [very thick, color8, mark=diamond*, mark size=2, mark options={solid}, forget plot]
table {%
0 0
0 0
50000 0
100000 0
150000 0
200000 0
250000 0
300000 0
350000 0
400000 0
450000 0
500000 1
550000 1
600000 1
650000 1
700000 1
750000 1
};
\addplot [very thick, color9, mark=diamond*, mark size=2, mark options={solid}, forget plot]
table {%
0 0
0 0
50000 0
100000 0
150000 0
200000 0
250000 0
300000 0
350000 0
400000 0
450000 0
500000 0
550000 0
600000 0
650000 0
700000 0
750000 0
};
\addplot [very thick, color10, mark=diamond*, mark size=2, mark options={solid}, forget plot]
table {%
0 0
0 0
50000 0
100000 0
150000 0
200000 0
250000 0
300000 0
350000 1
400000 1
450000 1
500000 1
550000 1
600000 1
650000 1
700000 1
750000 1
};
\addplot [very thick, color11, mark=diamond*, mark size=2, mark options={solid}, forget plot]
table {%
0 0
0 0
50000 0
100000 0
150000 0
200000 0
250000 0
300000 0
350000 0
400000 0
450000 0
500000 0
550000 0.565
600000 0.565
650000 0.835
700000 0.923
750000 0.986
};
\addplot [line width=2pt, black, dashed]
table {%
0 0.95
750000 0.95
};
\addlegendentry{Required Probability of Success}
\addplot [line width=2pt, blue, mark=diamond*, mark size=3, mark options={solid}]
table {%
0 0
0 0
50000 0
100000 0
150000 0
200000 0
250000 0
300000 0
350000 0
400000 0
450000 0
500000 0
550000 0
600000 0.565
650000 0.835
700000 0.923
750000 0.986
};
\addlegendentry{Lower Bound on Probability of Task Success}
\addplot [line width=2pt, black, mark=diamond*, mark size=3, mark options={solid}]
table {%
0 0
0 0
50000 0
100000 0
150000 0
200000 0
250000 0
300000 0
350000 0
400000 0
450000 0
500000 0
550000 0
600000 0.641
650000 1
700000 1
750000 1
};
\addlegendentry{Empirically Measured Probability of Task Success}
\addplot [line width=4pt, red, dashed, forget plot, on layer=foreground]
table {%
250000 -0.05
250000 1.05
};
\end{axis}

\end{tikzpicture}

%% file: figures/contLab_training_schedule.tex
\begin{tikzpicture}

\definecolor{color0}{rgb}{0.12156862745098,0.466666666666667,0.705882352941177}
\definecolor{color1}{rgb}{0.596078431372549,0.874509803921569,0.541176470588235}
\definecolor{color2}{rgb}{0.549019607843137,0.337254901960784,0.294117647058824}
\definecolor{color3}{rgb}{1,0.733333333333333,0.470588235294118}
\definecolor{color4}{rgb}{0.580392156862745,0.403921568627451,0.741176470588235}
\definecolor{color5}{rgb}{0.768627450980392,0.611764705882353,0.580392156862745}
\definecolor{color6}{rgb}{0.682352941176471,0.780392156862745,0.909803921568627}

\begin{axis}[
height=4.5cm,
width=0.99\textwidth,
tick align=inside,
tick pos=left,
x grid style={white!69.0196078431373!black},
ticklabel style = {font=\footnotesize},
xmajorgrids,
xmin=0, xmax=750000,
xtick style={color=black},
y grid style={white!69.0196078431373!black},
ymajorgrids,
ymin=-0.55, ymax=11.55,
ytick style={color=black},
ytick={0,1,2,3,4,5,6,7,8,9,10,11},
yticklabels = {0,1,2,3,4,5,6,7,8,9,10,11},
xtick={0, 100000, 200000, 300000, 400000, 500000, 600000, 800000, 1000000, 1200000, 1400000},
xticklabels={0.0, 1.0, 2.0, 3.0, 4.0, 5.0, 6.0, 8.0, 10.0, 12.0, },
every x tick scale label/.style={
    at={(1,0)},xshift=-15.5pt,yshift=-10.0pt,anchor=south west,inner sep=0pt},
ylabel={Subsystem Index},
xlabel={Elapsed Total Training Steps}
]
\addplot [line width=8pt, color0]
table {%
0 0
50000 0
};
\addplot [line width=8pt, color1]
table {%
50000 4
100000 4
};
\addplot [line width=8pt, color1]
table {%
100000 4
150000 4
};
\addplot [line width=8pt, color1]
table {%
150000 4
200000 4
};
\addplot [line width=8pt, color1]
table {%
200000 4
250000 4
};
\addplot [line width=8pt, color2]
table {%
250000 10
300000 10
};
\addplot [line width=8pt, color2]
table {%
300000 10
350000 10
};
\addplot [line width=8pt, color3]
table {%
350000 3
400000 3
};
\addplot [line width=8pt, color3]
table {%
400000 3
450000 3
};
\addplot [line width=8pt, color4]
table {%
450000 8
500000 8
};
\addplot [line width=8pt, color5]
table {%
500000 11
550000 11
};
\addplot [line width=8pt, color6]
table {%
550000 1
600000 1
};
\addplot [line width=8pt, color5]
table {%
600000 11
650000 11
};
\addplot [line width=8pt, color5]
table {%
650000 11
700000 11
};
\addplot [line width=8pt, color5]
table {%
700000 11
750000 11
};
\addplot [line width=5.6pt, red, dashed]
table {%
250000 -0.55
250000 11.55
};
\end{axis}

\end{tikzpicture}

%% file: tex/06_4_modular_task_transfer.tex
\subsection{Compositionality for Planning and Verification in Task Transfer}
\label{sec:experiments_modular_task_transfer}

A key benefit of the proposed compositional approach to modeling and training RL systems, is the modular nature of our abstract representations of the subsystems.
By creating models of each subsystem in terms of its entry conditions, exit conditions, and its probability of subtask success, we are able to define mathematical interfaces between them.
These interfaces enable us to compose the subsystem models, and to verify their performance, before executing the composite system itself.
Trained subsystems can be mathematically composed in new ways to solve new tasks.
The high-level model representing this new composite system, can then be used to plan and to verify the system's execution of the new task, without requiring additional interactions with the task environment.

\begin{figure*}[t!]
    \input{figures/shuffle_labyrinth}
    \caption{
    An illustration of the four labyrinth environments and the results of the ICRL algorithm after consecutive training and task transfer.
    Each color is associated with a particular environment. 
    The coloring of the trajectory components indicates the environment that the subsystem policy was initially trained on.
    A compositional RL system is trained on each task until its task specification is satisfied.
    We begin training on labyrinth \(1\).
    The arrows indicate the sequence in which the tasks were trained.
    The ICRL algorithm automatically reuses previously trained subsystems wherever possible when applied to a new environment.
    The six unique rooms that were used to construct the four labyrinth environments are illustrated at the bottom of the figure.
    } 
    \label{fig:four_labyrinth_figure}
\end{figure*}
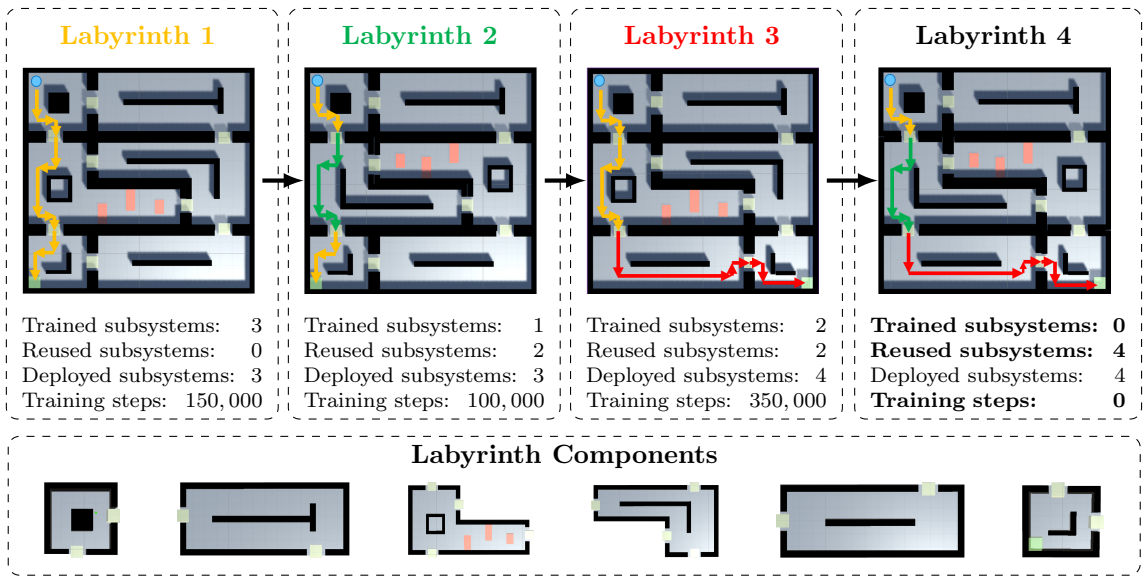

\subsubsection{The Four Labyrinth Environments}
To demonstrate this capability, we present the results of applying the ICRL algorithm to the four labyrinth environments illustrated in Figure \ref{fig:four_labyrinth_figure}.
Each of these four labyrinths are constructed by arranging the same set of six rooms (illustrated by the figure's bottom row) into different configurations.
Similarly as to the environment illustrated in Figure \ref{fig:experiments_labyrinth_environments}, the RL system's objective in each of these environments is to roll a ball from the initial state to the goal state, while avoiding the unsafe lava states.
In all four of the labyrinths, the task specification requires that this objective be completed from the initial state with a probability of at least \(0.95\).
The environment states and actions, as well as the implementation of the subtask entry and exit conditions, are identical to those defined in \S \ref{sec:experiments_continuous_labyrinth}.

\subsubsection{Transfering Subsystem Policies Between Labyrinths}

We apply the ICRL algorithm to the four labyrinth environments separately, in the order illustrated by the black arrows in Figure \ref{fig:four_labyrinth_figure}.
We again train the subsystem policies using the reward function and the PPO algorithm setup described in \S \ref{sec:experiments_continuous_labyrinth}, however, we now make a slight modification to the agent's observations.
Instead of observing its location relative to the initial state of the entire labyrinth, each subsystem policy now receives observations of its position relative to the entry condition of its current subtask.

Many of the subtask policies are transferable between the different environments. 
Once a policy has learned to navigate a particular room in one of the labyrinths, the policy may be reused to navigate the same room in other labyrinths without necessarily requiring further training.
Our compositional approach to RL is able to model the transfer of such previously trained subsystem policies, and to take advantage of this information in planning solutions for new tasks.

We begin by applying the ICRL algorithm to labyrinth 1.
Let HLM\(_{1}\) denote the high-level model that is constructed to represent the corresponding compositional system.
Because this is the first task that we train on, there are no pretrained subsystem policies to be reused.
Once the task specification in the first labyrinth has been satisfied, we move on to the task in labyrinth 2.
We construct a separate high-level model, HLM\(_{2}\), for this separate environment.
While the graph structure of HLM\(_{2}\) will be different from that of HLM\(_{1}\), we may associate specific transitions in HLM\(_{2}\)---which represent executions of particular subsystems---with the transitions in HLM\(_{1}\) corresponding to the executions of the same subsystems.
More specifically, information pertaining to the model's parametric transition probability values \(\bernoulliProb_{\controller}\)---representing lower bounds on the probability of success for the various subsystem policies---may be reused.
In particular, if we intend to reuse a subsystem policy \(\policy_{\controller}\) in labyrinth 2, then we impose the constraints \eqref{eq:hlm_opt_lb_constraints}-\eqref{eq:hlm_opt_ub_constraints} that were empirically estimated while training \(\policy_{\controller}\) in labyrinth 1, when solving the parameter synthesis problem \eqref{eq:hlm_opt_objective}-\eqref{eq:hlm_opt_ub_constraints} in labyrinth 2.
In other words, we reuse our abstract models of the pretrained subsystems, when using the ICRL algorithm to plan and to verify a solution for the new task.

We run the ICRL algorithm until the task specification in labyrinth 2 is satisfied, before repeating the process in labyrinths 3 and 4.
We note that the latter labyrinth environments may reuse subsystem policies that were originally trained in any of the previously solved task environments.

\subsubsection{Numerical Results}

Figure \ref{fig:four_labyrinth_figure} illustrates the numerical results of this approach to compositional subtask transfer.
In labyrinth 1, three subsystems are trained to reach the goal, using a total of \(150,000\) training steps.
In labyrinth 2, the ICRL algorithm automatically plans a meta-policy that reuses two of the previously trained subsystem policies (orange), and that deploys one subsystem that has not yet been trained (green).
Using the empirically estimated probabilities of subtask success for the pretrained subsystems, a subtask specification is synthesized for the untrained one.
Once this subsystem (green) has been trained to meet its subtask specification, the composite system is guaranteed to meet its overall task specification in labyrinth 2 without having to train the reused subsystems (orange) in this new environment at all.
We empirically verify that the composite system achieves the same probability of task success as is predicted using HLM\(_{2}\).

The process is repeated for labyrinths 3 and 4.
In labyrinth 4, no additional training is required in the new task environment; all of the deployed subsystems policies were previously trained in one of the first three labyrinths.
We use the high-level model HLM\(_{4}\) to synthesize a meta-policy, and to verify that the corresponding compositional system will satisfy its probabilistic task specification, without taking a single step in the new environment.

\subsubsection{Discussion on Planning and Verification in Task Transfer}

We note that by incorporating the previously estimated constraint values \eqref{eq:hlm_opt_lb_constraints}-\eqref{eq:hlm_opt_ub_constraints} into optimization problem \eqref{eq:hlm_opt_objective}-\eqref{eq:hlm_opt_ub_constraints} while solving a new task, the ICRL algorithm automatically synthesizes meta-policies that take advantage of the transferred subsystem policies as much as possible.
In other words, it automatically synthesizes a meta-policy that takes a route that minimizes the number of additional subsystems requiring training.

Furthermore, we emphasize the fact that this high-level reasoning and decision-making is performed without actually deploying the composite reinforcement learning system in the new task environment.
By using the HLM to predict the result of creating new compositions of previously trained subsystems, we may verify the system's performance on new tasks before deploying said system at all.

%% file: figures/shuffle_labyrinth.tex
\begin{tikzpicture}[
    labelText/.style={font=\scriptsize, inner xsep=0.0cm, inner ysep=0.0cm},
    eqText/.style={font=\scriptsize, inner xsep=0.0cm, inner ysep=0.0cm},
]

\definecolor{color1}{RGB}{255,192,0}
\definecolor{color2}{RGB}{0,176,80}
\definecolor{color3}{RGB}{255,0,0}
\definecolor{color4}{RGB}{0,0,0}

\newlength{\interLabyrinthSpacing}
\setlength{\interLabyrinthSpacing}{0.55cm}

\node (labyrinth1) [
    labelText,
    xshift=0cm, 
    yshift=0,
    align=center,
    ] {
        \includegraphics[height=3cm]{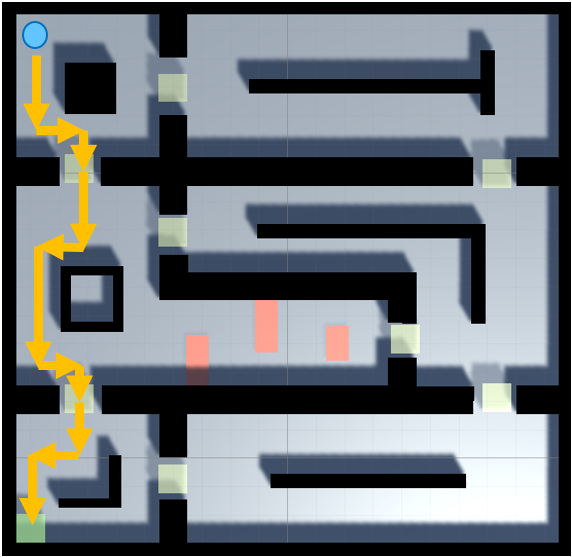}
    };
    
\node (labyrinth1Title) [
    font=\small,
    xshift=0cm, 
    yshift=0,
    align=center,
    above=0.1cm of labyrinth1,
    ] {
        {\color{color1} \textbf{Labyrinth 1}}
    };
    
\node (labyrinth1Results) [
    labelText,
    align=left,
    text width=3.2cm,
    below = 0.3cm of labyrinth1,
] {
    Trained subsystems: \hfill 3\\
    Reused subsystems: \hfill 0 \\
    Deployed subsystems: \hfill 3\\
    Training steps: \hfill \(150,000\)\\ 
};

\node (labyrinth1Enclosure) [
    draw=black,
    dashed,
    inner xsep=2mm,
    inner ysep=1mm,
    rounded corners=0.15cm,
    fit={(labyrinth1) (labyrinth1Title) (labyrinth1Results)}] {};
    
\node (labyrinth2) [
    labelText,
    right=\interLabyrinthSpacing of labyrinth1,
    xshift=0cm, 
    yshift=0,
    align=center,
    ] {
        \includegraphics[height=3cm]{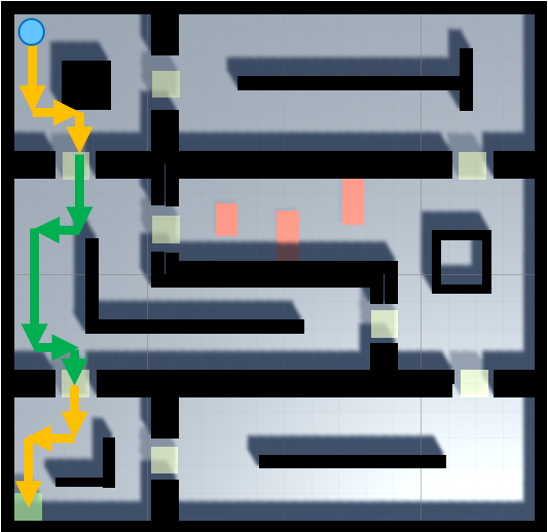}
    };

\node (labyrinth2Title) [
    font=\small,
    xshift=0cm, 
    yshift=0,
    align=center,
    above=0.1cm of labyrinth2,
    ] {
        {\color{color2} \textbf{Labyrinth 2}}
    };
    
\node (labyrinth2Results) [
    labelText,
    align=left,
    text width=3.2cm,
    below = 0.3cm of labyrinth2,
] {
    Trained subsystems: \hfill 1\\
    Reused subsystems: \hfill 2 \\
    Deployed subsystems: \hfill 3\\
    Training steps: \hfill \(100,000\) \\
};

\node (labyrinth2Enclosure) [
    draw=black,
    dashed,
    inner xsep=2mm,
    inner ysep=1mm,
    rounded corners=0.15cm,
    fit={(labyrinth2) (labyrinth2Title) (labyrinth2Results)}] {};
    
\node (labyrinth3) [
    labelText,
    right=\interLabyrinthSpacing of labyrinth2,
    xshift=0cm, 
    yshift=0,
    align=center,
    ] {
        \includegraphics[height=3cm]{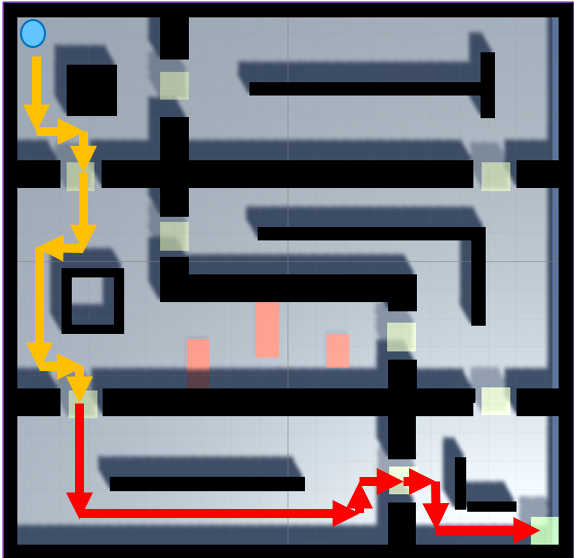}
    };

\node (labyrinth3Title) [
    font=\small,
    xshift=0cm, 
    yshift=0,
    align=center,
    above=0.1cm of labyrinth3,
    ] {
        {\color{color3} \textbf{Labyrinth 3}}
    };
    
\node (labyrinth3Results) [
    labelText,
    align=left,
    text width=3.2cm,
    below = 0.3cm of labyrinth3,
] {
    Trained subsystems: \hfill 2\\
    Reused subsystems: \hfill 2 \\
    Deployed subsystems: \hfill 4\\
    Training steps: \hfill \(350,000\) \\
};

\node (labyrinth3Enclosure) [
    draw=black,
    dashed,
    inner xsep=2mm,
    inner ysep=1mm,
    rounded corners=0.15cm,
    fit={(labyrinth3) (labyrinth3Title) (labyrinth3Results)}] {};
    
\node (labyrinth4) [
    labelText,
    right=0.68cm of labyrinth3,
    xshift=0cm, 
    yshift=0,
    align=center,
    ] {
        \includegraphics[height=3cm]{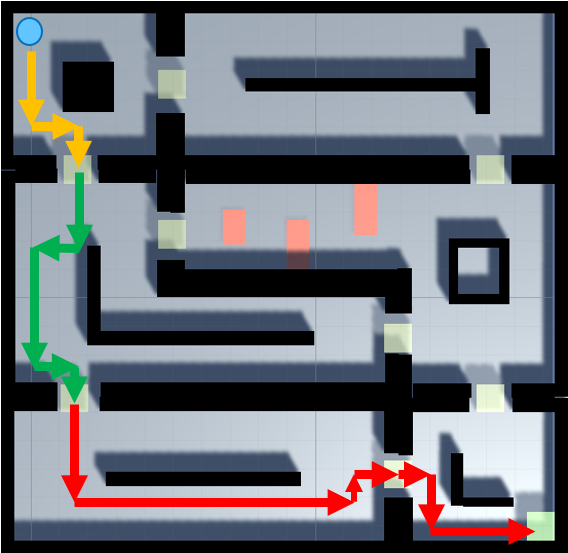}
    };

\node (labyrinth4Title) [
    font=\small,
    xshift=0cm, 
    yshift=0,
    align=center,
    above=0.1cm of labyrinth4,
    ] {
        {\color{color4} \textbf{Labyrinth 4}}
    };

\node (labyrinth4Results) [
    labelText,
    align=left,
    text width=3.4cm,
    below = 0.3cm of labyrinth4,
] {
    \textbf{Trained subsystems:} \hfill \textbf{0}\\
    \textbf{Reused subsystems:} \hfill \textbf{4} \\
    Deployed subsystems: \hfill 4\\
    \textbf{Training steps:} \hfill \textbf{0}\\
};

\node (labyrinth4Enclosure) [
    draw=black,
    dashed,
    inner xsep=2mm,
    inner ysep=1mm,
    rounded corners=0.15cm,
    fit={(labyrinth4) (labyrinth4Title) (labyrinth4Results)}] {};

\draw [-latex, color=black, ultra thick] (labyrinth1.east) -- (labyrinth2.west);
\draw [-latex, color=black, ultra thick] (labyrinth2.east) -- (labyrinth3.west);
\draw [-latex, color=black, ultra thick] (labyrinth3.east) -- (labyrinth4.west);


\newlength{\interRoomSpacing}
\setlength{\interRoomSpacing}{0.5cm}

\newlength{\roomHeight}
\setlength{\roomHeight}{1.0cm}

\node (room1) [
    below=0.7cm of labyrinth1Enclosure.south,
    xshift=-0.8cm,
] {\includegraphics[height=\roomHeight]{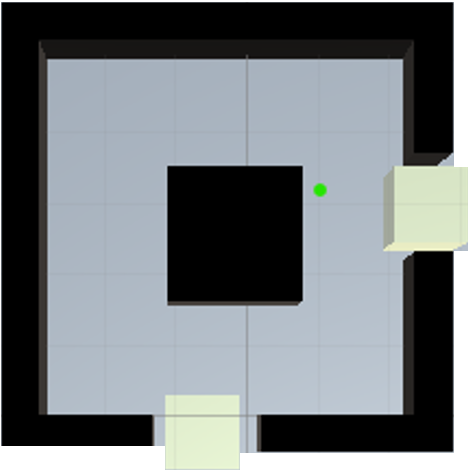}};

\node (room2) [
    right=\interRoomSpacing of room1.east,
] {\includegraphics[height=\roomHeight]{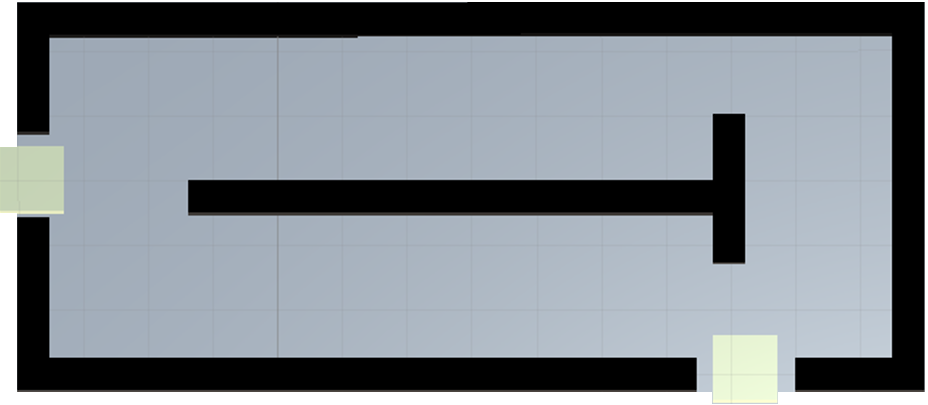}};

\node (room3) [
    right=\interRoomSpacing of room2.east
] {\includegraphics[height=\roomHeight]{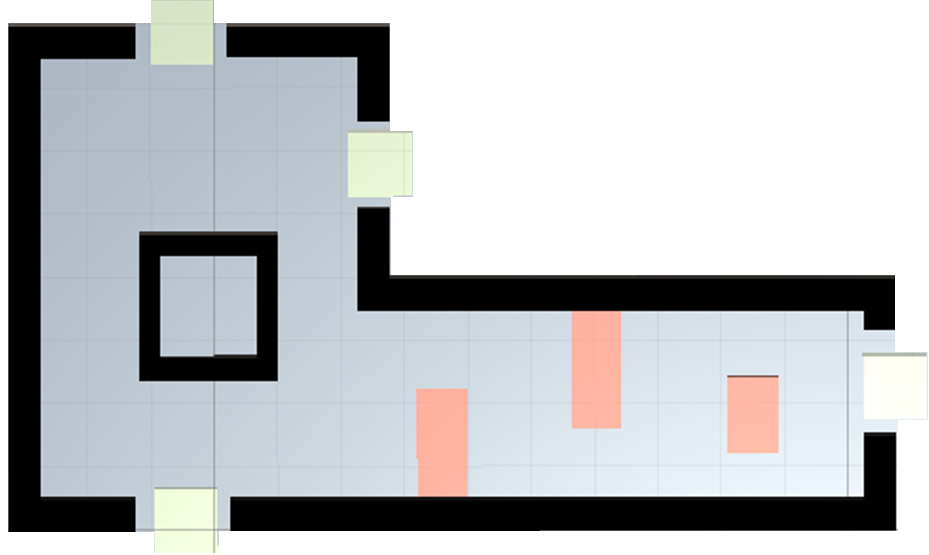}};

\node (room4) [
    right=\interRoomSpacing of room3.east,
] {\includegraphics[height=\roomHeight]{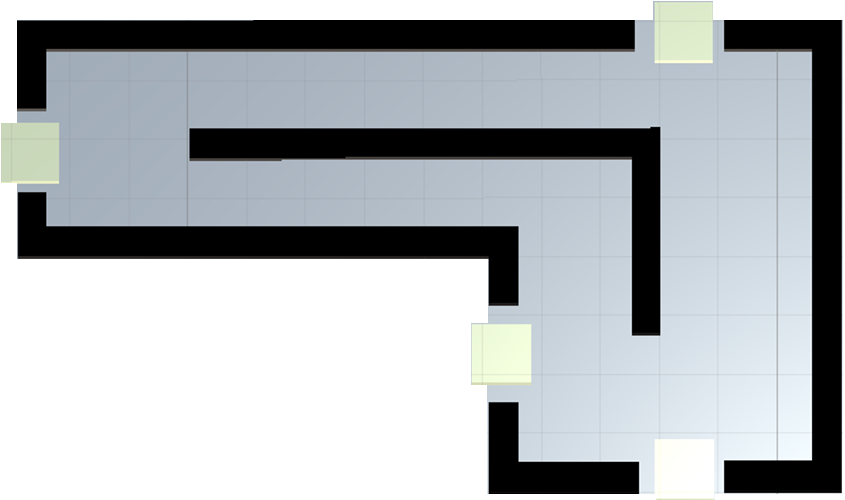}};

\node (room5) [
    right=\interRoomSpacing of room4.east,
] {\includegraphics[height=\roomHeight]{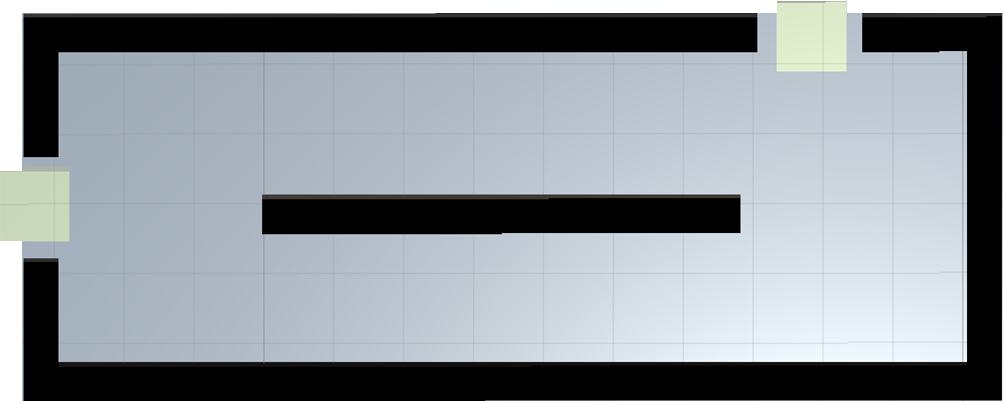}};

\node (room6) [
    right=\interRoomSpacing of room5.east,    
] {\includegraphics[height=\roomHeight]{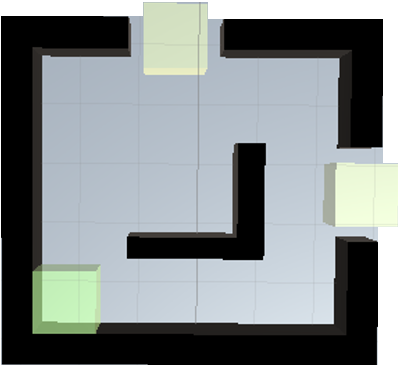}};

\node (componentsLabel) [
    labelText,
    above=0.05cm of room3.north,
    xshift=1.25cm,
    font=\small,
] {\textbf{Labyrinth Components}};

\node (componentEnclosure) [
    draw=black,
    dashed,
    inner xsep=3.5mm,
    inner ysep=1mm,
    rounded corners=0.15cm,
    fit={(room1)(room2)(room3)(room4)(room5)(room6)(componentsLabel)}
] {};

\end{tikzpicture}

%% file: tex/09_conclusions.tex
\section{Conclusions}

\label{sec:conclusions}

The verification of reinforcement learning (RL) systems is a critical step towards their widespread deployment in engineering applications.
We develop a framework for verifiable and compositional RL in which collections of RL subsystems are composed to achieve an overall task.
We automatically decompose system-level task specifications into individual subtask specifications, and iteratively refine these subtask specifications while training subsystems to satisfy them.
Future directions will study extensions of the framework to incorporate subsystems that are trained from human demonstrations, multi-level task hierarchies, and multi-agent RL systems.

%% file: appendices/thm1_proof.tex
\section{Proof of Theorem 1}
\label{sec:thm1_proof}

In words, we wish to show the following result: if each subsystem policy's probability of success is lower bounded by \(\bernoulliProb_{\controller}\), then under any given meta-policy, the probability of completing the task in the true environment is lower-bounded by the probability of reaching the goal state in the high-level model (HLM).

We begin by presenting definitions of the relevant probability values in the HLM \(\abstractMDP\) and in the environment POMDP \(\mdp\).
We then define what it means for a sequence of states and actions in \(\mdp\) to be \textit{consistent} with a sequence of high-level states and actions in \(\abstractMDP\).
Next, we present a lemma relating the probability of a sequence in \(\abstractMDP\) to the probability of its set of consistent sequences in \(\mdp\).
Finally, using the result of the lemma, we provide a proof for Theorem \ref{thm:hlm_bounds_true_performance}.

\paragraph{Defining the reachability probabilities in the HLM.}
Recall that the HLM \(\abstractMDP = (\abstractStateSet, \abstractInitialState, \abstractSuccessState, \abstractFailureState, \controllerSet, \abstractTransition)\) is a parametric Markov decision process with states \(\abstractStateSet\), action set \(\controllerSet\), and transition probability function \(\abstractTransition\) parametrized by \(\bernoulliProb_{\controller}\) for \(\controller \in \controllerSet\). 
Also, \(\hlmPolicy : \abstractStateSet \times \controllerSet \to [0,1]\) is a stationary policy on \(\abstractMDP\).


Let \(\mathbb{P}^{\abstractInitialState}_{\abstractMDP}(\abstractState_0 \controller_0 \abstractState_1 \controller_1 ... \abstractState_{\numMetaDecision} | \hlmPolicy)\) denote the probability of the \textit{finite path fragment} \(\abstractState_0 \controller_0 \abstractState_1 \controller_1 ... \abstractState_{\numMetaDecision}\) in the Markov chain induced by meta-policy \(\abstractPolicy\) acting in \(\abstractMDP\) beginning from the initial state \(\abstractInitialState \in \abstractStateSet\).
This probability value can be expressed in terms of \(\abstractMDP\)'s transition probability parameters \(\bernoulliProb_{\controller}\) as 
\[
\measure^{\abstractInitialState}_{\abstractMDP}(\abstractState_0 \controller_0 \abstractState_1 \controller_1 ... \abstractState_{\numMetaDecision} | \hlmPolicy) = \mathbbm{1}_{\abstractState_{0} = \abstractInitialState}\prod_{i=0}^{\numMetaDecision} \abstractPolicy(\abstractState_i, \controller_i) \bernoulliProb_{\controller_{i}}.
\]

We express the probability \(\measure^{\abstractInitialState}_{\abstractMDP}(\Diamond \abstractSuccessState | \hlmPolicy)\) of reaching the goal state \(\abstractSuccessState\in \abstractStateSet\) in \(\abstractMDP\) as
\[
\measure^{\abstractInitialState}_{\abstractMDP}(\Diamond \abstractSuccessState | \hlmPolicy)
= \measure_{\abstractMDP}^{\abstractInitialState}(\reachTrajectories_{\abstractMDP}^{\abstractInitialState, \abstractSuccessState} | \hlmPolicy), \label{eq:def_hlm_reach_prob}
\]
where \(\measure_{\abstractMDP}^{\abstractInitialState}(\reachTrajectories_{\abstractMDP}^{\abstractInitialState, \abstractSuccessState} | \hlmPolicy)\) is the sum of the probabilities of all of the finite path fragments that reach the goal state \(\abstractSuccessState \in \abstractStateSet\) from initial state \(\abstractInitialState \in \abstractStateSet\).
That is,
\[
\reachTrajectories_{\abstractMDP}^{\abstractState, \abstractSuccessState} \defeq \{ \abstractState_{0}\controller_{0} \ldots \abstractState_{m} | \abstractState_{0} = \abstractInitialState, \abstractState_{m} = \abstractSuccessState, \abstractState_{i} \neq \abstractSuccessState \quad \forall i < m \}.
\]

\paragraph{Relating HLM path fragments to environment path fragments.} 
Given a finite path fragment \(\abstractState_0 \controller_0 ... \abstractState_{\numMetaDecision} \) in the HLM \(\abstractMDP\) such that \(\abstractState_{0}, \ldots, \abstractState_{\numMetaDecision} \neq \abstractFailureState\), we define the set of  all \textit{consistent environment path fragments} \(\reachTrajectories_{\mdp}(\abstractState_0 \controller_0...\abstractState_{\numMetaDecision})\) to be a set of finite path fragments \(\mdpState_{0} \hat{\controller}_{0} \mdpAction_{0} \ldots \mdpState_{n}\) (with \(\numTimeStep > \numMetaDecision\)) of states \(\mdpState_{t} \in \mdpStateSet\), subsystems \(\hat{\controller}_t \in \controllerSet\), and actions \(\mdpAction_{t} \in \mdpActionSet\) within the environment POMDP \(\mdp\).
Note that we include the subsystems \(\hat{\controller}_{t} \in \controllerSet\) in these path fragments, in order to explicitly keep track of which subsystem is active at time \(t\) within the environment path fragment.
Intuitively, \(\reachTrajectories_{\mdp}(\abstractState_0 \controller_0...\abstractState_{\numMetaDecision})\) is the set of all finite path fragments in \(\mdp\), that are consistent with the given path fragment \(\abstractState_0 \controller_0...\abstractState_{\numMetaDecision}\) in \(\abstractMDP\).
More formally, we define \(\reachTrajectories_{\mdp}(\abstractState_0 \controller_0...\abstractState_{\numMetaDecision})\) as the set of all path fragments \(\mdpState_{0} \hat{\controller}_{0} \mdpAction_{0} \ldots \mdpState_{n}\), such that the following four conditions hold.
\begin{enumerate}
    \item There exists some collection of \textit{subsystem switching times} \(0 = \metaDecisionTime_{0} < \metaDecisionTime_{1} < \ldots < \metaDecisionTime_{\numMetaDecision} = \numTimeStep\) at which the meta-policy is used to select at new subsystem to execute.
    I.e., there are \(\numMetaDecision\) time points \(\metaDecisionTime_{i}\), indexed by \(i\), at which the active subsystem changes \(\hat{\controller}_{\metaDecisionTime_{i} - 1} \neq \hat{\controller}_{\metaDecisionTime_{i}}\).
    At all other times in the path fragment, the active subsystem remains constant  \(\hat{\controller}_{\metaDecisionTime_{i}} = \hat{\controller}_{\metaDecisionTime_{i} + 1} = \ldots = \hat{\controller}_{\metaDecisionTime_{i+1} - 1}\).
    
    \item The subsystems \(\hat{\controller}_{t} \in \controllerSet\) specified by the environment path fragment are consistent with the subsystems \(\controller_{i} \in \controllerSet\) specified by the HLM path fragment.
    That is, for every subsystem switching time \(\metaDecisionTime_{i}\), we have \(\hat{\controller}_{\metaDecisionTime_{i}} = \controller_{i}\).
    
    \item The environment states \(\mdpState_{t} \in \mdpStateSet\) specified by the environment path fragment are consistent with the HLM states \(\abstractState_{i} \in \abstractStateSet\) specified by the HLM path fragment.
    That is, for every subsystem switching time \(\metaDecisionTime_{i}\), we have \(\mdpState_{\metaDecisionTime_{i}} \in \abstractState_{i}\) and \(\mdpState_{\metaDecisionTime_i}, \mdpState_{\metaDecisionTime_i + 1},...,\mdpState_{\metaDecisionTime_{i+1} - 1} \notin \abstractState_{i+1}\) (recall that HLM states \(\abstractState \in \abstractStateSet\) correspond to collections of environment states \(\mdpState \in \mdpStateSet\)).
    
    
    \item For every pair of consecutive decision times, we have \(\metaDecisionTime_{i} - \metaDecisionTime_{i-1} \leq \timeHorizon_{\hat{\controller}_{\metaDecisionTime_{i}}}\).
    This corresonds to the assumption that the subsystem \(\hat{\controller}_{\metaDecisionTime_{i}}\) activated at time \(t=\metaDecisionTime_{i}\) completes its subtask within its allowed time horizon \(\timeHorizon_{\hat{\controller}_{\metaDecisionTime_{i}}}\).
    

\end{enumerate}

\paragraph{Defining the probability of environment path fragments.}

The probability of observing environment path fragment \(\mdpState_{0} \hat{\controller}_{0} \mdpAction_{0} \ldots \mdpState_{n} \in \reachTrajectories_{\mdp}(\abstractState_0 \controller_0...\abstractState_{\numMetaDecision})\) when executing meta-policy \(\abstractPolicy\) with subsystem policies \(\policy_{\controller_{1}}, \ldots, \policy_{\controller_{\numControllers}}\) in \(\mdp\), may be expressed as
\begin{align*}
\mathbb{P}^{\mdpInitialState}_{\mdp}&(\mdpState_{0}\hat{\controller}_{0}\mdpAction_{0} \ldots \mdpState_{\numTimeStep} | \abstractPolicy, \policy_{\controller_{1}}, \ldots, \policy_{\controller_{\numControllers}}) = \\
& \mathbbm{1}_{\mdpState_{0} = \mdpInitialState} \prod_{t=0}^{\numTimeStep -1}
\begin{cases}
\abstractPolicy([\mdpState_{t}]_{\eqRelation}, \hat{\controller}_{t}) \policy_{\hat{\controller}_{t}}(\history_{\metaDecisionTime_{i}:t}, \mdpAction_{t}) \mdpTransition(\mdpState_{t}, \mdpAction_{t}, \mdpState_{t+1}), \quad \text{if } t = \metaDecisionTime_{i} \text{ for some } i \in \{1,\ldots, \numMetaDecision\} \\ 
\policy_{\hat{\controller}_{i}}(\history_{\metaDecisionTime_{i}:t}, \mdpAction_{t}) \mdpTransition(\mdpState_{t}, \mdpAction_{t}, \mdpState_{t+1}), \hfill \text{otherwise}.
\end{cases}
\end{align*}
Here, \(\history_{t_{1}:t_{2}} = \pomdpObservation_{t_{1}} \mdpAction_{t_{1}} \pomdpObservation_{t_{1}+1} \mdpAction_{t_{1} + 1} \ldots \pomdpObservation_{t_{2}}\) denotes the history of observations and actions from time \(t_{1}\) up to time \(t_{2}\), where \(\pomdpObservation_{t} \sim \pomdpObservationFunction(\mdpState_{t})\).
Also, recall that \([\mdpState_{t}]_{\eqRelation}\) denotes the equivalence relation on \(\mdpStateSet\) used to define the HLM states \(\abstractStateSet\), i.e. \([\mdpState_{t}]_{\eqRelation} = \abstractState\) such that \(\mdpState_{t} \in \abstractState\).

We now present Lemma \ref{thm:lemma1}, which relates the probability of an HLM path fragment, to the probability of the set of all the consistent environment path fragments.

\begin{lemma}
\label{thm:lemma1}
Let \(\abstractState_0 \controller_0 ...\abstractState_{\numMetaDecision} \in Paths_{fin}(\abstractMDP, \hlmPolicy, \abstractInitialState)\) be a finite path fragment in \(\abstractMDP\) such that \(\abstractState_0, \ldots, \abstractState_{\numMetaDecision} \neq \abstractFailureState\). If, for every subsystem \(\controller \in \controllerSet\) and for every entry condition \(\mdpState \in \controllerInitialStateSet_{\controller}\) we have \(\mathbb{P}_{\mdp}^{\mdpState}(\Diamond_{\leq \timeHorizon_{\controller}}\controllerFinalStateSet_{\controller} | \policy_{\controller}) \geq \bernoulliProb_{\controller}\), then the following inequality holds.

\begin{equation}
    \measure_{\mdp}^{\mdpInitialState}(\reachTrajectories_{\mdp}(\abstractState_0 \controller_0 ... \abstractState_{\numMetaDecision}) |  \abstractPolicy, \policy_{\controller_{1}}, \ldots, \policy_{\controller_{\numControllers}}) \geq \measure_{\abstractMDP}^{\abstractInitialState}(\abstractState_0 \controller_0 ... \abstractState_{\numMetaDecision} | \hlmPolicy)
\end{equation}
\end{lemma}
\begin{proof}[Proof of Lemma \ref{thm:lemma1}]
This inequality follows by induction from the Lemma's assumption.

Given the finite path fragment \(\abstractState_0 \controller_0 ... \abstractState_{\numMetaDecision}\) we begin by considering the trivial prefix \(\abstractState_{0}\).
We have that
\[\measure_{\abstractMDP}^{\abstractInitialState}(\abstractState_0 | \hlmPolicy) = \measure_{\mdp}^{\mdpInitialState}(\reachTrajectories_{\mdp}(\abstractState_0) | \abstractPolicy, \policy_{\controller_{1}}, \ldots, \policy_{\controller_{\numControllers}}) = 1.\]

Now, for any \(0 \leq l \leq \numMetaDecision - 1\) we consider the prefix \(\abstractState_0 \controller_0 ... \abstractState_{l}\) of the path fragment \(\abstractState_0 \controller_0 ... \abstractState_l \controller_l ... \abstractState_{\numMetaDecision}\).
Suppose that the following inequality holds.
\[\measure_{\mdp}^{\mdpInitialState}(\reachTrajectories_{\mdp}(\abstractState_0 \controller_0 ... \abstractState_l) | \abstractPolicy, \policy_{\controller_{1}}, \ldots, \policy_{\controller_{\numControllers}}) \geq \measure_{\abstractMDP}^{\abstractInitialState}(\abstractState_0 \controller_0 ... \abstractState_l | \hlmPolicy).
\]
We may write the probability \(\measure_{\mdp}^{\mdpInitialState}(\reachTrajectories_{\mdp}(\abstractState_0 \controller_0 ... \abstractState_l \controller_l \abstractState_{l+1}) | \abstractPolicy, \policy_{\controller_{1}}, \ldots, \policy_{\controller_{\numControllers}})\), corresponding to the prefix of length \(l+1\), in terms of the probability of the prefix of length \(l\) and the probability of all environment path fragments consistent with the HLM transition from \(\abstractState_l\) to \(\abstractState_{l+1}\).
\begin{align*}
\measure_{\mdp}^{\mdpInitialState}(\reachTrajectories_{\mdp}&(\abstractState_0 \controller_0 ... \abstractState_l \controller_l \abstractState_{l+1}) | \abstractPolicy, \policy_{\controller_{1}}, \ldots, \policy_{\controller_{\numControllers}}) = \\ 
& \measure_{\mdp}^{\mdpInitialState}(\reachTrajectories_{\mdp}(\abstractState_0 \controller_0 ... \abstractState_l) | \abstractPolicy, \policy_{\controller_{1}}, \ldots, \policy_{\controller_{\numControllers}}) \sum_{\mdpState \in \abstractState_l} \alpha(\mdpState) * \abstractPolicy(\abstractState_{l}, \controller_l) * \mathbb{P}_{\mdp}^{\mdpState}(\Diamond_{\leq \timeHorizon_{\controller_l}} \abstractState_{l+1} | \policy_{\controller_l}).
\end{align*}
Here, \(\alpha(\mdpState)\) is \textit{some} distribution over the states \(\mdpState \in \abstractState_{l}\) such that \(\sum_{\mdpState \in \abstractState_l} \alpha(\mdpState) = 1\).

Note that because \(\abstractState_{l+1} \neq \abstractFailureState \) by assumption, it must be the case that \(\abstractState_{l+1} = succ(\controller_l)\). This implies, by definition, that \(\controllerFinalStateSet_{\controller_l} \subseteq \abstractState_{l+1}\). Thus, 
\[
\mathbb{P}_{\mdp}^{\mdpState}(\Diamond_{\leq \timeHorizon_{\controller_l}}\abstractState_{l+1} | \policy_{\controller_l}) \geq \mathbb{P}_{\mdp}^{\mdpState}(\Diamond_{\leq \timeHorizon_{\controller_l}} \controllerFinalStateSet_{\controller_l} | \policy_{\controller_l}) \geq \bernoulliProb_{\controller_l},
\] 
where the final inequality follows from the lemma's assumption and from the fact that \(\mdpState \in \controllerInitialStateSet_{\controller_l}\) for every \(\mdpState \in \abstractState_l\).
Putting all this together, we have 
\[\measure_{\mdp}^{\mdpInitialState}(\reachTrajectories_{\mdp}(\abstractState_0 \controller_0 ... \abstractState_l \controller_l \abstractState_{l+1}) | \abstractPolicy, \policy_{\controller_{1}}, \ldots, \policy_{\controller_{\numControllers}}) \geq \measure_{\mdp}^{\mdpInitialState}(\reachTrajectories_{\mdp}(\abstractState_0 \controller_0 ... \abstractState_l) | \abstractPolicy, \policy_{\controller_{1}}, \ldots, \policy_{\controller_{\numControllers}}) * \hlmPolicy(\abstractState_l, \controller_l) * \bernoulliProb_{\controller_l}
\]
and given the assumption of our induction step that \(\measure_{\mdp}^{\mdpInitialState}(\reachTrajectories_{\mdp}(\abstractState_0 \controller_0 ... \abstractState_l) | \abstractPolicy, \policy_{\controller_{1}}, \ldots, \policy_{\controller_{\numControllers}}) \geq \measure_{\abstractMDP}^{\abstractInitialState}(\abstractState_0 \controller_0 ... \abstractState_l | \hlmPolicy)\), we obtain the inequality 
\begin{align*}
\measure_{\mdp}^{\mdpInitialState}(\reachTrajectories_{\mdp}(\abstractState_0 \controller_0 ... \abstractState_l \controller_l \abstractState_{l+1}) | \abstractPolicy, \policy_{\controller_{1}}, \ldots, \policy_{\controller_{\numControllers}}) &\geq \measure_{\abstractMDP}^{\abstractInitialState}(\abstractState_0 \controller_0 ... \abstractState_l | \hlmPolicy) * \hlmPolicy(\abstractState_l, \controller_l) * \bernoulliProb_{\controller_l} \\
&= \measure_{\abstractMDP}^{\abstractInitialState}(\abstractState_0 \controller_0 ... \abstractState_l \controller_l \abstractState_{l+1} | \hlmPolicy).
\end{align*}
By induction, we conclude the proof.

\end{proof}

With this lemma in place, we now proceed to the proof of Theorem \ref{thm:hlm_bounds_true_performance}.

\begin{manualtheorem}{1}
Let \(\controllerSet = \{\controller_1, \controller_2, ..., \controller_{\numControllers}\}\) be a collection of composable subsystems with respect to initial state \(\mdpInitialState\) and target set \(\controllerFinalStateSet_{targ}\) within the environment POMDP \(\mdp\). Define \(\abstractMDP\) to be the corresponding HLM and let \(\hlmPolicy\) be a policy in \(\abstractMDP\). 
If, for every subsystem \(\controller_{i} \in \controllerSet\) and for every entry condition \(\mdpState \in \controllerInitialStateSet_{\controller_{i}}\), 
\(\mathbb{P}^{\mdpState}_{\mdp}(\Diamond_{\leq \timeHorizon_{\controller_{i}}}\controllerFinalStateSet_{\controller_{i}} | \policy_{\controller_{i}}) \geq \bernoulliProb_{\controller_{i}}\)
, then 
\[\mathbb{P}^{\mdpInitialState}_{\mdp}(\Diamond \controllerFinalStateSet_{targ} | \abstractPolicy, \policy_{\controller_{1}}, \ldots, \policy_{\controller_{\numControllers}}) \geq \mathbb{P}^{\abstractInitialState}_{\abstractMDP}(\Diamond \abstractSuccessState | \hlmPolicy).\]
\end{manualtheorem}%

\begin{proof}[Proof of Theorem \ref{thm:hlm_bounds_true_performance}]
Consider any finite path fragment \(\abstractState_0 \controller_0 ...\abstractState_{\numMetaDecision} \in \reachTrajectories_{\abstractMDP}^{\abstractInitialState, \abstractSuccessState}\) that reaches the goal state \(\abstractSuccessState\) in the HLM \(\abstractMDP\).
By Lemma 1, we know that \(\measure_{\mdp}^{\mdpInitialState}(\reachTrajectories_{\mdp}(\abstractState_0 \controller_0 ... \abstractState_{\numMetaDecision}) | \abstractPolicy, \policy_{\controller_{1}}, \ldots, \policy_{\controller_{\numControllers}}) \geq \measure_{\abstractMDP}^{\abstractInitialState}(\abstractState_0 \controller_0 ... \abstractState_{\numMetaDecision} | \hlmPolicy)\), which implies that
\begin{align*}
    \sum_{\abstractState_0 \controller_0 ... \abstractState_{\numMetaDecision} \in \reachTrajectories_{\abstractMDP}^{\abstractInitialState, \abstractSuccessState}} \measure_{\mdp}^{\mdpInitialState}(\reachTrajectories_{\mdp}(\abstractState_0 \controller_0 ... \abstractState_{\numMetaDecision}) | \abstractPolicy, \policy_{\controller_{1}}, \ldots, \policy_{\controller_{\numControllers}}) &\geq \sum_{\abstractState_0 \controller_0 ... \abstractState_{\numMetaDecision} \in \reachTrajectories_{\abstractMDP}^{\abstractInitialState, \abstractSuccessState}} \measure_{\abstractMDP}^{\abstractInitialState}(\abstractState_0 \controller_0 ... \abstractState_{\numMetaDecision} | \hlmPolicy)\\
    &= \mathbb{P}_{\abstractMDP}^{\abstractInitialState}(\Diamond \abstractSuccessState | \hlmPolicy).
\end{align*}
Let \(\reachTrajectories_{\mdp}^{\mdpInitialState, \controllerFinalStateSet_{targ}}\) denote the set of all environment path fragments that reach the target set \(\controllerFinalStateSet_{targ}\) from the initial state \(\mdpInitialState\), i.e. \(\reachTrajectories_{\mdp}^{\mdpInitialState, \controllerFinalStateSet_{targ}} \defeq \{\mdpState_{0} \hat{\controller}_{0} \mdpAction_{0} \ldots \mdpState_{\numTimeStep} | \mdpState_{0} = \mdpInitialState, \mdpState_{\numTimeStep} \in \controllerFinalStateSet_{targ}\}\).
We note that for any HLM path fragment \(\abstractState_0 \controller_0 ...\abstractState_{\numMetaDecision} \in \reachTrajectories_{\abstractMDP}^{\abstractInitialState, \abstractSuccessState}\), it follows from the definition of the HLM goal state \(\abstractState_{\numMetaDecision} = \abstractSuccessState\) that \(\reachTrajectories_{\mdp}(\abstractState_0 \controller_0 ... \abstractState_{\numMetaDecision}) \subseteq \reachTrajectories_{\mdp}^{\mdpInitialState, \controllerFinalStateSet_{targ}}\).
Furthermore, the sets of consistent environment path fragments \(\reachTrajectories_{\mdp}(\abstractState_0 \controller_0 ... \abstractState_{\numMetaDecision})\) for every HLM path fragment \(\abstractState_0 \controller_0 ...\abstractState_{\numMetaDecision} \in \reachTrajectories_{\abstractMDP}^{\abstractInitialState, \abstractSuccessState}\) are pairwise disjoint. So, we conclude that 

\begin{align*}
    \measure_{\mdp}^{\mdpInitialState}(\Diamond \controllerFinalStateSet_{targ} | \abstractPolicy, \policy_{\controller_{1}}, \ldots, \policy_{\controller_{\numControllers}}) &= \measure_{\mdp}^{\mdpInitialState}(\reachTrajectories_{\mdp}^{\mdpInitialState, \controllerFinalStateSet_{targ}} | \abstractPolicy, \policy_{\controller_{1}}, \ldots, \policy_{\controller_{\numControllers}}) \\
    &\geq \sum_{\abstractState_0 \controller_0 ... \abstractState_{\numMetaDecision} \in \reachTrajectories_{\abstractMDP}^{\abstractInitialState, \abstractSuccessState}} \measure_{\mdp, \abstractPolicy}^{\mdpInitialState}(\reachTrajectories_{\mdp}(\abstractState_0 \controller_0 ... \abstractState_{\numMetaDecision}) | \abstractPolicy, \policy_{\controller_{1}}, \ldots, \policy_{\controller_{\numControllers}}) \\
    &\geq \mathbb{P}_{\abstractMDP, \hlmPolicy}^{\abstractInitialState}(\Diamond \abstractSuccessState | \hlmPolicy).
\end{align*}

\end{proof}


%% file: appendices/ppo_algorithm_parameters.tex
\section{Reinforcement Learning Algorithm Parameters}
\label{sec:appendix_ppo_parameters}

In all of our numerical case studies we used the proximal policy optimization (PPO) algorithm with the same collection of algorithm parameters to train the subsystem policies.
The value of the discount factor is \(\mdpDiscount = 0.99\). 
The value of the generalized advantage estimation parameter is \(\lambda = 0.95\).
The value of the clipping parameter is \(\varepsilon =0.2\).
The value of the entropy coefficient is \(0.0\).
The the value function coefficient is \(0.5\).
The value of the learning rate is \(2.5\text{e}-4\).
We use a batch size of \(62\) and \(10\) epochs per update of the PPO algorithm.

%% file: main.bbl
\begin{thebibliography}{}

\bibitem[\protect\BCAY{Baier\ \BBA\ Katoen}{Baier\ \BBA\ Katoen}{2008}]{baier2008principles}
Baier, C.\BBACOMMA\  \BBA\ Katoen, J.-P. \BBOP2008\BBCP.
\newblock {\Bem Principles of model checking}.
\newblock MIT press.

\bibitem[\protect\BCAY{Barto\ \BBA\ Mahadevan}{Barto\ \BBA\ Mahadevan}{2003}]{barto2003recent}
Barto, A.~G.\BBACOMMA\  \BBA\ Mahadevan, S. \BBOP2003\BBCP.
\newblock \BBOQ Recent advances in hierarchical reinforcement learning\BBCQ\
\newblock {\Bem Discrete event dynamic systems}, {\Bem 13\/}(1), 41--77.

\bibitem[\protect\BCAY{Camacho, Chen, Sanner,\ \BBA\ McIlraith}{Camacho et~al.}{2017}]{camacho2017non}
Camacho, A., Chen, O., Sanner, S., \BBA\ McIlraith, S.~A. \BBOP2017\BBCP.
\newblock \BBOQ Non-markovian rewards expressed in ltl: guiding search via reward shaping\BBCQ\
\newblock \textit{10th Annual Symposium on Combinatorial Search}.

\bibitem[\protect\BCAY{Chevalier-Boisvert, Willems,\ \BBA\ Pal}{Chevalier-Boisvert et~al.}{2018}]{gym_minigrid}
Chevalier-Boisvert, M., Willems, L., \BBA\ Pal, S. \BBOP2018\BBCP.
\newblock \BBOQ Minimalistic gridworld environment for openai gym\BBCQ\
\newblock \url{https://github.com/maximecb/gym-minigrid}.

\bibitem[\protect\BCAY{Cubuktepe, Jansen, Junges, Katoen,\ \BBA\ Topcu}{Cubuktepe et~al.}{2018}]{cubuktepe2018synthesis}
Cubuktepe, M., Jansen, N., Junges, S., Katoen, J.-P., \BBA\ Topcu, U. \BBOP2018\BBCP.
\newblock \BBOQ {Synthesis in pMDPs: A Tale of 1001 Parameters}\BBCQ\
\newblock In {\Bem International Symposium on Automated Technology for Verification and Analysis}, \BPGS\ 160--176. Springer.

\bibitem[\protect\BCAY{Etessami, Kwiatkowska, Vardi,\ \BBA\ Yannakakis}{Etessami et~al.}{2007}]{etessami2007multi}
Etessami, K., Kwiatkowska, M., Vardi, M.~Y., \BBA\ Yannakakis, M. \BBOP2007\BBCP.
\newblock \BBOQ Multi-objective model checking of markov decision processes\BBCQ\
\newblock In {\Bem International Conference on Tools and Algorithms for the Construction and Analysis of Systems}, \BPGS\ 50--65. Springer.

\bibitem[\protect\BCAY{Feng, Kwiatkowska,\ \BBA\ Parker}{Feng et~al.}{2011}]{feng2011automated}
Feng, L., Kwiatkowska, M., \BBA\ Parker, D. \BBOP2011\BBCP.
\newblock \BBOQ Automated learning of probabilistic assumptions for compositional reasoning\BBCQ\
\newblock In {\Bem International Conference on Fundamental Approaches to Software Engineering}, \BPGS\ 2--17. Springer.

\bibitem[\protect\BCAY{Furelos-Blanco, Law, Jonsson, Broda,\ \BBA\ Russo}{Furelos-Blanco et~al.}{2021}]{furelos2021induction}
Furelos-Blanco, D., Law, M., Jonsson, A., Broda, K., \BBA\ Russo, A. \BBOP2021\BBCP.
\newblock \BBOQ Induction and exploitation of subgoal automata for reinforcement learning\BBCQ\
\newblock {\Bem Journal of Artificial Intelligence Research}, {\Bem 70}, 1031--1116.

\bibitem[\protect\BCAY{{Gurobi Optimization, LLC}}{{Gurobi Optimization, LLC}}{2021}]{gurobi}
{Gurobi Optimization, LLC} \BBOP2021\BBCP.
\newblock \BBOQ Gurobi optimizer reference manual\BBCQ\
\newblock \url{https://www.gurobi.com/documentation/9.1/refman/index.html}.
\newblock Accessed: 2021-12-15.

\bibitem[\protect\BCAY{Haberfellner, Nagel, Becker, B{\"u}chel,\ \BBA\ von Massow}{Haberfellner et~al.}{2019}]{haberfellner2019systems}
Haberfellner, R., Nagel, P., Becker, M., B{\"u}chel, A., \BBA\ von Massow, H. \BBOP2019\BBCP.
\newblock {\Bem Systems engineering}.
\newblock Springer.

\bibitem[\protect\BCAY{Hahn, Perez, Schewe, Somenzi, Trivedi,\ \BBA\ Wojtczak}{Hahn et~al.}{2019}]{hahn2019omega}
Hahn, E.~M., Perez, M., Schewe, S., Somenzi, F., Trivedi, A., \BBA\ Wojtczak, D. \BBOP2019\BBCP.
\newblock \BBOQ Omega-regular objectives in model-free reinforcement learning\BBCQ\
\newblock In {\Bem International Conference on Tools and Algorithms for the Construction and Analysis of Systems}, \BPGS\ 395--412. Springer.

\bibitem[\protect\BCAY{Juliani, Berges, Teng, Cohen, Harper, Elion, Goy, Gao, Henry, Mattar, et~al.}{Juliani et~al.}{2018}]{juliani2018unity}
Juliani, A., Berges, V.-P., Teng, E., Cohen, A., Harper, J., Elion, C., Goy, C., Gao, Y., Henry, H., Mattar, M., et~al. \BBOP2018\BBCP.
\newblock \BBOQ Unity: A general platform for intelligent agents\BBCQ\
\newblock \textit{arXiv preprint arXiv:1809.02627}.

\bibitem[\protect\BCAY{Junges, Abraham, Hensel, Jansen, Katoen, Quatmann,\ \BBA\ Volk}{Junges et~al.}{2019}]{junges2020parameter}
Junges, S., Abraham, E., Hensel, C., Jansen, N., Katoen, J.-P., Quatmann, T., \BBA\ Volk, M. \BBOP2019\BBCP.
\newblock \BBOQ Parameter synthesis for markov models\BBCQ.

\bibitem[\protect\BCAY{Kulkarni, Narasimhan, Saeedi,\ \BBA\ Tenenbaum}{Kulkarni et~al.}{2016}]{kulkarni2016hierarchical}
Kulkarni, T.~D., Narasimhan, K., Saeedi, A., \BBA\ Tenenbaum, J. \BBOP2016\BBCP.
\newblock \BBOQ Hierarchical deep reinforcement learning: Integrating temporal abstraction and intrinsic motivation\BBCQ\
\newblock In {\Bem Advances in Neural Information Processing Systems}, \lowercase{\BVOL}~29.

\bibitem[\protect\BCAY{Kwiatkowska, Norman, Parker,\ \BBA\ Qu}{Kwiatkowska et~al.}{2010}]{kwiatkowska2010assume}
Kwiatkowska, M., Norman, G., Parker, D., \BBA\ Qu, H. \BBOP2010\BBCP.
\newblock \BBOQ Assume-guarantee verification for probabilistic systems\BBCQ\
\newblock In {\Bem International Conference on Tools and Algorithms for the Construction and Analysis of Systems}, \BPGS\ 23--37. Springer.

\bibitem[\protect\BCAY{Levy, Konidaris, Platt,\ \BBA\ Saenko}{Levy et~al.}{2019}]{levy2017learning}
Levy, A., Konidaris, G.~D., Platt, R.~W., \BBA\ Saenko, K. \BBOP2019\BBCP.
\newblock \BBOQ Learning multi-level hierarchies with hindsight\BBCQ\
\newblock In {\Bem International Conference on Learning Representations}.

\bibitem[\protect\BCAY{Littman, Topcu, Fu, Isbell, Wen,\ \BBA\ MacGlashan}{Littman et~al.}{2017}]{littman2017environment}
Littman, M.~L., Topcu, U., Fu, J., Isbell, C., Wen, M., \BBA\ MacGlashan, J. \BBOP2017\BBCP.
\newblock \BBOQ Environment-independent task specifications via gltl\BBCQ\
\newblock \textit{arXiv preprint arXiv:1704.04341}.

\bibitem[\protect\BCAY{Mnih, Kavukcuoglu, Silver, Rusu, Veness, Bellemare, Graves, Riedmiller, Fidjeland, Ostrovski, et~al.}{Mnih et~al.}{2015}]{mnih2015human}
Mnih, V., Kavukcuoglu, K., Silver, D., Rusu, A.~A., Veness, J., Bellemare, M.~G., Graves, A., Riedmiller, M., Fidjeland, A.~K., Ostrovski, G., et~al. \BBOP2015\BBCP.
\newblock \BBOQ Human-level control through deep reinforcement learning\BBCQ\
\newblock {\Bem Nature}, {\Bem 518\/}(7540), 529--533.

\bibitem[\protect\BCAY{Nachum, Gu, Lee,\ \BBA\ Levine}{Nachum et~al.}{2018}]{nachum2019data}
Nachum, O., Gu, S.~S., Lee, H., \BBA\ Levine, S. \BBOP2018\BBCP.
\newblock \BBOQ Data-efficient hierarchical reinforcement learning\BBCQ\
\newblock In {\Bem Advances in Neural Information Processing Systems}, \lowercase{\BVOL}~31.

\bibitem[\protect\BCAY{Nam, Madhusudan,\ \BBA\ Alur}{Nam et~al.}{2008}]{namautomatic2008}
Nam, W., Madhusudan, P., \BBA\ Alur, R. \BBOP2008\BBCP.
\newblock \BBOQ Automatic symbolic compositional verification by learning assumptions\BBCQ\
\newblock {\Bem Formal Methods in System Design}, {\Bem 32\/}(3), 207–234.

\bibitem[\protect\BCAY{Neary, Verginis, Cubuktepe,\ \BBA\ Topcu}{Neary et~al.}{2022}]{neary2022verifiable}
Neary, C., Verginis, C., Cubuktepe, M., \BBA\ Topcu, U. \BBOP2022\BBCP.
\newblock \BBOQ Verifiable and compositional reinforcement learning systems\BBCQ\
\newblock In {\Bem Proceedings of the International Conference on Automated Planning and Scheduling}, \lowercase{\BVOL}~32, \BPGS\ 615--623.

\bibitem[\protect\BCAY{Neary, Xu, Wu,\ \BBA\ Topcu}{Neary et~al.}{2021}]{neary2020reward}
Neary, C., Xu, Z., Wu, B., \BBA\ Topcu, U. \BBOP2021\BBCP.
\newblock \BBOQ Reward machines for cooperative multi-agent reinforcement learning\BBCQ\
\newblock In {\Bem Proceedings of the 20th International Conference on Autonomous Agents and MultiAgent Systems}, AAMAS '21, \BPG\ 934–942.

\bibitem[\protect\BCAY{Nuseibeh\ \BBA\ Easterbrook}{Nuseibeh\ \BBA\ Easterbrook}{2000}]{nuseibeh2000requirements}
Nuseibeh, B.\BBACOMMA\  \BBA\ Easterbrook, S. \BBOP2000\BBCP.
\newblock \BBOQ Requirements engineering: a roadmap\BBCQ\
\newblock In {\Bem Proceedings of the Conference on the Future of Software Engineering}, \BPGS\ 35--46.

\bibitem[\protect\BCAY{Pateria, Subagdja, Tan,\ \BBA\ Quek}{Pateria et~al.}{2021}]{pateria2021hierarchical}
Pateria, S., Subagdja, B., Tan, A.-h., \BBA\ Quek, C. \BBOP2021\BBCP.
\newblock \BBOQ Hierarchical reinforcement learning: A comprehensive survey\BBCQ\
\newblock {\Bem ACM Computing Surveys (CSUR)}, {\Bem 54\/}(5), 1--35.

\bibitem[\protect\BCAY{Puterman}{Puterman}{2014}]{puterman2014markov}
Puterman, M.~L. \BBOP2014\BBCP.
\newblock {\Bem Markov decision processes: discrete stochastic dynamic programming}.
\newblock John Wiley \& Sons.

\bibitem[\protect\BCAY{Raffin, Hill, Gleave, Kanervisto, Ernestus,\ \BBA\ Dormann}{Raffin et~al.}{2021}]{stable-baselines3}
Raffin, A., Hill, A., Gleave, A., Kanervisto, A., Ernestus, M., \BBA\ Dormann, N. \BBOP2021\BBCP.
\newblock \BBOQ Stable-baselines3: Reliable reinforcement learning implementations\BBCQ\
\newblock {\Bem Journal of Machine Learning Research}, {\Bem 22\/}(268), 1--8.

\bibitem[\protect\BCAY{Sarathy, Kasenberg, Goel, Sinapov,\ \BBA\ Scheutz}{Sarathy et~al.}{2021}]{sarathy2020spotter}
Sarathy, V., Kasenberg, D., Goel, S., Sinapov, J., \BBA\ Scheutz, M. \BBOP2021\BBCP.
\newblock \BBOQ Spotter: Extending symbolic planning operators through targeted reinforcement learning\BBCQ\
\newblock In {\Bem Proceedings of the 20th International Conference on Autonomous Agents and MultiAgent Systems}, AAMAS '21, \BPG\ 1118–1126.

\bibitem[\protect\BCAY{Schulman, Wolski, Dhariwal, Radford,\ \BBA\ Klimov}{Schulman et~al.}{2017}]{schulman2017proximal}
Schulman, J., Wolski, F., Dhariwal, P., Radford, A., \BBA\ Klimov, O. \BBOP2017\BBCP.
\newblock \BBOQ Proximal policy optimization algorithms\BBCQ\
\newblock \textit{arXiv preprint arXiv:1707.06347}.

\bibitem[\protect\BCAY{Sutton, Precup,\ \BBA\ Singh}{Sutton et~al.}{1999}]{sutton1999between}
Sutton, R.~S., Precup, D., \BBA\ Singh, S. \BBOP1999\BBCP.
\newblock \BBOQ Between mdps and semi-mdps: A framework for temporal abstraction in reinforcement learning\BBCQ\
\newblock {\Bem Artificial Intelligence}, {\Bem 112\/}(1-2), 181--211.

\bibitem[\protect\BCAY{Toro~Icarte, Klassen, Valenzano,\ \BBA\ McIlraith}{Toro~Icarte et~al.}{2018}]{icarte2018using}
Toro~Icarte, R., Klassen, T., Valenzano, R., \BBA\ McIlraith, S. \BBOP2018\BBCP.
\newblock \BBOQ Using reward machines for high-level task specification and decomposition in reinforcement learning\BBCQ\
\newblock In {\Bem International Conference on Machine Learning}, \BPGS\ 2107--2116. PMLR.

\bibitem[\protect\BCAY{Toro~Icarte, Klassen, Valenzano,\ \BBA\ McIlraith}{Toro~Icarte et~al.}{2022}]{icarte2022reward}
Toro~Icarte, R., Klassen, T.~Q., Valenzano, R., \BBA\ McIlraith, S.~A. \BBOP2022\BBCP.
\newblock \BBOQ Reward machines: Exploiting reward function structure in reinforcement learning\BBCQ\
\newblock {\Bem Journal of Artificial Intelligence Research}, {\Bem 73}, 173--208.

\bibitem[\protect\BCAY{Toro~Icarte, Waldie, Klassen, Valenzano, Castro,\ \BBA\ McIlraith}{Toro~Icarte et~al.}{2019}]{toro2019learning}
Toro~Icarte, R., Waldie, E., Klassen, T., Valenzano, R., Castro, M., \BBA\ McIlraith, S. \BBOP2019\BBCP.
\newblock \BBOQ Learning reward machines for partially observable reinforcement learning\BBCQ\
\newblock In {\Bem Advances in Neural Information Processing Systems}, \lowercase{\BVOL}~32.

\bibitem[\protect\BCAY{Vezhnevets, Osindero, Schaul, Heess, Jaderberg, Silver,\ \BBA\ Kavukcuoglu}{Vezhnevets et~al.}{2017}]{vezhnevets2017feudal}
Vezhnevets, A.~S., Osindero, S., Schaul, T., Heess, N., Jaderberg, M., Silver, D., \BBA\ Kavukcuoglu, K. \BBOP2017\BBCP.
\newblock \BBOQ Feudal networks for hierarchical reinforcement learning\BBCQ\
\newblock In {\Bem International Conference on Machine Learning}, \BPGS\ 3540--3549. PMLR.

\bibitem[\protect\BCAY{Xu, Gavran, Ahmad, Majumdar, Neider, Topcu,\ \BBA\ Wu}{Xu et~al.}{2020}]{xu2020joint}
Xu, Z., Gavran, I., Ahmad, Y., Majumdar, R., Neider, D., Topcu, U., \BBA\ Wu, B. \BBOP2020\BBCP.
\newblock \BBOQ Joint inference of reward machines and policies for reinforcement learning\BBCQ\
\newblock In {\Bem Proceedings of the International Conference on Automated Planning and Scheduling}, \lowercase{\BVOL}~30, \BPGS\ 590--598.

\end{thebibliography}
